\documentclass[12pt,reqno]{amsart}
\usepackage{graphicx,amsmath,amsfonts,amssymb}
\usepackage{psfrag}
\usepackage{epsf}
\usepackage{latexsym}
\usepackage{rotating}

\newtheorem{theorem}{Theorem}

\newtheorem{lemma}{Lemma}

\newtheorem{remark}{Remark}

\newcommand{\bsup}[2]{ \raisebox{#1 ex}{${\mbox{Sup }}\atop{#2}$}\ }
\newcommand{\ve}{\vspace{1ex}}
\newcommand{\vsp}[1]{\vspace{#1 ex}}

\begin{document}

%%%%%%%%%%%%%%%%%
%preprint style
%%%%%%%%%%%%%%%%%
\title[Nonlinear waves in Newton's cradle]{Nonlinear waves in Newton's cradle and the discrete $p$-Schr\"odinger equation}
\author{Guillaume James}
\address{Laboratoire Jean Kuntzmann, Universit\'e de Grenoble and CNRS,
BP 53 \\ 
38041 Grenoble Cedex 9, France.}
\email{Guillaume.James@imag.fr}
\date{\today}
%%%%%%%%%%%%%%%%%

\maketitle

\begin{center}
Laboratoire Jean Kuntzmann,\\
Universit\'e de Grenoble and CNRS,\\
BP 53, 38041 Grenoble Cedex 9, France.
\end{center}

\begin{abstract}
We study nonlinear waves in {\em Newton's cradle}, a 
classical mechanical system consisting of a chain of beads attached to
linear pendula and interacting nonlinearly via Hertz's contact forces. 
We formally derive a spatially discrete modulation equation, for small amplitude nonlinear waves
consisting of slow modulations of time-periodic linear oscillations.
The fully-nonlinear and unilateral interactions between beads yield a 
nonstandard modulation equation that we call the  
{\em discrete $p$-Schr\"odinger} (DpS) equation. It
consists of a spatial discretization of a generalized Schr\"odinger
equation with $p$-Laplacian, with fractional $p>2$ depending on the
exponent of Hertz's contact force. 
We show that the DpS equation admits explicit periodic
travelling wave solutions, and numerically find a plethora of
standing wave solutions given by the orbits of a discrete map,
in particular spatially localized breather solutions.
Using a modified Lyapunov-Schmidt technique, we prove
the existence of exact periodic travelling waves in the 
chain of beads, close to the
small amplitude modulated waves given by the DpS equation.
Using numerical simulations, we show 
that the DpS equation captures several other important features of the dynamics
in the weakly nonlinear regime, namely
modulational instabilities, the existence of static and travelling breathers,
and repulsive or attractive interactions of these localized structures.
\end{abstract}

%\vspace{2ex}

\section{\label{intro}Introduction and main results}

This paper concerns the mathematical analysis and numerical simulation of
nonlinear waves in
the {\em Newton's cradle}, a classical 
mechanical system consisting
of a chain of beads suspended from a bar by inelastic strings (see figure \ref{boules}).
All beads are identical and behave like a linear pendula in the absence of contact
with nearest neighbours, but mechanical constraints between touching
beads introduce geometric nonlinearities.

In the last decades, this system has been considered as a
reference problem to test multiple impacts laws, 
aimed at evaluating e.g. the post-impact velocities of all beads
after a bead is released at one end of the cradle
\cite{ceanga,acary,liu1,liu2}
(see also \cite{hutzler} for additional references).
This problem is very delicate, because the collisional dynamics
involves nonlinear elastic waves that propagate along the granular chain.

One of the important factors that influence wave propagation 
and multiple collisions is the nature of elastic interactions between beads
(see \cite{hinch,ma,sekimoto} and references therein).
Hertz's theory \cite{ll,johnsonbook} allows to compute the 
repulsive force between two identical and initially tangent spherical beads 
that are compressed and slightly flatten. When  
the distance between their centers decreases by $\delta \approx 0$ (see figure \ref{boules}), 
the repulsive force $f$ is $f(\delta )=k\, \delta^{\alpha}$ at leading order in $\delta$,
where $k$ depends on the ball radius and material properties and
$\alpha = 3/2$. This result remains valid for much more general geometries
\cite{spence,johnsonbook}, but $\alpha$ can be larger for irregular contacts
($\alpha=2$ in the presence of conical asperities \cite{fu}) or smaller for surfaces that
squeeze more easily ($\alpha=1$ for solid cylinders in contact at two ends).
Hertz's type contact forces have several properties that make the analysis of wave
propagation particularly difficult. Firstly, they consist of unilateral
constraints i.e. no force is present when beads are not in contact. 
Moreover, they are fully nonlinear for $\alpha >1$, and in that case classical linear
wave theory becomes useless. In addition $f^{\prime\prime}(0)$ is not defined
for $\alpha <2$, therefore the use of perturbative methods is
more delicate. 
Developping analytical tools to overcome these obstacles
is important, because the latter are at the heart of wave propagation in granular media.

\begin{figure}[h]
\begin{center}
\includegraphics[scale=0.15]{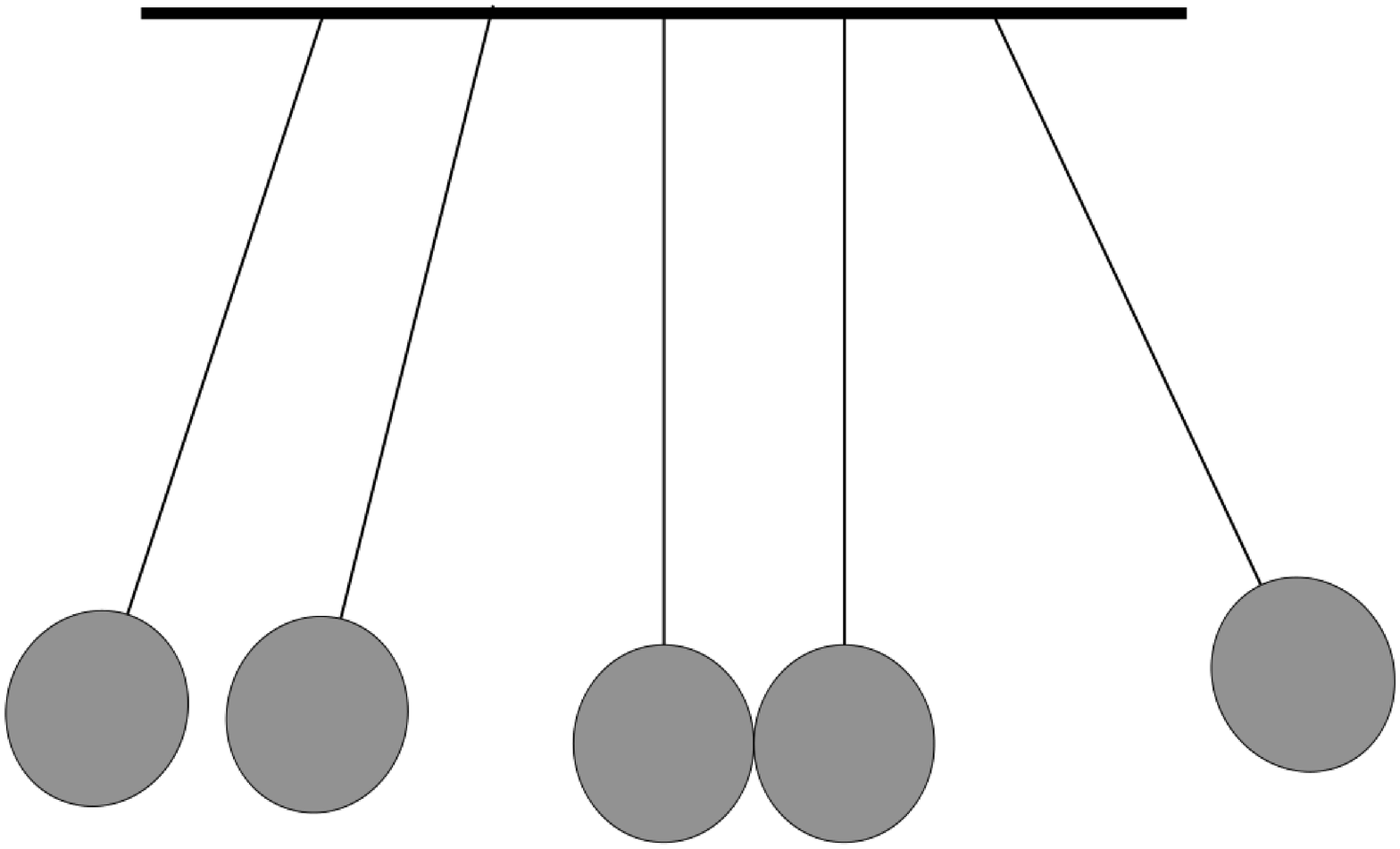}
\includegraphics[scale=0.1]{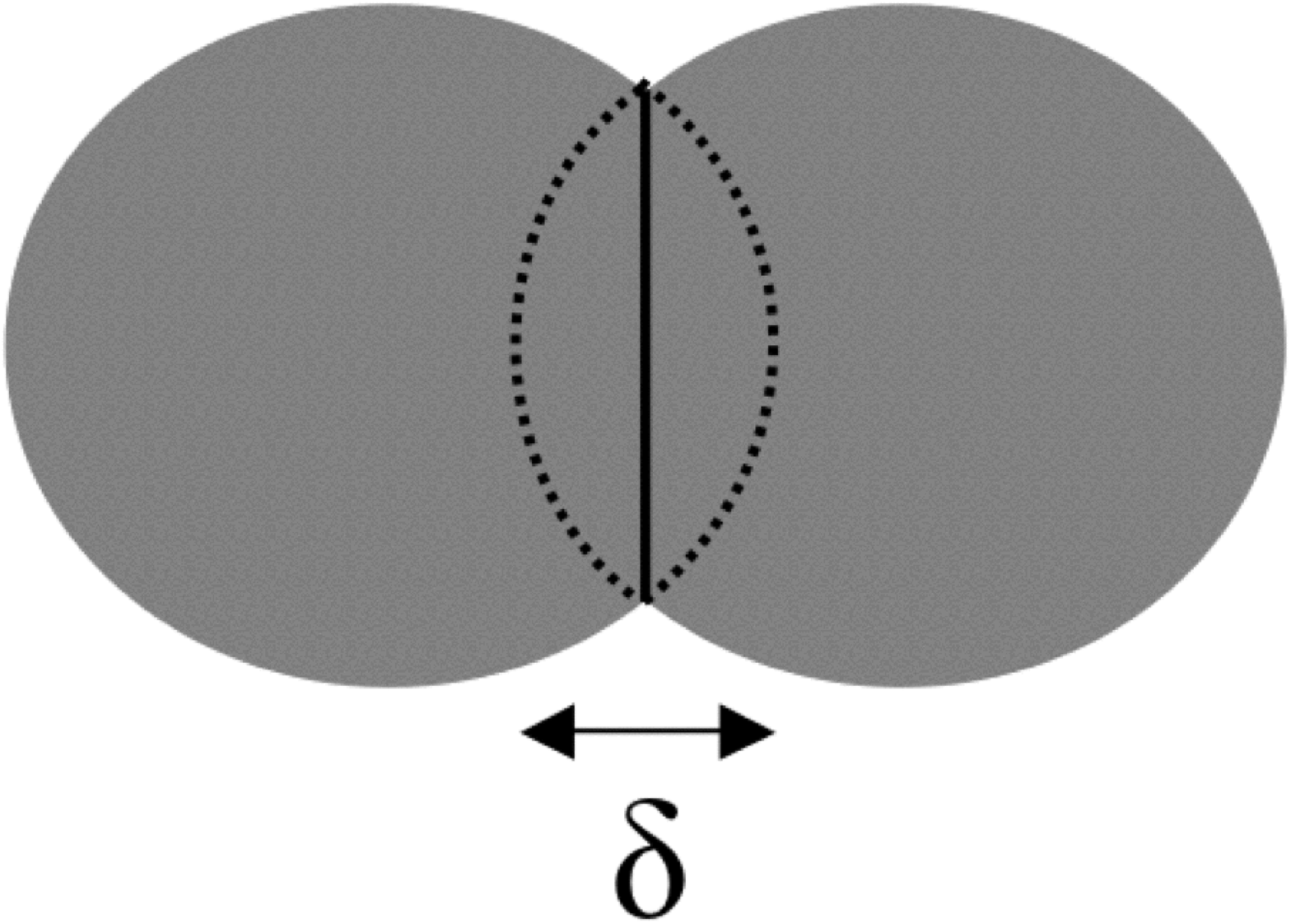}
\end{center}
\caption{\label{boules} 
Left~: Newton's cradle. Right~:
schematic representation of two compressed beads.}
\end{figure}

\ve

A simplified model for Newton's cradle
reads in dimensionless form \cite{hutzler} 
\begin{equation}
\label{nc}
\frac{d^2 x_{n}}{dt^2}+ x_{n} = 
V^\prime(x_{n+1}-x_n)-V^\prime(x_{n}-x_{n-1}),
\ \ \
n\in \mathbb{Z},
\end{equation}
where $x_{n}(t)\in \mathbb{R}$ is
the horizontal displacement of the $n$th bead from the ground-state,
i.e. the equilibrium position
at which each pendulum is vertical, beads are in contact at a single point
and uncompressed. 
The interaction potential $V$ takes the following form as $x \rightarrow 0$
\begin{equation}
\label{vhertz}
V(x)=V_0(x) + H(-x)\, W(x), \ \ \
V_0(x)=\frac{1}{1+\alpha}\, |x|^{1+\alpha}\, H(-x) ,
 \ \ \
W(x)=o(|x|^{1+\alpha}),
\end{equation}
where $H$ denotes the Heaviside function vanishing on $\mathbb{R}^-$ and
equal to unity on $\mathbb{R}^+$, $o$~is the usual
Landau's symbol and $\alpha >1$ a fixed constant.
The potential $V_0$ corresponds to a generalized Hertz contact force
and $W$ incorporates higher-order corrections. 
System (\ref{vhertz}) is Hamiltonian with total energy
\begin{equation}
\label{ham}
{\mathcal H}=
\sum_{n\in \mathbb{Z}}{\frac{1}{2}\, (\frac{dx_{n}}{dt})^2+\frac{1}{2}\, x_n^2 +V(x_{n+1}-x_n)}
.
\end{equation}
Equation (\ref{nc}) considered on the infinite lattice $\mathbb{Z}$ applies
in principle to large ensembles of beads (so that boundary effects
can be neglected for waves propagating in their core), but this model is
also suitable to describe strongly localized standing waves
in relatively small systems.  

\ve
 
The simplest model to encompass 
difficulties inherent to Hertz's contact nonlinearities 
consists of a chain
of identical spherical beads without local oscillators, in contact with their neighbours at a single point
when the chain is at rest. Intensive research has been carried out on
the analytical description of compression pulses in this system. 
Nesterenko analyzed the problem using a formal continuum limit \cite{neste1,neste2},
and found an approximate pulse solution consisting of
a compression solitary wave with compact support
(see also \cite{ap} for a slightly different continuum limit,
and \cite{porter} for an extension to dimer chains). In addition,
Chatterjee \cite{chat} has subsequently
obtained an improved approximation revealing a double-exponential
decay of the wave profile. As shown by MacKay \cite{mackay}
(see also \cite{ji}), exact solitary waves exist since a theorem of Friesecke and Wattis \cite{friesecke} proving their
existence readily applies to the chain of beads with Hertz contact forces.
In addition, English and Pego \cite{english} have shown that
these solitary waves have a doubly-exponential decay in the case of
Hertz's force. More details on the dynamics of these solitary waves can be
found in the review \cite{sen}. 
In addition, much more properties of solitary waves are known
when an external load $f_0$ is applied at both ends of the chain and
all beads undergo a small compression $\delta_0$.
The dynamics around this new equilibrium state reduces to the one of 
the classical Fermi-Pasta-Ulam (FPU) lattice \cite{neste2,sen},
in the case of which solitary waves with exponential decay
\cite{friesecke,mackay,pego,smets,iooss,ioossj,pankovbook,schw,herrmann}
and two-soliton solutions \cite{hoffman} are known to exist,
small amplitude solitary waves are well described by the KdV equation \cite{pego,iooss,kal,sw,bambusi} and are stable \cite{pego2,pego3,pego4,hoffman3,hoffman,hoffman2}.

\ve

This mapping between the chain of compressed beads and the
FPU model is interesting, because the latter
is known to display a rich dynamical behaviour and
sustain many other kinds of nonlinear waves
(see \cite{cam,gal} for recent reviews on this topic).
Among the most fundamental excitations existing in FPU chains,
periodic travelling waves \cite{filip,iooss,dreyer,pankovbook,herrmann}
are particularly important to
understand energy propagation and dispersive shocks \cite{dreyer2}.
However, no solutions of this type are known for the 
the chain of beads in the absence of external load.
The above mentioned periodic travelling waves are degenerate
when $f_0 \rightarrow 0$, because the sound velocity
(i.e. the maximal velocity of linear waves) vanishes
as $f_0^{1/6}$. For this
reason the uncompressed chain of beads is commonly
denoted as a ``sonic vacuum"  \cite{neste2}.

\ve

In contrast with the above statement, we prove in this paper that
nonlinear periodic travelling waves exist in Newton's cradle, 
due to the interplay between the on-site oscillators and 
fully-nonlinear contact interactions among beads. 

\ve

\begin{theorem}
\label{existthm}
Consider a potential $V \in C^{2}(\mathbb{R})$ taking the form
(\ref{vhertz}) with $\alpha  >1$ and
$W^{\prime\prime}(x)=o(|x|^{\alpha -1})$ as
$x \rightarrow 0$.
There exists $a_0 >0$ such that for all $a \in (0, a_0)$ and
$q \in (-\pi ,\pi ]$, system (\ref{nc})-(\ref{vhertz}) 
admits a periodic travelling wave solution
\begin{equation}
\label{xleadingorder}
x_n (t)= a\, \sin{(q\, n -\omega_a t)}+a^\alpha v_a (q\, n -\omega_a t),
\end{equation}
with amplitude $a$ and wavenumber $q$,
where the wave frequency $\omega_a >1$ satisfies a nonlinear dispersion relation
\begin{equation}
\label{wleadingorder}
\omega_a =1 + \frac{2}{\tau_0}\, a^{\alpha -1}\, |\sin{(\frac{q}{2})}|^{\alpha +1}+o(a^{\alpha -1}),
\end{equation}
\begin{equation}
\label{tau0}
\tau_0 = \frac{\sqrt{\pi}\, (\alpha^2 -1) \Gamma{(\frac{\alpha -1}{2})}}{\alpha 2^\alpha \Gamma{(\frac{\alpha}{2})}},
\end{equation}
and $\Gamma{(x)}=\int_{0}^{+\infty}{e^{-t}\, t^{x-1}\, dt}$ denotes Euler's Gamma function.
The function $v_a$ is $2\pi$-periodic, odd and belongs to
$C^{2 }(\mathbb{R})$. It takes the form
\begin{equation}
\label{va}
v_a (\xi )= K_h P_h [V_0^\prime (\sin{(\xi +q)}-\sin{\xi})-V_0^\prime (\sin{\xi}-\sin{(\xi -q})) ] + R_a(\xi ), 
\end{equation}
where $\| R_a\|_{L^\infty (\mathbb{R})} \rightarrow 0$ as $a\rightarrow 0$ and the linear operators
$P_h$, $K_h$ are defined by
\begin{equation}
\label{ph}
(P_h f )(\xi )=
f(\xi)-(\sin{\xi})\, \frac{2}{\pi} \int_{0}^{\pi}{\sin{(s)}\, f(s)\, ds}
,
\end{equation}
\begin{equation}
\label{kh}
(K_h f)(\xi )=(\sin{\xi})\, \frac{1}{\pi} \int_{0}^{\pi}{(s-\pi )\, \cos{(s)}\, f(s)\, ds}+
\int_{0}^{\xi}{\sin{(\xi -s)}\, f(s)\, ds}.
\end{equation}
\end{theorem}

These waves display several unusual features.
Firstly, the family of periodic travelling waves
(parametrized by $q$ and $a\approx 0$) is singular
when $\alpha \in (1,2)$ (which is the case for Hertz's contact law), in the sense that
the frequency $\omega_a$ defined by (\ref{wleadingorder})
is not differentiable with respect to $a$ at $a=0$.
This originates from the limited smoothness of
the potential $V_0$ at the origin ($V_0^{(3)}(0^-)$ is not defined).
Moreover, as a result of fully nonlinear interactions between beads,
for all wavenumber $q$ the wave frequency
$\omega_{a}$ converges towards unity in the small amplitude limit 
$a\rightarrow 0$ (i.e. the linear phonon band reduces to a single frequency).

The periodic travelling waves (\ref{xleadingorder})-(\ref{wleadingorder}) 
carry an energy flux in the direction of wave propagation when $q \notin \{0,\pi \}$. 
In the special case
$q=0$, solution (\ref{xleadingorder}) reduces to $x_n(t)=-a\, \sin{(t )}$ and corresponds
to in-phase oscillations of uninteracting linear pendula. For $q=\pi$ 
these solutions correspond to binary oscillations
$x_n (t)= a\, (-1)^{n+1}\, (\sin{( \omega_a t)}+a^{\alpha -1}  v_a (\omega_a t))$.

The proof of theorem \ref{existthm} proceeds in two steps.
Firstly, we formally derive in section \ref{gdnls} an amplitude equation
\begin{equation}
\label{dfnls}
2i \tau_0 \frac{\partial A_n}{\partial \tau}=
(A_{n+1}-A_n)\, |A_{n+1}-A_n |^{\alpha -1} -
(A_{n}-A_{n-1})\, |A_{n}-A_{n-1} |^{\alpha -1},
\end{equation}
which describes small amplitude approximate solutions of (\ref{nc})
taking the form 
\begin{equation}
\label{ansatz}
x_n^{\rm{app}} (t)=\epsilon\, (A_n(\tau )\, e^{it} + \bar{A}_n(\tau) \, e^{-it}) ,
\ \ \
\tau = \epsilon^{\alpha -1} t ,
\end{equation}
where $\epsilon >0$ is a small parameter and $A_n(\tau ) \in \mathbb{C}$.
The approximate solutions determined by
(\ref{dfnls})-(\ref{ansatz})
consist of slow modulations in time of solutions of 
problem (\ref{nc}) linearized at $x_n =0$.
The ansatz (\ref{ansatz}) and amplitude equation (\ref{dfnls})
are consistent with the nonlinear problem (\ref{nc}), in the sense that (\ref{ansatz})
approximately satisfies
(\ref{nc}) up to an $o(\epsilon^\alpha)$ error term.
We show in section \ref{twsol} that (\ref{dfnls}) admits an explicit family of periodic
travelling wave solutions, which provide the harmonic part of
(\ref{xleadingorder}) and the leading order terms of the
dispersion relation (\ref{wleadingorder}).
In a second step, we prove in section \ref{fixedpoint} the existence of exact 
periodic travelling wave solutions
of (\ref{nc}) close to these approximate solutions.
For this purpose we consider an advance-delay differential equation
that determines travelling waves of given period. We solve this problem
by adapting the Lyapunov-Schmidt technique
to the kind of nonlinearity with limited smoothness present in (\ref{nc})
(using the contraction mapping theorem instead of the implicit function
theorem employed in the usual case).

Our approach has
the advantage of giving explicitly the principal part of
travelling waves in the small amplitude limit, since
expressions (\ref{xleadingorder})-(\ref{va}) provide
an approximation of the exact solutions up to
an $o(a^\alpha)$ error term. In fact,  
we will show that far more complex approximations
of order $o(a^\beta)$ can be obtained for all $\beta >0$,
which constitutes an unexpected result given the limited
smoothness of $V$ at the origin.

In section \ref{stable} we complete theorem \ref{existthm} 
by numerical results, in order to study the limits of the
above analysis when one leaves the 
small amplitude regime. Along the same line,
it would be interesting to analyze periodic travelling wave
solutions of (\ref{nc}) using variational methods
\cite{filip,dreyer,pankovbook,herrmann} or degree theory 
\cite{filip}, since this should relax the assumption of
small amplitude waves (with the limitation of
describing travelling
wave profiles and dispersion relations with
less precision than perturbative methods). 

\ve

We call equation (\ref{dfnls}) the time-dependent
{\em discrete $p$-Schr\"odinger} (DpS) equation, because 
it can be seen as a finite-difference spatial discretization
of the generalized Schr\"odinger equation
$i\, {\partial_\tau A}= \Delta_p A$,
where the usual Laplacian operator is replaced by the one-dimensional
$p$-Laplacian $\Delta_p A = \partial_\xi (\, \partial_\xi A\,  |\partial_\xi A |^{p-2}\, )$
and $p=\alpha +1$.
Its fully nonlinear structure and fractional power nonlinearities
originate from
the fully nonlinear interaction forces of (\ref{nc}).
Interestingly, the unilateral character of Hertz's contact forces
is averaged out along the fast oscillations of the pendula,
which results in much simpler nonlinearities at the
level of the DpS equation.
In the present context, the DpS equation substitutes to the classical
discrete nonlinear Schr\"odinger (DNLS) equation
\begin{equation}
\label{dnls}
i  \frac{\partial A_n}{\partial \tau}=
A_{n+1}-2A_n +A_{n-1} \pm A_n\, |A_n |^{2},
\end{equation}
which can be derived for Hamiltonian lattices with 
smooth (at least $C^4$) potential energies and linear dispersion
\cite{kivshar,daumont,morgante,eilbeckj}.

\ve

Equation (\ref{dfnls}) tells in fact much more on the dynamics
of Newton's cradle than what is described in theorem \ref{existthm}.
In section \ref{unstable} we numerically
observe that certain periodic
travelling waves of (\ref{nc}) are unstable through the phenomenon
of modulational instability. Their envelope self-localizes under the
effect of small perturbations, yielding spatially localized and 
time-periodic intermittent compressions of the beads.
These localized oscillations propagate along the chain 
and sometimes remain pinned at some lattice sites. 
These waves correspond to {\em discrete breathers}, 
a class of nonlinear excitations ubiquitous in spatially
discrete systems \cite{flg,maca,james,aubkadel,guill,pankovbook}.
They have been studied both experimentally and theoretically in 
different types of granular chains, in the absence of local oscillators and under precompression 
\cite{boe,theo}. We check numerically that the DpS equation accurately describes the modulational
instability of certain small amplitude periodic travelling waves, 
provided one avoids near-critical cases and higher-harmonic instabilities
(section \ref{unstable}). Moreover, numerical 
computations reveal that breather solutions of the DpS equation 
(determined as homoclinic orbits of a discrete map)
describe with high precision the envelope of small amplitude breathers of (\ref{nc})
(sections \ref{twsol} and \ref{loc}).  
The DpS equation also qualitatively reproduces 
differents kinds of slow interactions of discrete breathers, namely
the fissionning of a localized perturbation into slowly travelling breathers, and  
the merging of two breathers into a single one after a long transcient
(section \ref{loc}). More generally, we numerically find a plethora of
standing wave solutions of the DpS equation
given by the orbits of an area-preserving and reversible map (sections \ref{twsol}), which
suggests the existence of many time-periodic standing wave solutions
of (\ref{nc}) with a huge variety  of spatial behaviours.

\ve

The paper is organized as follows. Section \ref{gdnlsgen} contains
the formal derivation of the DpS equation and a short discussion
of some travelling wave and standing wave solutions. The
proof of theorem \ref{existthm} is given in section \ref{fixedpoint} and
most numerical computations performed in section \ref {numeric}.  
Lastly, section \ref{conclu} discusses new perspectives and open problems
resulting from the present work.

\section{\label{gdnlsgen}The DpS equation for modulated waves}

The aim of this section is twofold. Section \ref{gdnls} presents a formal derivation of
the DpS equation, starting from Newton's cradle equations (\ref{nc}). One has to
stress that this gives by no means a justification that the DpS equation
approximates the original system for appropriate initial conditions and timescales
(these problems will be studied analytically in section \ref{fptw} for exact travelling waves, and numerically in section
\ref{numeric} for various initial conditions). In section \ref{twsol}, we study particular classes of
solutions of the DpS equation, i.e. periodic travelling waves and standing waves. 
Periodic travelling waves are computed explicitly, a result which will serve as a
basis for the analytical and numerical studies of sections \ref{fixedpoint} and \ref{numeric}.
Standing waves are studied numerically for $\alpha = 3/2$, with a special emphazis on spatially localized standing waves
(these solutions are important to describe modulational instabilities in system (\ref{nc}),
as section \ref{numeric} will show).

\subsection{\label{gdnls}Formal derivation and elementary properties}

In this section we formally derive the DpS equation (\ref{dfnls})
from the original dynamical equations (\ref{nc}), using a
multiscale expansion technique. We look for
solutions of (\ref{nc}) in the form of slow
modulations in time of $2\pi$-periodic functions
($2\pi$ being the period of the linear on-site oscillators). 
More precisely, we set
\begin{equation}
\label{ansatzX}
x_n (t)=X_n (\tau, t),
\end{equation}	
where $\tau = \tilde{\mu} t$ corresponds to a slow time,
$\tilde{\mu} >0$ is a small parameter, and 
\begin{equation}
\label{pbc}
X_n (\tau , t+2\pi) = X_n (\tau , t)
\end{equation}
for all $\tau , t \in \mathbb{R}$. 
We replace (\ref{nc}) by the infinite system of coupled PDE
\begin{equation}
\label{ncX}
[\, (\tilde{\mu} \partial_\tau + \partial_t )^2 + 1 \, ]\, 
X_{n} = 
V^\prime(X_{n+1}-X_n)-V^\prime(X_{n}-X_{n-1}),
\ \ \
n\in \mathbb{Z}, 
\ \ \
(\tau , t ) \in \mathbb{R}^2
\end{equation}
with periodic boundary conditions (\ref{pbc}).
The restriction to the line $\tau =\tilde{\mu} t $
of any solution of (\ref{ncX}) defines
a solution of (\ref{nc}).
Now we look for a family of small amplitude 
solutions of (\ref{ncX}) corresponding to
nearly harmonic oscillations
\begin{equation}
\label{ansatz2}
X_n^{\epsilon} (\tau , t)=\epsilon\, (A_n^\epsilon (\tau )\, e^{it} + \bar{A}_n^\epsilon (\tau) \, e^{-it}) 
+\epsilon^\beta R_n^\epsilon ( \tau , t )
\end{equation}
and parametrized by $\epsilon > 0$ close to $0$. 
The complex amplitude $A_n^\epsilon$ varies slowly in time, but no assumptions are
made on its spatial behaviour.
The remainder $R_n^\epsilon \in \mathbb{R}$ satisfies
$\int_0^{2\pi}{R_n^\epsilon( \tau , t)\, e^{\pm i t}\, dt}=0$ and we fix $\beta >1$.
The assumptions made on $V$ in theorem \ref{existthm}
imply
\begin{equation}
\label{vhertzd}
V^\prime (x)=V_0^\prime (x) + H(-x)\, o(|x|^{\alpha}), \ \ \
V_0^\prime (x)=- |x|^{\alpha}\, H(-x) .
\end{equation}
This leads us to
fix $\tilde{\mu} = \epsilon^{\alpha -1}$ and $\beta = \alpha$, so that 
in equation (\ref{ncX})
the time-modulation term
$\tilde{\mu} \partial^2_{\tau  t} X_n^{\epsilon}$, the nonlinear interaction forces
$V^\prime$ and the remainder $\epsilon^\beta R_n^\epsilon$
have the same order $\epsilon^\alpha$.
Setting $A_n^\epsilon =A_n + o(1)$ and $R_n^\epsilon =R_n + o(1)$
as $\epsilon \rightarrow 0$,
inserting (\ref{ansatz2}) in equation (\ref{ncX}) and multiplying the latter by 
$\epsilon^{-\alpha}$, one obtains
\begin{eqnarray}
\label{calculapp}
&&
2 i \partial_\tau A_n\, e^{it} - 2 i \partial_\tau \bar{A}_n\, e^{-it}
+(\,  \partial_t^2 + 1 \, )\, R_n  \\
\nonumber
&=&V_0^\prime [\, (A_{n+1}-A_n)\, e^{it} +  c.c.   \, ]
-V_0^\prime [\, (A_{n}-A_{n-1})\, e^{it} +  c.c.   \, ] +o(1)
 \end{eqnarray}
as $\epsilon \rightarrow 0$, where we denote by $c.c.$ the complex conjugates.
Now assume that
a family of solutions $X_n^{\epsilon}$ exists for all $\epsilon \approx 0$.
Letting $\epsilon \rightarrow 0$ in
the above equation yields
\begin{equation}
\label{amp}
2 i \partial_\tau A_n = f(A_{n+1}-A_n)-f(A_{n}-A_{n-1}),
\end{equation}
$$
f(z)=\frac{1}{2\pi}\int_0^{2\pi}{ V_0^\prime (\, z\, e^{it} +  \bar{z}\, e^{-it}    \, )  \, e^{- i t}\, dt}
$$
by projection on the first Fourier harmonic.
Moreover, for all fixed $\tau$
the $2\pi$-periodic function $R_n (\tau , .)$ is determined as a 
function of $A_{n\pm 1}(\tau )$ and $A_{n}(\tau )$,
and defined as the unique $2\pi$-periodic solution of
$$
(\,  \partial_t^2 + 1 \, )\, R_n = V_0^\prime [\, (A_{n+1}-A_n)\, e^{it} +  c.c.   \, ]
-V_0^\prime [\, (A_{n}-A_{n-1})\, e^{it} +  c.c.   \, ] - 2 i \partial_\tau A_n\, e^{it} +  c.c.
$$
being $L^2$-orthogonal to $e^{\pm i t}$.
Now there remains to evaluate $f(z)$.
Setting $z=r\, e^{i \theta}$ and using the change of variable $s=t+\theta$
in the integral defining $f$, one obtains
$$
f(z)=\frac{e^{i \theta }}{2\pi}
\int_{\mathbb{T}}{ V_0^\prime (\, 2 r \cos{s}    \, )  \, e^{- i s}\, ds}
$$ 
where $\mathbb{T}$ denotes any interval of length $2\pi$.
Now we fix $\mathbb{T}=(-\pi ,\pi )$ and 
use the fact that $V_0^\prime (x)=- |x|^{\alpha}\, H(-x)$.
One obtains
after elementary computations
\begin{equation}
\label{nonl}
f(z)=2^{\alpha -1}\, c_\alpha \, z \, |z|^{\alpha -1}
\end{equation}
where
\begin{equation}
\label{wallis}
c_\alpha = 
\frac{2}{\pi}
\int_{0}^{\pi /2}{ (\cos{t} )^{\alpha +1}   \, dt}
\end{equation}
is a Wallis integral with fractional index $\alpha +1$.
It follows that
\begin{equation}
\label{cal}
c_\alpha = 
\frac{1}{\pi}\frac{ \Gamma{(\frac{1}{2})}   \Gamma{(\frac{\alpha}{2}+1)}  }{  \Gamma{(\frac{\alpha +1}{2}+1)}   },
\end{equation}
where $\Gamma{(x)}=\int_{0}^{+\infty}{e^{-t}\, t^{x-1}\, dt}$ denotes Euler's Gamma function
(see \cite{abra}, formula 6.2.1 and 6.2.2 p.258).
Since $\Gamma{(1/2)}=\sqrt{\pi}$ and $\Gamma{(a+1)}=a\, \Gamma{(a)}$, we obtain finally
\begin{equation}
\label{thec}
c_\alpha = 
\frac{
2 \alpha \Gamma{(\frac{\alpha}{2})}
}
{
\sqrt{\pi} (\alpha^2 -1)  \Gamma{(\frac{\alpha -1}{2})}
}.
\end{equation}
As a conclusion,
introducing $\tau_0 = 2^{1-\alpha} c_\alpha^{-1}$,
equations (\ref{amp})-(\ref{nonl}) yield
\begin{equation}
\label{dps}
2i \tau_0 \frac{\partial A_n}{\partial \tau}=
(A_{n+1}-A_n)\, |A_{n+1}-A_n |^{\alpha -1} -
(A_{n}-A_{n-1})\, |A_{n}-A_{n-1} |^{\alpha -1}, \ \
n\in \mathbb{Z},
\end{equation}
i.e. one recovers the DpS equation (\ref{dfnls})
introduced in section \ref{intro} and
expression (\ref{tau0}) of coefficient $\tau_0$.

Note that
in the particular case of Hertz's law, one has
$$
\alpha = 3/2, \ \ \
\tau_0 = \frac{5  (\Gamma{(\frac{1}{4})})^2 }{24 \sqrt{\pi}} \approx 1.545
$$
(reference \cite{abra}, formula 6.1.32, 6.1.10 p.255-256).
Moreover, in the case of nearly linear unilateral interaction forces
(i.e. when $\alpha \rightarrow 1^+$), one can use
expression (\ref{wallis}) to obtain
$\lim\limits_{\alpha \rightarrow 1^+}\tau_0 = 2$.

\begin{remark}
\label{remros}
The discrete quasilinear Schr\"odinger (D-QLS) equation
\begin{equation}
\label{dqls}
i  \frac{\partial A_n}{\partial \tau}=
k\, \big[ \, 
(A_{n+1}-A_n)\, |A_{n+1}-A_n |^{2} -
(A_{n}-A_{n-1})\, |A_{n}-A_{n-1} |^{2}
\,\big]
+ A_n\, |A_n |^{2}
\end{equation}
was derived in reference \cite{ros}
for Hamiltonian lattices with
\begin{equation}
\label{ham2}
{\mathcal H}=
\sum_{n\in \mathbb{Z}}{\frac{1}{2}\, (\frac{dx_{n}}{dt})^2+\frac{1}{2}\, x_n^2 - \frac{1}{4}\, x_n^4+k\, (x_{n+1}-x_n)^4}
\end{equation}
($k >0$), where the purely quartic interaction potential
yields a nonlinear dispersion and the on-site potential is  
anharmonic. For stationary solutions $A_n (\tau )=a_n\, e^{i\, \Omega \tau }$,
the additional local nonlinearity makes equation (\ref{dqls})
easier to analyze than (\ref{dps}), because (\ref{dqls}) can be analytically studied
near the anticontinuum limit $k \rightarrow 0$ \cite{aubry,alfimov}.
In addition, the existence of discrete breather solutions of (\ref{dqls})
has been proved in reference \cite{ros}, as well as several results
concerning their stability, spatial decay and continuum limit.
Unfortunately, system (\ref{dqls}) does not apply to the Newton's craddle problem,
because nonlinear terms of the local potential $U(x)=1-\cos{x}$ of nonlinear pendula
are of higher order than the classical Hertz contact interactions
(for this reason on-site nonlinearities have been neglected in Hamiltonian (\ref{ham})). 
\end{remark}

\ve

Equation (\ref{dps}) possesses some interesting properties, in particular
conserved quantities. Let us consider 
spatially localized solutions of (\ref{dps})
satisfying $\{ A_n (\tau)\}\in \ell_2 (\mathbb{Z})$.
One can check that
the squared $\ell_2$ norm $\sum_{n\in \mathbb{Z}}{|A_n|^2}$ is invariant
by the flow of (\ref{dps}). 
In addition, (\ref{dps}) can be written as a Hamiltonian system
$$
\partial_\tau A_n = \frac{\partial H}{\partial p_n}, \ \ \
\partial_\tau p_n = -\frac{\partial H}{\partial A_n},
$$
with $p_n=i\, \bar{A}_n$ and the Hamiltonian
$$
H=-\frac{1}{(\alpha +1)\tau_0}\sum_{n\in \mathbb{Z}}{|A_{n+1}-A_n|^{\alpha +1}}.
$$
It follows that the $\ell_{\alpha +1}$ norm of the forward difference $\{ A_{n+1}-A_n  \}$
is also a conserved quantity. 
If $\{ A_n (\tau)\}\in \ell_1 (\mathbb{Z})$ then
there exists a third conserved quantity
$$
P=\sum_{n\in \mathbb{Z}}{A_{n}}.
$$
In addition the DpS equation admits the gauge invariance
$A_n \rightarrow A_n \, e^{i \varphi}$, the translational invariance $A_n \rightarrow A_n + c$ and
a scale invariance, since any
solution $A_n (\tau )$ of (\ref{dps}) generates a one-parameter family of solutions
$a\, A_n (|a|^{\alpha -1}\, \tau )$, $a\in \mathbb{R}$. This property explains
the structure of velocity-amplitude relations for travelling wave
solutions and the frequency-amplitude scaling of
standing wave solutions that will be obtained in the next section.

To end this section, we point out an interesting 
formal continuum limit of (\ref{dps})
for solutions $A_n(\tau)$ varying slowly in space and time
with an appropriate scaling. 
Setting $p=\alpha +1$,
$A_n(\tau)=\psi( \xi, s )$
with $\xi = h\, n$, $s=(2\tau_0)^{-1}h^{\alpha+1}\tau$ 
in equation (\ref{dps}) and
letting $h$ go to $0$, one obtains 
the generalized Schr\"odinger equation
\begin{equation}
\label{cgnls}
i\, {\partial_s \psi}= \Delta_p \psi, \ \ \
\xi \in \mathbb{R},
\end{equation}
where the usual Laplacian operator is replaced by the one-dimensional
$p$-Laplacian $\Delta_p \psi = \partial_\xi (\, \partial_\xi \psi\,  |\partial_\xi \psi|^{p-2}\, )$.
In addition, more general fully nonlinear generalized Schr\"odinger equations
would result from different
continuum limits, setting $A_n(\tau)=B( \xi, \tau )\, e^{i\, q\, n}$.

\subsection{\label{twsol}Time-periodic travelling and standing waves}

It is well known (see e.g. \cite{eilbeckj})
that the classical DNLS equation (\ref{dnls}) admits
explicit periodic travelling wave solutions of the form
\begin{equation}
\label{ansatz3}
A_n (\tau )=R\, e^{i\, (\Omega \tau - q n - \phi )}, \ \ \ R>0,
\end{equation}
where the frequency $\Omega$ is fixed by the amplitude $R$ and 
the wavenumber $q$ through a nonlinear dispersion relation.
The same property holds true for the DpS equation (\ref{dfnls})
due to its gauge invariance.
Inserting the ansatz (\ref{ansatz3}) in (\ref{dfnls}), one obtains after
elementary computations
\begin{equation}
\label{ds}
\tau_0 \, \Omega = 2^\alpha R^{\alpha -1} |\sin{(\frac{q}{2})}|^{\alpha +1}.
\end{equation}
Since (\ref{dfnls}) was derived as an amplitude equation from
the original system (\ref{nc}), combining the ansatz (\ref{ansatz3}) with
(\ref{ansatzX})-(\ref{ansatz2}) yields {\em approximate}
periodic travelling wave solutions of (\ref{nc}). These solutions
take the form
\begin{eqnarray}
\label{ansatzapprox1}
x_n^{\rm{tw}} (t)&=&\epsilon\, (A_n(\epsilon^{\alpha -1} t )\, e^{it} 
+ \bar{A}_n(\epsilon^{\alpha -1} t) \, e^{-it}) \\
\label{ansatzapprox}
&=& a\, \cos{(q n - \omega_{\rm{tw}} t + \phi)} ,
\end{eqnarray}
where we denote by $a=2 \epsilon R >0$ the wave amplitude and by
\begin{equation}
\label{freq}
\omega_{\rm{tw}} =1 + \frac{2}{\tau_0}\, a^{\alpha -1}\, |\sin{(\frac{q}{2})}|^{\alpha +1}
\end{equation}
its frequency. 

Interestingly,
these expressions are exact for $q=0$
and correspond to linear in-phase oscillations of all oscillators of (\ref{nc})
at frequency $\omega_{\rm{tw}}=1$. More generally,
as announced in theorem \ref{existthm} there exist in fact
exact periodic travelling wave solutions of (\ref{nc}) close
to the above approximate solutions. The proof of this statement
will be the object of section \ref{fixedpoint}.

The nonlinear dispersion relation (\ref{ds}) displays an
unusual feature, since  
for all wavenumber $q$ the wave frequency
$\omega_{\rm{tw}}$ converges towards $1$ in the small amplitude limit 
$a\rightarrow 0$.
This is due to the fact that dispersion originates
from fully nonlinear 
interactions between beads.
Additionally, one can notice that
the band of frequencies $\omega_{\rm{tw}}$ of these periodic travelling
waves lies above $\omega =1$.

\ve

Another special class of solutions of (\ref{dfnls}) takes the form
\begin{equation}
\label{ansatzst}
A_n (\tau )=R_n\, e^{i\, (\Omega \tau  + \phi )}, \ \ \ R_n \in \mathbb{R}.
\end{equation}
Introducing $a_n=2 \epsilon R_n$ and 
$\omega_{\rm{sw}} = 1+\Omega \epsilon^{\alpha -1}$, 
the approximate solutions deduced from (\ref{ansatz2}) read
\begin{equation}
\label{ansatzapproxst}
x_n^{\rm{sw}} (t)= a_n\, \cos{(\omega_{\rm{sw}} t + \phi)},
\end{equation}
and the DpS equation (\ref{dfnls}) becomes
\begin{equation}
\label{dpsst}
-\mu\, a_n =
(a_{n+1}-a_n)\, |a_{n+1}-a_n |^{\alpha -1} -
(a_{n}-a_{n-1})\, |a_{n}-a_{n-1} |^{\alpha -1} ,
\end{equation}
with 
$$
\mu = (\omega_{\rm{sw}} -1)\, 2^\alpha \tau_0.
$$
Solutions of the form (\ref{ansatzapproxst}) correspond to standing waves,
i.e. they oscillate periodically with time and their nodes and extrema
do not change.
We shall refer to (\ref{dpsst}) as the (real) stationary DpS equation.

\ve

Equation (\ref{dpsst}) has been previously derived 
for special classes of lattices that sustain standing wave
solutions $x_n (t) = a_n \, \varphi (t)$ with space-time separation
(see remark~\ref{rem} for more details).

\ve

\begin{remark}
\label{rem}
Equation (\ref{dpsst}) can be simply derived
in some class of nonlinear lattices
with special interaction potentials.  
Consider the symmetrized potential
$
U(x)=V_0(-|x|)= {(1+\alpha)}^{-1}\, |x|^{\alpha +1} 
$
satisfying $U^\prime (x) = x\, |x|^{\alpha -1}$
and the nonlinear system
\begin{equation}
\label{meca}
\frac{d^2 x_{n}}{dt^2} + \kappa \, x_n= 
U^\prime(x_{n+1}-x_n)-U^\prime(x_{n}-x_{n-1}),
\ \ \
n\in \mathbb{Z} .
\end{equation}
System (\ref{meca}) admits some simple mechanical interpretations.
When $\kappa >0$, it describes e.g. the oscillations
a chain of linear pendula coupled by anharmonic 
torsion springs with potential $U$.
For $\kappa <0$ the same analogy can be carried out,
replacing the linear pendula by inverted pendula.
For $\kappa =0$, system (\ref{meca}) reduces to
the Fermi-Pasta-Ulam lattice with potential $U$,
that represents e.g. a chain of masses coupled by anharmonic springs.
Equilibria of (\ref{meca}) satisfy equation (\ref{dpsst})
with $\mu = -\kappa$. Moreover,
due to the special form of $U$, one can find
standing wave solutions of (\ref{meca}) taking the form
$x_n (t) = a_n \, \varphi (t)$,
where $\varphi$ is a periodic solution of
$$
\frac{d^2 \varphi}{dt^2} + \kappa \, \varphi + \mu U^\prime (\varphi )
=0
$$
and $\{ a_n \}$ is a bounded solution of (\ref{dpsst}).
For $\kappa =0$ and when 
$\alpha \geq 3$ is an odd integer, spatially localized
standing wave solutions of (\ref{meca}) have been 
studied in the physics literature
\cite{kivcomp,flach2,flachb,sievpage}
using equation (\ref{dpsst}).
The case $\kappa >0$ and $\alpha = 3$ has been studied
in \cite{fdmf} along the same lines.
\end{remark}

\ve

The simplest case of equation (\ref{dpsst})
corresponds to $\mu =0$ (i.e. $\omega_{\rm{sw}}=1$), 
where it follows that $a_{n+1}-a_n$ is constant, and thus all solutions
of (\ref{dpsst}) read $a_n = \lambda\, n + \beta$, $\lambda , \beta \in \mathbb{R}$. In that case, (\ref{ansatzapproxst}) yields
in fact exact solutions of (\ref{nc}) taking the form
$x_n (t)= ( \lambda\, n + \beta )\, \cos{( t + \phi)}$ and corresponding to collective in-phase oscillations.

\ve

For $\mu \neq 0$, it is interesting to note that $\tilde{a}_n = |\mu|^{\frac{1}{1-\alpha}}\, a_n$
satisfies the renormalized equation
\begin{equation}
\label{dpsstn}
-{\rm sign}(\mu )\, \tilde{a}_n =
(\tilde{a}_{n+1}-\tilde{a}_n)\, |\tilde{a}_{n+1}-\tilde{a}_n |^{\alpha -1} -
(\tilde{a}_{n}-\tilde{a}_{n-1})\, |\tilde{a}_{n}-\tilde{a}_{n-1} |^{\alpha -1} ,
\end{equation}
where by definition $\rm{sign}(\mu ) =1$ for $\mu >0$ and
$\rm{sign}(\mu ) =-1$ for $\mu <0$.
Consequently it is sufficient to restrict to the cases $\mu \in \{-1,0,1 \}$ to find
all solutions of (\ref{dpsst}).
In what follows we fix $\mu = \pm 1$ and drop the tilde in equation (\ref{dpsstn})
for notational simplicity.

\ve
 
Multiplying (\ref{dpsstn}) by $a_n$ and summing over $n$,  one obtains
for all $\{ a_n \} \in \ell_2 (\mathbb{Z}) $
$$
{\rm sign}(\mu ) \sum_{n\in \mathbb{Z}}{a_n^2}=\sum_{n\in \mathbb{Z}}{|a_{n+1}-a_n|^{\alpha +1}},
$$
where an index change has been performed to simplify the right-hand side.
Consequently, nontrivial localized solutions $\{ a_n \} \in \ell_2 (\mathbb{Z}) $ can be found
only for $\mu >0$, and the same conclusion holds true for periodic solutions
(just slightly modifying the above computation). As a consequence
we shall restrict ourselves to the case $\mu =1$.

\ve

In order to reformulate (\ref{dpsstn})
as a two-dimensional mapping, we introduce the auxiliary variable
\begin{equation}
\label{homeo}
b_n = (a_{n}-a_{n-1})\, |a_{n}-a_{n-1} |^{\alpha -1} .
\end{equation}
Inverting (\ref{homeo})
and shifting the index $n$, one finds
\begin{equation}
\label{comp1}
a_{n+1}=a_{n}+b_{n+1} \, |b_{n+1}|^{\frac{1}{\alpha}-1} .
\end{equation}
Moreover, equation (\ref{dpsstn}) at rank $n+1$ reads 
\begin{equation}
\label{comp2}
b_{n+2}= b_{n+1}-\, a_{n+1}.
\end{equation}
Introducing the variable 
$$U_n =(a_n , b_{n+1}),$$
equations (\ref{comp1})-(\ref{comp2}) define a two-dimensional mapping
$U_{n+1}=G (U_n)$. 
One can check that ${\rm det}\, DG(U)=1$ for all $U\in \mathbb{R}^2 \setminus \{0\}$, 
which implies that $G$ is an area-preserving map. 
Moreover, the mapping possesses the invariance $U_n \rightarrow -U_n$
since $G(-U)=-G(U)$.
In addition the map $G$ is reversible with respect to the symmetry $R$ defined by
$R(a,b)=(a,-a-b)$, i.e. one has $R\, G^{-1} = G \circ R$. Equivalently,
$\{ U_n \}$ is an orbit of $G$ if and only if $\{ R\, U_{-n} \}$ is an orbit of $G$. 
This property originates from the invariance $n\rightarrow -n$ of (\ref{dpsstn}).

\ve

In what follows we consider (\ref{comp1})-(\ref{comp2}) for $\alpha = 3/2$.
Several orbits of $G$ are shown in figures \ref{traj3} and \ref{traj2}.
Figure \ref{traj3} shows an orbit of $G$, for
an initial condition with $a_0 =0$ and $b_1$ 
very small. The trajectory suggests
the existence of pairs of symmetric orbits homoclinic to $0$, i.e. satisfying
$\lim\limits_{n\rightarrow \pm \infty}U_n=0$.
Figure \ref{traj2} shows many kinds of trajectories frequently encountered
in area-preserving or reversible mappings. 
Plots in the left column show periodic orbits and invariant tori,
while right plots reveal intricate trajectories in the vicinity of stable and unstable manifolds.

\begin{figure}[!h]
\begin{center}
\includegraphics[scale=0.25]{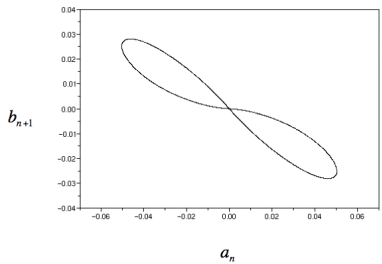}
\end{center} 
\caption{\label{traj3}
Orbit of $G$ for the initial condition $a_0 =0$, $b_1 = 10^{-13}$.}
\end{figure}

\begin{figure}[!h]
\begin{center}
\includegraphics[scale=0.55]{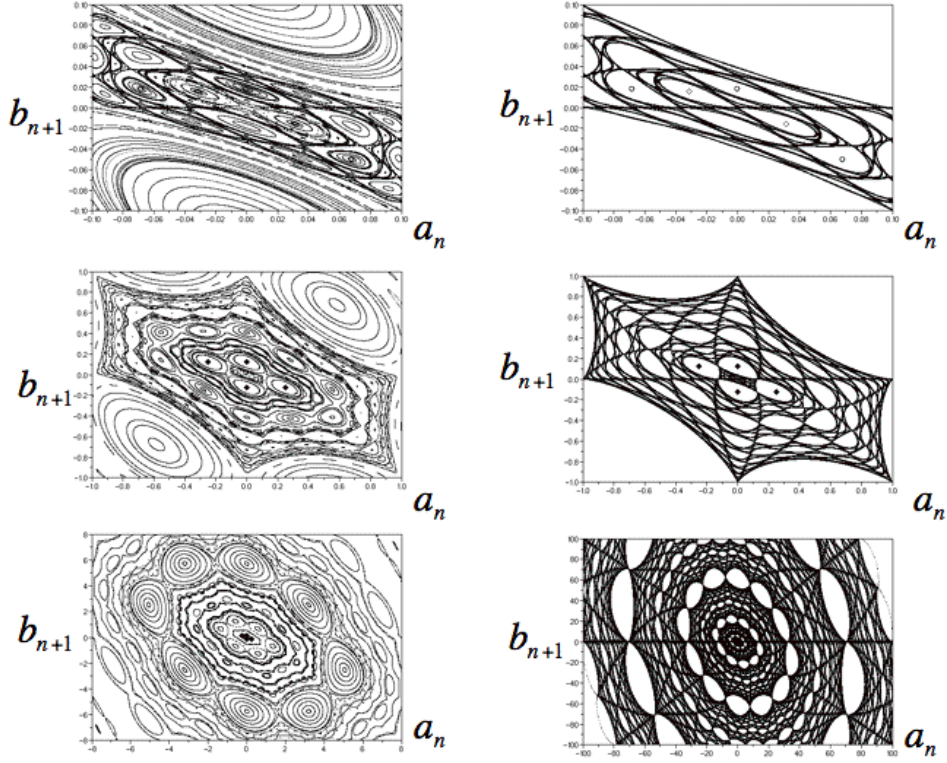}
\end{center} 
\caption{\label{traj2}
Trajectories of (\ref{comp1})-(\ref{comp2}) for $\alpha = 3/2$, in neighbourhoods of $U=0$ of different sizes
(top~: $\| U_n \|_\infty \leq 0.1$, middle~: $\| U_n \|_\infty \leq 1$, bottom left~: $\| U_n \|_\infty \leq 8$, 
bottom right~: $\| U_n \|_\infty \leq 100$). Left plots shows periodic orbits and invariant tori,
and right plots mainly focus on trajectories in the vicinity of stable and unstable manifolds.
Marks in the first and second plot of the right column correspond to simple periodic orbits
that can be explicitly computed, a period-$2$ orbit $a_n = 2^{\frac{\alpha +1}{1-\alpha}}\, (-1)^n$,  
a period-$3$ orbit $\{a_n \} = \{ \ldots , 0,a,-a,0,a,-a,\ldots  \}$ with
$a=(1+2^{\alpha})^{\frac{1}{1-\alpha}}$, 
a period-$4$ orbit $\{a_n \} = \{ \ldots , 0,-a,0,a,0,-a,\ldots  \}$ with
$a=2^{\frac{1}{1-\alpha}}$.
The trajectory of figure \ref{traj3} near the
double homoclinics is visible in the top right plot, surrounding the period-$2$ orbit.
}
\end{figure}

\ve

Now we concentrate on solutions homoclinic to $0$. To approximate
these solutions numerically, 
we consider equation (\ref{dpsstn}) with periodic boundary conditions $a_{n+N}=a_n$
(computations are performed for $N=50$), and look for solutions 
close to homoclinic ones. This approach is consistent with figure \ref{traj2},
which shows periodic orbits and invariant tori very close to the pair of
homoclinics. We numerically solve this nonlinear equation using the
fsolve function of the software package Scilab, based on a modified
Powell hybrid method. 

We start by computing an homoclinic solution of (\ref{dpsstn}) with 
site-centered symmetry $a_{N/2 -n}=a_{N/2 +n}$. 
The numerical iteration is initialized by
setting $a_{N/2}=-a$ and $a_n=0$ elsewhere.
We fix $a=0.1$, of the order of the size of the homoclinic loop of figure \ref{traj2}.
In that case the iteration converges towards a spatially localized symmetric
solution $\{ a_n \}$, whose profile is shown in the top left plot of
figure \ref{bcsc}. As shown by the top right plot in semi-logarithmic scale, 
the homoclinic solution has a super-exponential spatial decay.
This unusual feature comes from the fact that $G$ is not differentiable
at the fixed point $U=0$ (since $\alpha >1$), hence the classical results 
on exponential convergence along stable and unstable manifolds do not apply.

One can also compute an homoclinic solution of (\ref{dpsstn}) with 
bond-centered symmetry $a_{N/2 -n}=-a_{N/2 +n+1}$. 
In that case one starts the numerical iteration by
setting $a_{N/2}=a$, $a_{N/2 +1}=-a$  and $a_n=0$ elsewhere.
The iteration converges towards a spatially localized antisymmetric
solution $\{ a_n \}$ shown in figure \ref{bcsc}
(bottom left plot), which also decays super-exponentially. 
More complex homoclinic solutions can be
computed by changing the initial condition (one example is shown in the bottom right plot)
and the lattice size.

\begin{figure}[!h]
\psfrag{n}[0.9]{\huge $n$}
\psfrag{a}[1][Bl]{\huge $a_n$}
\psfrag{b}[1][Bl]{\huge $|a_n |$}
\begin{center}
\includegraphics[angle=-90,scale=0.22]{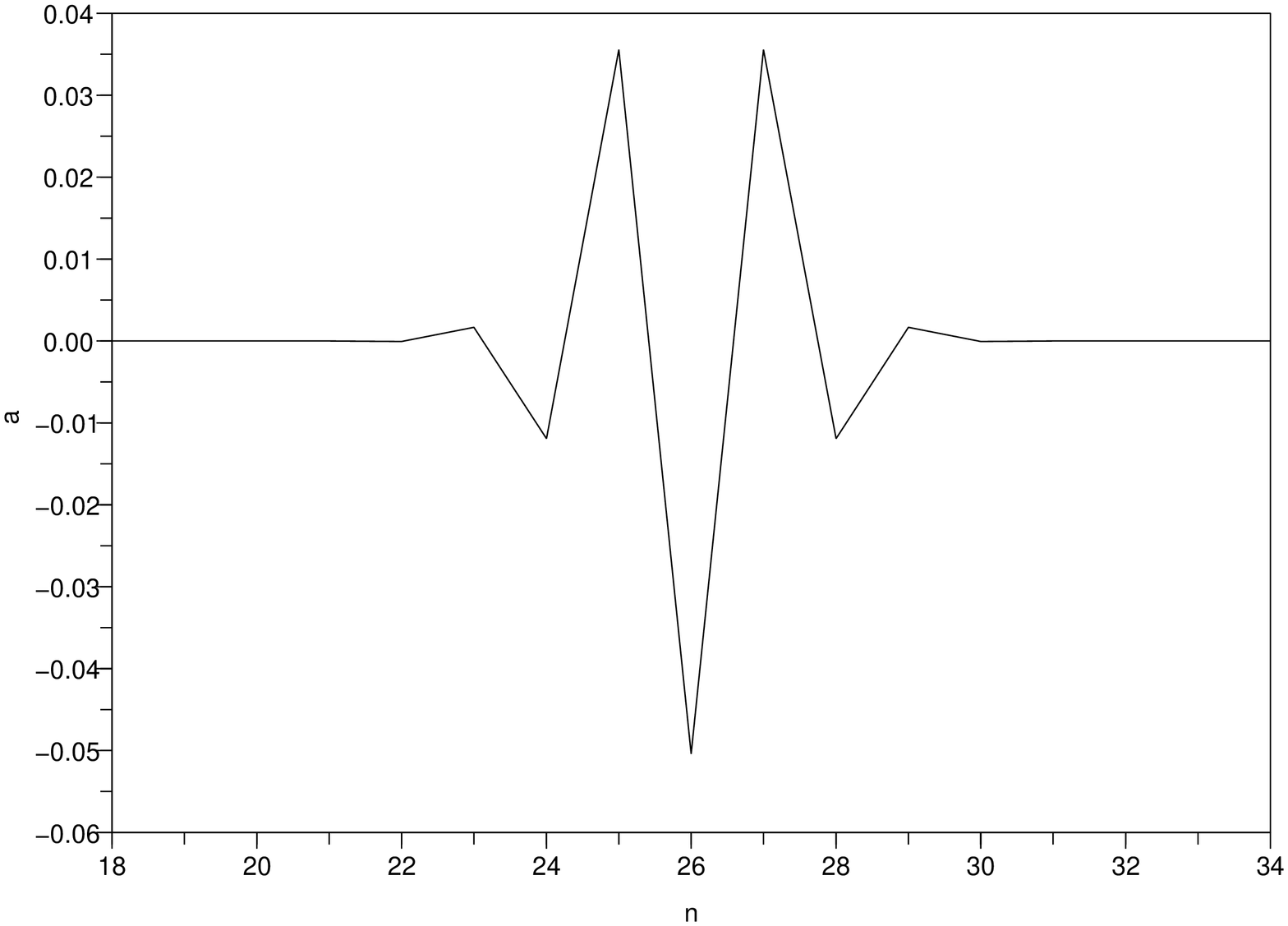}
\includegraphics[angle=-90,scale=0.22]{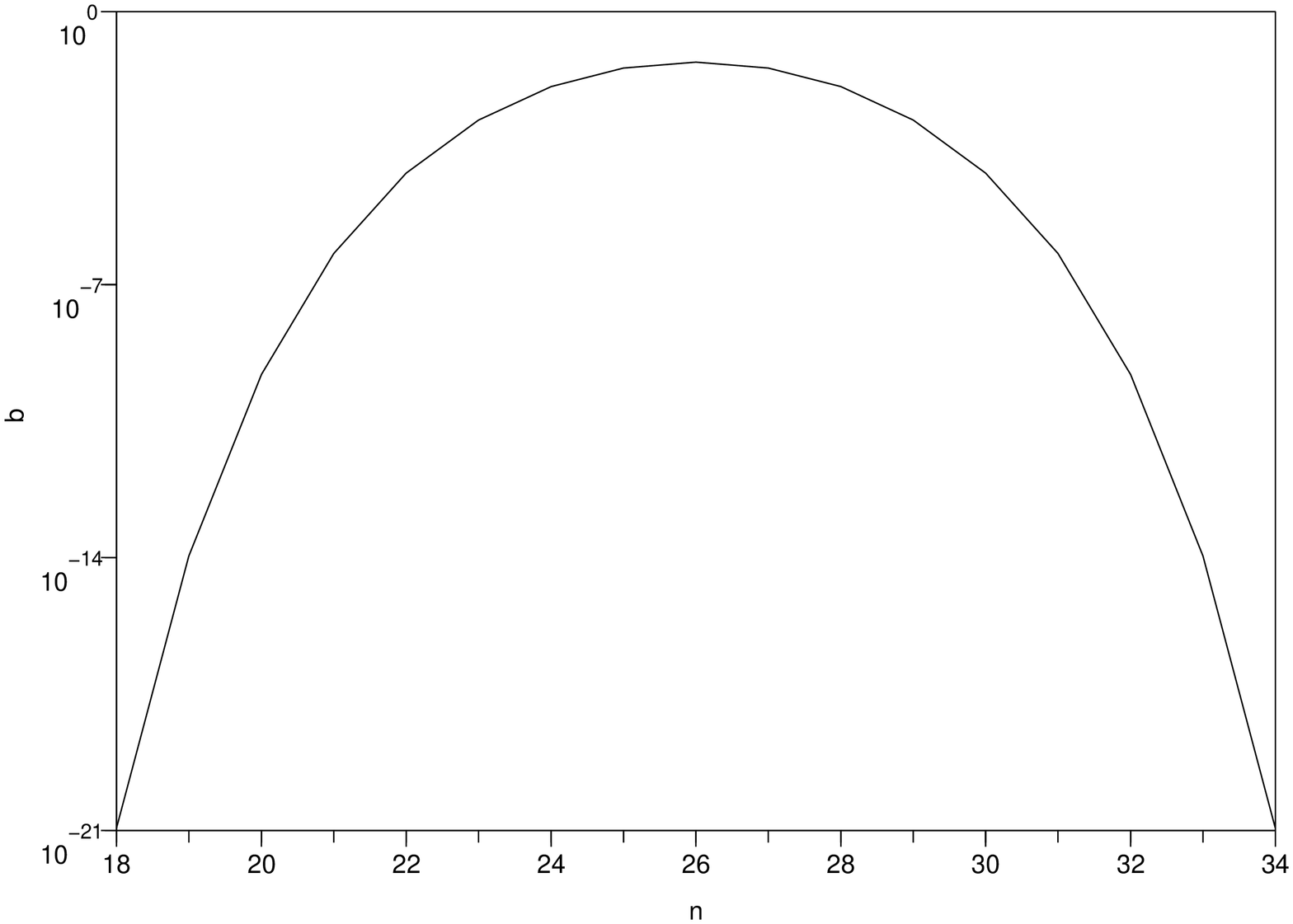}
\includegraphics[angle=-90,scale=0.22]{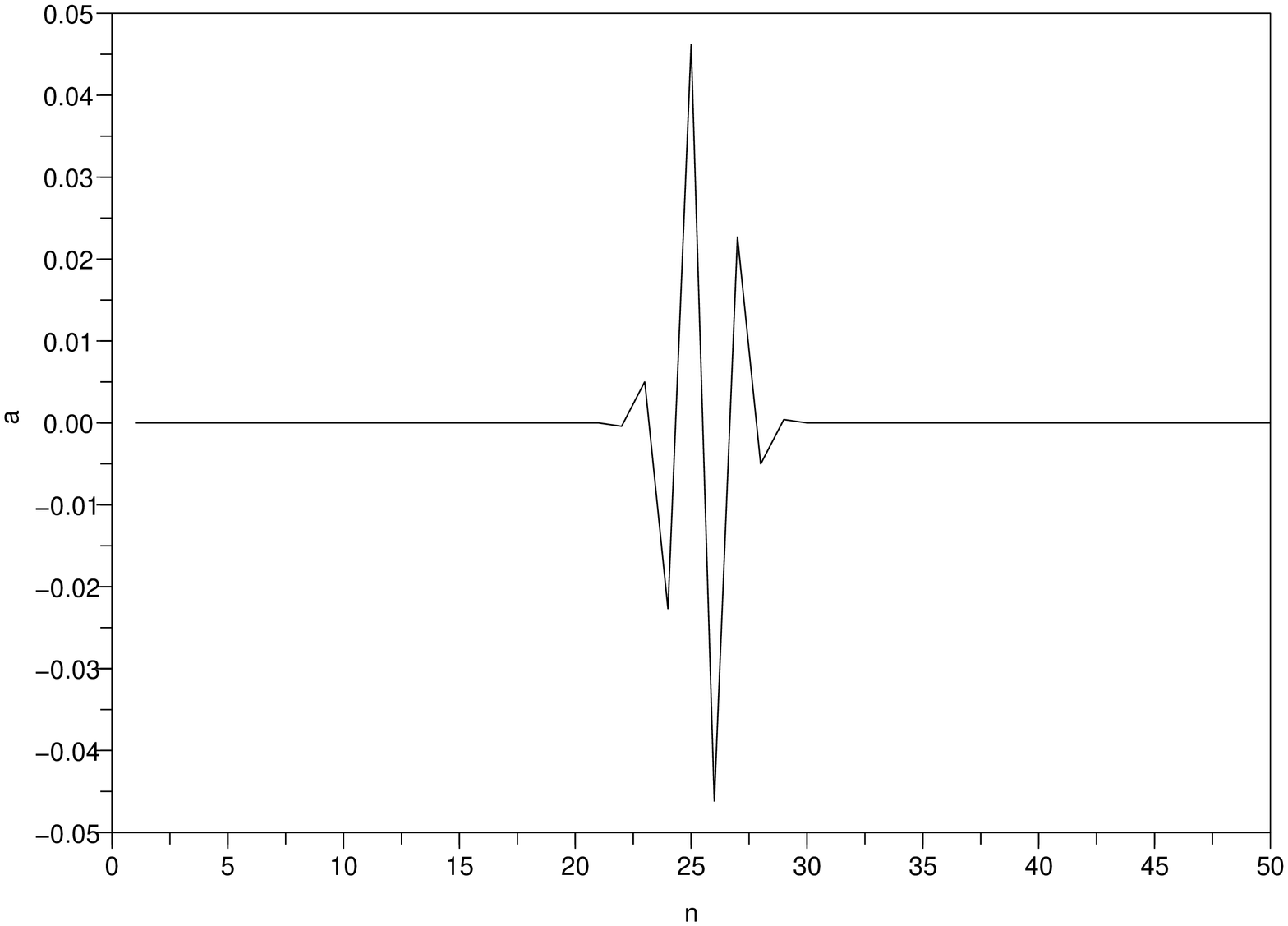}
\includegraphics[angle=-90,scale=0.22]{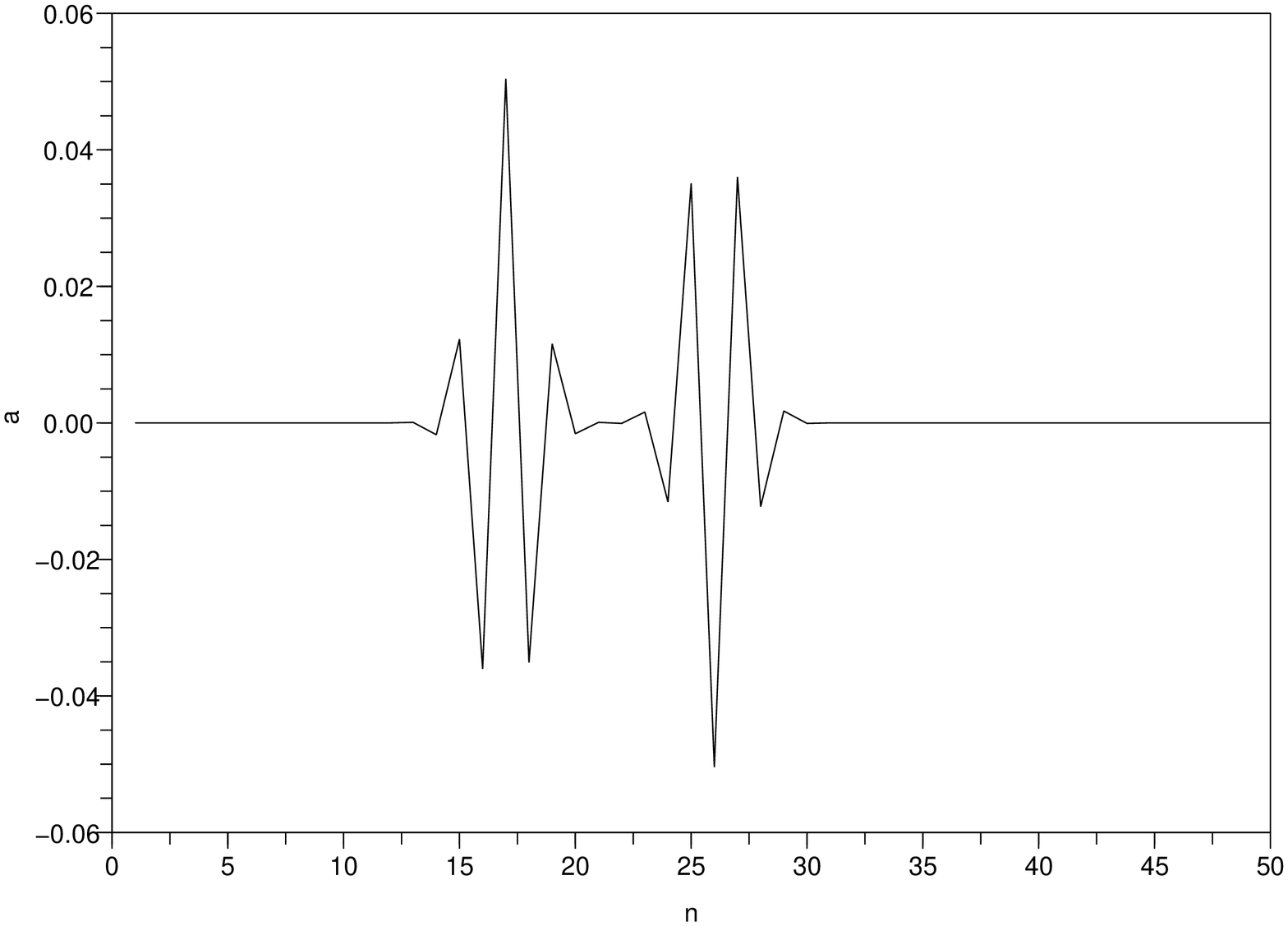}
\end{center} 
\caption{\label{bcsc} 
This figure shows three different solutions of (\ref{dpsstn}) homoclinic to $0$.
Top left plot~: spatially symmetric solution.
Top right plot~: amplitude of the same homoclinic solution in semi-logarithmic scale,
showing its super-exponential decay
(profiles are only plotted for $18\leq n \leq 34$,
because $|a_n|$ drops below machine precision at the other sites).
Bottom left plot~: spatially antisymmetric solution.
Bottom right plot~: double-humped homoclinic solution.}
\end{figure}

\ve

Any solution of (\ref{dpsstn}) corresponds to the approximate solution of (\ref{nc})
\begin{equation}
\label{ansatzapproxst2}
x_n^{\rm{sw}} (t)=  |\mu|^{\frac{1}{\alpha -1}}\, {a}_n\, \cos{(\omega_{\rm{sw}} t )},
\ \ \
\mu = (\omega_{\rm{sw}} -1)\, 2^\alpha \tau_0.
\end{equation}
In particular, any orbit homoclinic to $0$ corresponds via (\ref{ansatzapproxst2}) to 
a one-parameter family of time-periodic
and spatially localized approximate solutions of (\ref{nc}), i.e. 
{\em discrete breathers}. 
Interestingly, it follows from expression (\ref{ansatzapproxst2}) that
the spatial extension of these breathers is independent of their amplitude 
(or equivalently of their frequency). This property is unusual
(the spatial extension of classical discrete breathers diverges
when their amplitude goes to $0$ \cite{james,guill}) and originates
from the fully-nonlinear interactions between beads present in (\ref{nc}). 
It is an analogue for discrete breathers of a similar invariance present
in Nesterenko's solitary wave \cite{neste1,neste2}, whose
spatial extension is independent of amplitude and velocity.

\ve

More generally, this study
leads us to conjecture the existence of time-periodic standing wave solutions
of (\ref{nc}) having a large variety of spatial behaviours, e.g. spatially-periodic (with arbitrarily large periods),
quasi-periodic, spatially localized, homoclinic to periodic or spatially disordered.

\begin{remark}
Reference \cite{hutzler} provides interesting 
numerical and experimental results on
standing waves in Newton's cradle with a small number of beads,
in particular its relaxation towards in-phase oscillations when
energy dissipation during collisions is not neglected,
and the near-recurrence of specific modes of oscillation in
the conservative case.
\end{remark}
  
\section{\label{fixedpoint}Proof of theorem \ref{existthm}}

The aim of this section is to prove theorem \ref{existthm}, i.e.
show the existence of small amplitude exact periodic travelling wave solutions
of (\ref{nc}) close to the approximate travelling waves provided by the 
DpS equation. For this purpose we reformulate the search of 
periodic travelling waves as a fixed point problem in section \ref{fptw}, using a suitable
scaling of the solutions deduced from the formal analysis of section \ref{gdnlsgen}.
In section \ref{sfp} we solve the resulting equation
by the contraction mapping theorem. We end this section by
pointing out some symmetry properties of periodic travelling waves
(section \ref{symprop}).

\subsection{\label{fptw}Fixed point problem for periodic travelling waves}

We search for solutions of (\ref{nc}) taking the form
$x_n (t)= u(\xi )$, where 
$u$ is $2\pi$-periodic, 
$\xi =q\, n -\omega\, t$ denotes a moving frame coordinate,
and $q \in [-\pi , \pi ]$, $\omega \in \mathbb{R}^\ast$ are the wave parameters. 
Equation (\ref{nc}) yields the advance-delay differential equation
\begin{equation}
\label{ad}
\omega^2 u^{\prime\prime} (\xi ) + u(\xi ) =
V^\prime (u(\xi +q)-u(\xi ))-V^\prime (u(\xi )-u(\xi -q)), \ \ \
\xi \in \mathbb{R}.
\end{equation}
Notice that (\ref{ad}) possesses the symmetry $u(\xi ) \rightarrow -u({-\xi})$.
To simplify the analysis, we look for solutions of (\ref{ad}) invariant under
this symmetry, i.e. odd in $\xi$. 

We look for solutions with frequency $\omega$
close to $1$ as in section \ref{twsol}. 
Problem (\ref{ad}) can be rewritten in more compact form
\begin{equation}
\label{ad2}
u^{\prime\prime}  + u =
(1-\omega^2 )\, u^{\prime\prime}  +N(u),
\end{equation}
where 
\begin{equation}
\label{defn}
[N(u)](\xi)=V^\prime (u(\xi +q)-u(\xi ))-V^\prime (u(\xi )-u(\xi -q)).
\end{equation}

Now we reformulate equation (\ref{ad2}) as a suitable fixed-point problem
for nontrivial solutions (i.e. our formulation will eliminate the
trivial solution $u=0$).

We begin by introducing some notations. 
We consider the Banach space
$$
X = \{\, 
v\in C^0_{\rm{per}}(0,2\pi ), \
v(-\xi )=-v(\xi )
\, \} 
$$
endowed with the $L^\infty$ norm,
where $C^k_{\rm{per}}(0,2\pi )$ denotes the classical
Banach space of $2\pi$-periodic and $C^k$ functions 
$v\, : \mathbb{R}\rightarrow \mathbb{R}$.
Note that the map $N$ maps $X$ into itself.
We consider the closed subspace of $X$
$$
X_h = \{\, 
v\in X, \
\zeta^\ast (v)=0
\, \} ,
$$
where
$$
\zeta^\ast (u)=
\frac{1}{\pi}\int_{-\pi}^{\pi}{u(\xi )\, \sin{\xi}\, d\xi}. 
$$
Note that $\zeta^\ast (u)=
\frac{2}{\pi}\int_{0}^{\pi}{u(\xi )\, \sin{\xi}\, d\xi}$ for all $u \in X$.
We have $X=  \mathbb{R}\, \sin{\xi} \oplus X_h $ and note
$
P\, u = \zeta^\ast (u)\, \sin{\xi}, \ \ \ P_h =I-P
$
the corresponding projectors.

Now we split small amplitude solutions as in section \ref{gdnlsgen}, using
the ansatz (\ref{ansatz2}) and the simplified expression
(\ref{ansatzapprox}), in which we fix the phase $\phi = -\pi /2$
to recoved odd solutions. More precisely,
we set
\begin{equation}
\label{ansatzu}
u(\xi)=a\, \sin{\xi}+ a^\alpha v(\xi),
\end{equation}
where $a>0$ is assumed close to $0$
and $v\in D_h = X_h \, \cap\, C^2_{\rm{per}}(0,2\pi )$.
This splitting is reminiscent of the classical Lyapunov-Schmidt 
technique, since the closed operator at the left side of (\ref{ad2}) 
has a kernel spanned by $\sin{\xi}$ in $X$ and 
its restriction to $X_h$ is invertible.
The nonlinear dispersion relation (\ref{freq}) suggests the scaling
$\omega^2 -1 = a^{\alpha -1} \lambda$, where $\lambda \in \mathbb{R}$ 
and $v\in D_h$ will be subsequently determined as functions of $a$.
Inserting (\ref{ansatzu}) in (\ref{ad2}) and
projecting the resulting equation on the first Fourier harmonic
results in
\begin{equation}
\label{adcentral}
\lambda  = F_0 (v,a),
\end{equation}
\begin{equation}
\label{f0}
F_0 (v,a)= -a^{-\alpha}\, \zeta^\ast [\, N( a\, \sin{\xi}+ a^\alpha v ) \, ] .
\end{equation}
Now there remains to complete (\ref{adcentral}) by a fixed point equation for $v$.
For this purpose we
project  (\ref{ad2}) onto $X_h$, apply
$K_h = (\frac{d^2}{d\xi^2}+I)^{-1} \in \mathcal{L}(X_h , D_h)$
to both sides of (\ref{ad2})
and use the fact that $K_h \frac{d^2}{d\xi ^2} = \frac{d^2}{d\xi ^2}K_h=I-K_h$
on $D_h$. This yields
\begin{equation}
\label{adhyp}
v=F_h (v,\lambda , a),
\end{equation}
\begin{equation}
\label{fh}
F_h (v,\lambda , a)= 
a^{\alpha -1}\, \lambda (K_h -I ) v\, 
+a^{-\alpha}\, K_h P_h N( a\, \sin{\xi}+ a^\alpha v ) ,
\end{equation}
where $K_h$ is explicitly given by equation (\ref{kh})
(checking that $K_h$ maps $X_h$ into $D_h$ is lengthy but straightforward). 
As a conclusion, nontrivial solutions of (\ref{ad}) taking the form
(\ref{ansatzu}) are given by the fixed point equation
\begin{equation}
\label{fpe}
(v,\lambda )=F_a (v,\lambda ), 
\end{equation}
where the map
$F_a = (F_h , F_0)$ is defined by (\ref{f0})-(\ref{fh}) and depends
on the parameter $a>0$.
We shall consider (\ref{fpe}) as a fixed point equation in
some closed ball of $X_h \times  \mathbb{R}$.
Then any solution of 
(\ref{fpe}) will have the automatic smoothness $v \in D_h$ since
\begin{equation}
\label{smoothness}
\omega^2\, v = 
K_h \, \big(\, 
a^{\alpha -1}\, \lambda  v\, 
+a^{-\alpha}\, P_h N( a\, \sin{\xi}+ a^\alpha v )
\, \big)
\end{equation}
and $K_h \in \mathcal{L}(X_h , D_h)$.

\subsection{\label{sfp}Solution of the fixed point problem}

In what follows we prove some key properties of
$F_a$ required to apply the contraction mapping theorem.
We shall note $V=(v,\lambda ) \in X_h \times  \mathbb{R}$ and
equip $X_h \times  \mathbb{R}$ with the supremum norm
$\| V  \|_\infty ={\rm Max}(\| v \|_{L^\infty},|\lambda|)$. 
We denote by $B(\rho)$ the closed ball of center $0$ and
radius $\rho$ in $X_h \times  \mathbb{R}$. 

\begin{lemma}
\label{contrac}
Consider the map $F_a = (F_h (.,a) , F_0(.,a))$ defined by 
expressions (\ref{f0})-(\ref{fh}), where the nonlinear term $N$
takes the form (\ref{defn}). Assume that the potential $V$
defining $N$ satisfies the same conditions as
in theorem \ref{existthm}. Then
there exist $\rho >0$, $a_0 >0$ independent of $q$
such that $F_a$ maps $B(\rho)$ into itself 
for all $a\in (0,a_0 )$. Moreover, one has as $a \rightarrow 0^+$
$$
\bsup{-1}{V_1 , V_2 \in B(\rho), V_1 \neq V_2}
\frac{
\| F_a (V_1) - F_a (V_2)\|_{\infty}
}
{
\| V_1-V_2  \|_\infty
}
=O(a^{\alpha -1})
$$
uniformly in $q\in [-\pi , \pi]$.
\end{lemma}

\begin{proof}
We first prove that $F_a$ maps a suitable ball
$B(\rho)$ into itself. For notational simplicity 
we shall use the same symbol $C$ for different multiplicative
constants that are all independent of $\rho$, $q$ and $a$.
The assumptions made on $V$ in theorem \ref{existthm}
imply property (\ref{vhertzd}), from which it follows
that $\| V^\prime (u) \|_{L^\infty}=O(\| u \|_{L^\infty}^\alpha )$
as $\| u \|_{L^\infty} \rightarrow 0$, and thus
$\| N (u) \|_{L^\infty}=O(\| u \|_{L^\infty}^\alpha )$
uniformly in $q\in [-\pi , \pi]$.
Consequently, given $\rho \geq 1$ and $a_0 (\rho ) = \rho^{\frac{2}{1-\alpha}}$,
there exists $C>0$
such that for all $a\in (0, a_0 (\rho ) )$, for all $(v,\lambda )\in B(\rho)$,
for all $q\in [-\pi , \pi]$,
\begin{equation}
\| a^{-\alpha}\, N ( a\, \sin{\xi}+ a^\alpha v ) \|_{L^\infty}\leq C,
\end{equation}
\begin{equation}
\| a^{\alpha -1}\, \lambda (K_h -I ) v \|_{L^\infty}\leq C .
\end{equation}
This yields immediately the same type of estimate for $F_a$, namely
\begin{equation}
\| F_a (v,\lambda ) \|_{\infty}\leq C .
\end{equation}
Since $C$ is independent of $\rho$, $F_a$ maps
$B(\rho)$ into itself provided $\rho \geq \rm{Max}(C,1)$ and $a<a_0(\rho )$.

Now we fix a value of $\rho$ such that $F_a$ maps
$B(\rho)$ into itself for small enough $a$ and
estimate the Lipschitz constant of $F_a$.
In what follows
we use the same notation $C$ for different multiplicative
constants independent of $q$ and $a$.
The assumptions made on $V$ in theorem \ref{existthm}
and the mean value inequality imply
$$
| V^\prime (x)-V^\prime (y) | \leq C\, 
|x-y|\, \| (x,y) \|^{\alpha -1}_\infty 
$$
for all $(x,y) \in \mathbb{R}^2$ close to $0$. 
Consequently, one obtains successively for all functions
$u_1,u_2 \in L^\infty (\mathbb{R})$ close to $0$
$$
\| V^\prime (u_1)-V^\prime (u_2) \|_{L^\infty} \leq C\, 
\|u_1-u_2\|_{L^\infty}  \, \| (u_1,u_2) \|^{\alpha -1}_{(L^\infty )^2}, 
$$
$$
\| N (u_1)-N (u_2) \|_{L^\infty} \leq C\, 
\|u_1-u_2\|_{L^\infty}  \, \| (u_1,u_2) \|^{\alpha -1}_{(L^\infty )^2}
$$
uniformly in $q\in [-\pi , \pi]$.
Consequently, there exists $C>0$
such that for all $a$ small enough, for all $(v_i,\lambda_i )\in B(\rho)$,
\begin{equation}
\| a^{-\alpha}\, 
\big( \, N ( a\, \sin{\xi}+ a^\alpha v_1 )
-N ( a\, \sin{\xi}+ a^\alpha v_2 )
\, \big)
 \|_{L^\infty}\leq C\, a^{\alpha -1}\, \|v_1-v_2\|_{L^\infty},
\end{equation}
\begin{equation}
\| F_a (v_1,\lambda_1 ) - F_a (v_2,\lambda_2 )\|_{\infty}\leq C \, a^{\alpha -1}\,
\| (v_1,\lambda_1 )-(v_2,\lambda_2 ) \|_\infty .
\end{equation}

\qquad\end{proof}

As a corollary of lemma \ref{contrac}, one obtains the following result.

\ve

\begin{theorem}
\label{efp}
Fix $B(\rho)$ as in lemma \ref{contrac}. There exists
$a_0 >0$ independent of $q$ such that 
$F_a$ admits a unique fixed point $V(a) = (v_a , \lambda_a)$
in $B(\rho)$ for all $a\in (0,a_0 )$.
Moreover, there exist $M,C>0$ 
independent of $q$
such that the sequence 
$\{ V_n (a)\}_{n \geq 0}$ in $B(\rho)$ defined by
$V_n (a)= F_a^{n+1}(0)$ satisfies
\begin{equation}
\label{approxseq}
\| V_n (a) - V(a) \|_\infty 
\leq M\, 
(C a^{\alpha -1})^{n+1} 
\mbox{ for all }n\geq 0.
\end{equation}
\end{theorem}

\begin{proof}
According to the estimate of lemma \ref{contrac}, $F_a$ defines
a contraction on $B(\rho)$ for all $a$ small enough, with
Lipschitz constant $\kappa \leq C\, a^{\alpha -1}$,
where $C$ is independent of $q$. In that case, by
the contraction mapping theorem, $F_a$
admits a unique fixed point $V(a)$ in $B(\rho)$ which is exponentially
attracting. Moreover, the classical estimate of the convergence speed yields
$$
\| V_n (a) - V(a) \|_\infty \leq \frac{\kappa^{n+1}}{1-\kappa}\,
\| V_0 (a)  \|_\infty ,
$$
which implies (\ref{approxseq}) for $a\approx 0$.
 
\qquad\end{proof}

Theorem \ref{efp} implies the existence of
nontrivial solutions $(u_a , \omega_a^2 )$
of (\ref{ad}) taking the form
\begin{equation}
\label{ansatzubis}
u_a(\xi)=a\, \sin{\xi}+ a^\alpha v_a(\xi),
\end{equation}
\begin{equation}
\label{oma2}
\omega_a^2 = 1+a^{\alpha -1} \lambda_a.
\end{equation}
This in turn implies the existence of periodic travelling wave solutions of (\ref{nc})
taking the form (\ref{xleadingorder}).
In order to complete the proof of theorem \ref{existthm},
we now need to determine the principal parts of 
$v_a$ and $\lambda_a$ for $a\approx 0$.
Considering $F_a (0)=(v_0 (a), \lambda_0 (a))$, we have
by estimate (\ref{approxseq})
\begin{equation}
\label{estim1}
\| v_0 (a) - v_a \|_{L^\infty} + |\lambda_0 (a) - \lambda_a|=  O(a^{\alpha -1}),
\end{equation}
hence we are led to compute $v_0(a)$ and $\lambda_0 (a)$.
One finds
\begin{equation}
\label{cl0}
\lambda_0 (a) = -a^{-\alpha}\, \zeta^\ast [\, N( a\, \sin{\xi} ) \, ] ,
\end{equation}
\begin{equation}
\label{fv0}
v_0 (a)=  
a^{-\alpha}\, K_h P_h N( a\, \sin{\xi} ) ,
\end{equation}
where $K_h$ is defined by (\ref{kh}). 
Keeping in mind expansion (\ref{vhertzd}) of $V^\prime $,
the principal part of the nonlinear map $N$ reads
\begin{equation}
\label{defn0}
[N_0(u)](\xi)=V_0^\prime (u(\xi +q)-u(\xi ))-V_0^\prime (u(\xi )-u(\xi -q)),
\end{equation}
where $V_0^\prime (x)=- |x|^{\alpha}\, H(-x)$.
For all $a>0$ one has $V_0^\prime (a x) = a^\alpha V_0^\prime (x) $,
and in the same way $N_0 (a\, u)=a^\alpha N_0 (u)$.
Using this property and expansion (\ref{vhertzd}) in equations (\ref{cl0})-(\ref{fv0}),
one obtains consequently as $a\rightarrow 0$
\begin{equation}
\label{l0}
\lambda_0 (a) = - \zeta^\ast [\, N_0(  \sin{\xi} ) \, ] +o(1),
\end{equation}
\begin{equation}
\label{v0}
v_0 (a)=   K_h P_h N_0(  \sin{\xi} ) +o(1)
\mbox{ in } D_h.
\end{equation}
Estimate (\ref{estim1}) and expansion (\ref{v0}) yield
expansion (\ref{va}) of theorem \ref{existthm}.
Now there remains to recover expansion (\ref{wleadingorder}), which
requires to compute leading order terms of (\ref{l0}).
We can restrict our computations 
to the case $q \in [0, \pi ]$, which allows one to 
recover the results for $q \in [-\pi ,0]$
by symmetry considerations (see section \ref{symprop}).
After lengthy but straightforward computations
based on classical trigonometric identities, one obtains
for all $q\in [0,\pi ]$
\begin{eqnarray*}
&&\zeta^\ast [\, V_0^\prime (\sin{(\xi +q)}-\sin{\xi}) \, ] \\
 &=&
-\frac{2}{\pi}\, [2 \sin{(\frac{q}{2})}]^\alpha  
\, \int_{\frac{\pi}{2}-\frac{q}{2}}^{\pi }{|\cos{(\xi + \frac{q}{2})}|^{\alpha}\, \sin{\xi}\, d\xi}
 \\
 &=&
-\frac{2}{\pi}\, [2 \sin{(\frac{q}{2})}]^\alpha 
\big(\, 
\sin{(\frac{q}{2})}\, \int_{\frac{\pi}{2}}^{\pi + \frac{q}{2}}{|\cos{s}|^{\alpha +1}\, ds}
+\frac{1}{\alpha +1}[\cos{(\frac{q}{2})}]^{\alpha +2}
\, \big) ,
\end{eqnarray*}
and in the same way
\begin{eqnarray*}
&&
\zeta^\ast [\, V_0^\prime ( \sin{\xi}-\sin{(\xi -q})) \, ] \\
&=&
\frac{2}{\pi}\, [2 \sin{(\frac{q}{2})}]^\alpha 
\big(\, 
\sin{(\frac{q}{2})}\, \int_{\frac{q}{2} }^{\frac{\pi}{2}}{(\cos{s})^{\alpha +1}\, ds}
-\frac{1}{\alpha +1}[\cos{(\frac{q}{2})}]^{\alpha +2}
\, \big) .
\end{eqnarray*}
Combining (\ref{defn0})-(\ref{l0}) with the above identities, one finds
after elementary computations
\begin{equation}
\label{l0bis}
\lambda_0 (a) =   [2 \sin{(\frac{q}{2})}]^{\alpha +1} c_\alpha +o(1),
\end{equation}
where $c_\alpha$ is defined in equation (\ref{wallis}).
Combining expression (\ref{oma2}) with estimate (\ref{estim1})
and expansion (\ref{l0bis}) yields finally for $q\in [0,\pi ]$
\begin{equation}
\label{oma2bis}
\omega_a^2 = 1+a^{\alpha -1}  [2 \sin{(\frac{q}{2})}]^{\alpha +1} c_\alpha + o(a^{\alpha -1} ) .
\end{equation}
Expressing $c_\alpha$ using Euler's Gamma function as in
(\ref{thec}), one finally obtains expansion (\ref{wleadingorder}) of theorem
\ref{existthm} for $q\in [0,\pi ]$.

\ve

\begin{remark}
As shown by estimate (\ref{approxseq}), computing $V_n (a)$ for
$n$ large enough allows one to obtain approximations of the fixed point $(v_a , \lambda_a)$
at any order in $a \approx 0$, which provides 
approximations of the wave profile (\ref{xleadingorder})
and wave frequency at arbitrary order. 
However the expression of $V_n (a)$ becomes
highly complex already at $n=2$, hence this series of explicit
approximations doesn't seem useful in practice. 
Nevertheless, it is interesting to see that approximations at
arbitrary order in $a$ can be obtained despite the limited smoothness
of $V$ at the origin.
\end{remark}

\ve

\begin{remark}
It is again interesting at this stage to reinterpret the results in the framework of
Lyapunov-Schmidt reduction. 
Adding an arbitrary
phase $\phi$ to solution (\ref{xleadingorder}) and introducing $A=a\, e^{i \phi}$,
the nonlinear dispersion relation (\ref{wleadingorder}) describes
nontrivial solutions of the bifurcation equation
\begin{equation}
\label{eqbifurc}
A\, (\omega -1 ) - \frac{2}{\tau_0}\, A\, |A|^{\alpha -1}\, |\sin{(\frac{q}{2})}|^{\alpha +1}+\rm{h.o.t.}=0 .
\end{equation}
More precisely, solutions (\ref{xleadingorder})-(\ref{wleadingorder}) and their
phase shifts appear through a pitchfork bifurcation with $SO(2)$ symmetry
which occurs at $(A,\omega )=(0,1)$. Note that the classical 
pitchfork bifurcation equation corresponds to the case
$\alpha =3$ of (\ref{eqbifurc}), and for $\alpha >1$ it can be recovered from (\ref{eqbifurc})
after the $C^0$ change of variable $Z=|A|^{\frac{\alpha -3}{2}}\, A$.
\end{remark}

\subsection{\label{symprop}Symmetries}

In this section we point out some symmetry properties of the periodic travelling waves.
We begin by examining more closely their dependency on the wavenumber $q$,
hence we will denote the map (\ref{defn}) by $N_q$, the map $F_a$ by $F_{a,q}$
and the fixed point of $F_{a,q}$ in $B(\rho)$ by $V(a,q) = (v_{a,q} , \lambda_{a,q})$.
Let us introduce the operator $[\tau_\pi u](\xi)=u(\xi +\pi)$ corresponding to  
half-period space-shift and the symmetry 
$S \in \mathcal{L}(X)$ defined by
$S=-\tau_\pi$. One can check that
\begin{equation}
\label{symn}
S\, N_q ( a\, \sin{\xi}+ a^\alpha v )=N_{-q} ( a\, \sin{\xi}+ a^\alpha S\, v ).
\end{equation}
Consequently, from the identity $(v_{a,q} , \lambda_{a,q})=F_{a,q}(v_{a,q} , \lambda_{a,q})$
and definition (\ref{adcentral})-(\ref{adhyp}) of $F_a$
we obtain
$$
(S\, v_{a,q}, \lambda_{a,q} )= F_{a,-q}(S\, v_{a,q}, \lambda_{a,q} )
$$
since $S$ commutes with $K_h$ and $P\, S=S\, P=P$.
By the uniqueness of the fixed point of $F_{a,-q}$ in $B(\rho)$, it
follows
\begin{equation}
\label{symq}
v_{a,-q}=S\, v_{a,q}, \ \ \ \lambda_{a,-q}=\lambda_{a,q}.
\end{equation}
In particular, expression (\ref{oma2bis}) becomes for all  $q\in [-\pi,\pi ]$
\begin{equation}
\label{oma2bisbis}
\omega_a^2 = 1+a^{\alpha -1}    |2\, \sin{(\frac{q}{2})}|^{\alpha +1} c_\alpha + o(a^{\alpha -1} ) ,
\end{equation}
hence one finally recovers expansion (\ref{wleadingorder}) of theorem
\ref{existthm} for all $q\in [-\pi,\pi ]$.

In what follows we concentrate on the particular case $q=\pi$. 
Since functions belonging to $X$ are $2\pi$-periodic, the operators $N_q$ and
$F_{a,q}$ (acting respectively on $X$ and $X_h \times \mathbb{R}$)
are $2\pi$-periodic in $q$, and the same holds true for 
$v_{a,q}, \lambda_{a,q}$. Consequently, property (\ref{symq}) yields
$v_{a,\pi}=S\, v_{a,\pi}$, i.e. for all $\xi \in \mathbb{R}$ one has 
\begin{equation}
\label{symqpi}
v_{a,\pi} (\xi + \pi)=-v_{a,\pi} (\xi).
\end{equation}
It follows that for $q=\pi$, the solution (\ref{xleadingorder}) takes the form 
\begin{equation}
\label{devbin}
x_n (t)= a\, (-1)^{n+1}\, (\sin{( \omega_a t)}+a^{\alpha -1}  v_a (\omega_a t)),
\end{equation}
i.e. it corresponds to a binary oscillation
for which nearest-neighbours oscillate out of phase with opposite displacements.

Note that there exists another simple
method to obtain binary oscillations in the special case $V=V_0$, without  restriction to the small amplitude limit.
One can find solutions
of (\ref{nc}) of the form
$
x_n (t)= (-1)^n\, a(t)
$
after a short discussion with respect to the sign of $a(t)$ and parity of $n$.
This expression defines a solution of (\ref{nc}) if and only if $a$ satisfies the integrable equation 
\begin{equation}
\label{inteq}
\frac{d^2 a}{dt^2}+a+2^\alpha a \, |a|^{\alpha -1}=0,
\end{equation}
the phase space of which is filled by periodic orbits.

\section{\label{numeric}Numerical comparison of the DpS model and Newton's cradle}

In this section we numerically analyze how the dynamics of the DpS equation
approximates the one of the nonlinear chain (\ref{nc}), 
in the case of classical Hertz's contact law ($V=V_0$, $\alpha = 3/2$) and for
some classes of initial conditions. 
We start our investigations with initial conditions leading to periodic travelling waves.
In section \ref{fixedpoint} we have proved the existence of exact periodic travelling waves
of (\ref{nc}) whose profile and velocity are determined at leading order by the DpS
equation in the small amplitude limit. However these results do not indicate the range
of validity of the asymptotic analysis and do not provide informations on the local
dynamics around the travelling waves. We examine such questions in section
\ref{stable} for stable travelling waves and in section \ref{unstable} for unstable ones.
As we shall see, system (\ref{nc}) and the DpS equation
display similar features concerning 
the occurence of modulational instability, which yields the formation of discrete breathers.
A more precise study of these localized structures requires 
well chosen localized initial conditions. In section \ref{loc}, we use
initial conditions deduced from the study of localized standing waves
of the DpS equation (section \ref{twsol}) in order to generate 
static and travelling breathers.

\ve

Let us now describe our general numerical procedure
before going further.
We consider a given 
initial condition $\{A_n (0) \}_{n \in \mathbb{Z}}$ and
use ansatz (\ref{ansatz}) to determine an initial
condition for (\ref{nc})
(computing $\dot{x}_n (0)$ requires the knowledge of $\partial_\tau A_n (0)$
which is determined by the case $\tau=0$ of (\ref{dfnls})). 
Then we integrate (\ref{nc}) and (\ref{dfnls}) numerically
(or use an explicit solution of (\ref{dfnls}) when available)
and compare the solution of (\ref{nc}) with the ansatz (\ref{ansatz})
at subsequent times. In what follows the corresponding solution of (\ref{nc})
will be referred to as the {\em exact solution}, and the 
{\em approximate solution} will denote the ansatz (\ref{ansatz}) computed with the
DpS equation.

We realize different
tests for chains ranging from $50$ or $200$ beads with periodic boundary conditions. 
Numerical integrations are performed using the standard ode solver of the
software package Scilab. The numerical integrator doesn't exactly conserve
the Hamiltonian (\ref{ham}), but in all simulations we have checked that 
the relative errors $|{\mathcal H }(t)- {\mathcal H }(0)|/{\mathcal H}(0)$
are less than $3,\! 5.10^{-4}$, this upper bound corresponding to simulations over
long times of the order of
$540$ periods of the linear oscillators (section \ref{loc}). 

\subsection{\label{stable}Approximation of stable travelling waves}

We begin by computing periodic travelling waves in the small amplitude regime 
in which the weakly nonlinear analysis of section \ref{gdnls} has been performed.
We normalize the travelling wave solution (\ref{ansatz3}) by fixing $R=1$,
so that $a=2 \epsilon$ in ansatz (\ref{ansatzapprox}).
For the first numerical run, we choose
(\ref{ansatzapprox}) as an initial condition,
with $\phi=0$, $q=\pi /5$ and $\omega_{\rm{tw}}=1.01$, which yields
$a\approx 2.10^{-2}$ in the nonlinear dispersion relation (\ref{freq}). These parameter values
correspond to a weakly-nonlinear regime, since
the maximal value of the nonlinear interaction forces $| V_0^\prime (x_{n+1}^{\rm{tw}} - x_n^{\rm{tw}} )  |$
estimated with (\ref{ansatzapprox}) is close to $a^{3/2}/2$, which corresponds 
to seven percent of the maximal linear restoring force $a$.
As shown by figure \ref{plot4-5} (left plot),
the agreement between the exact and approximate solutions
is excellent, even over long times corresponding to
$100$ multiples of the wave period $T_{\rm{tw}}=2\pi /\omega_{\rm{tw}}$.
More precisely, the relative error between exact and approximate solutions
(measured with the supremum norm) is less than five percent at the final
time of computation (figure \ref{plot4-5}, right plot).
These results show that the wave profile remains almost sinusoidal.
However, nonlinear effects cannot be neglected 
on the timescales we have considered, because they still affect the wave velocities. 
This appears clearly in the right plot of figure \ref{plotcomplin}, which
compares the exact and approximate solutions to the linear wave with
same amplitude $a$ and frequency $\omega =1 < \omega_{\rm{tw}}$.
This wave has a lower velocity than the exact solution, so that it 
becomes approximately out of phase at time $t_0=333.7$.
Nonlinear selection of the wave velocity is correctly captured by the DpS equation, 
since the graphs of the exact and approximate solutions
are almost superposed.

In a second numerical run,
we modify the previous initial condition by setting 
\begin{equation}
\label{perturb}
A_n (0 )=
(1+\rho_n^{(1)} )\, \cos{(q n)}
-i\, (1+\rho_n^{(2)} )\, \sin{(q n)},
\end{equation}
where $\rho_n^{(1)}, \rho_n^{(2)} \in [-1/2 , 1/2]$ are
uniformly distributed random variables. 
The mismatch between exact and approximate solutions
increases significantly compared to figure \ref{plot4-5} 
(see figure \ref{plotcomplin}, left plot), but the approximation
remains very satisfactory given the large time interval 
considered. A better agreement between the exact and approximate solutions
could probably be achieved by taking into account the higher order correction $R_n$
to (\ref{ansatz}) provided by (\ref{calculapp}).
Besides the comparison between exact and approximate
solutions, this numerical experiment reveals the
stability of the corresponding travelling waves 
of equations (\ref{nc}) and (\ref{dfnls}), at least
over $100$ periods of the linear oscillators
(however instabilities might occur on longer timescales).
As we shall see later, fast modulational instabilities can show up 
for other choices of wavenumbers $q$.

\begin{figure}[!h]
\psfrag{n}[0.9]{\huge $n$}
\psfrag{x}[1][Bl]{\huge $x_n(t)$}
\psfrag{t}[0.9]{\huge $t$}
\psfrag{er}[1][Bl]{\huge $E(t)$}
\begin{center}
\includegraphics[angle=-90,scale=0.22]{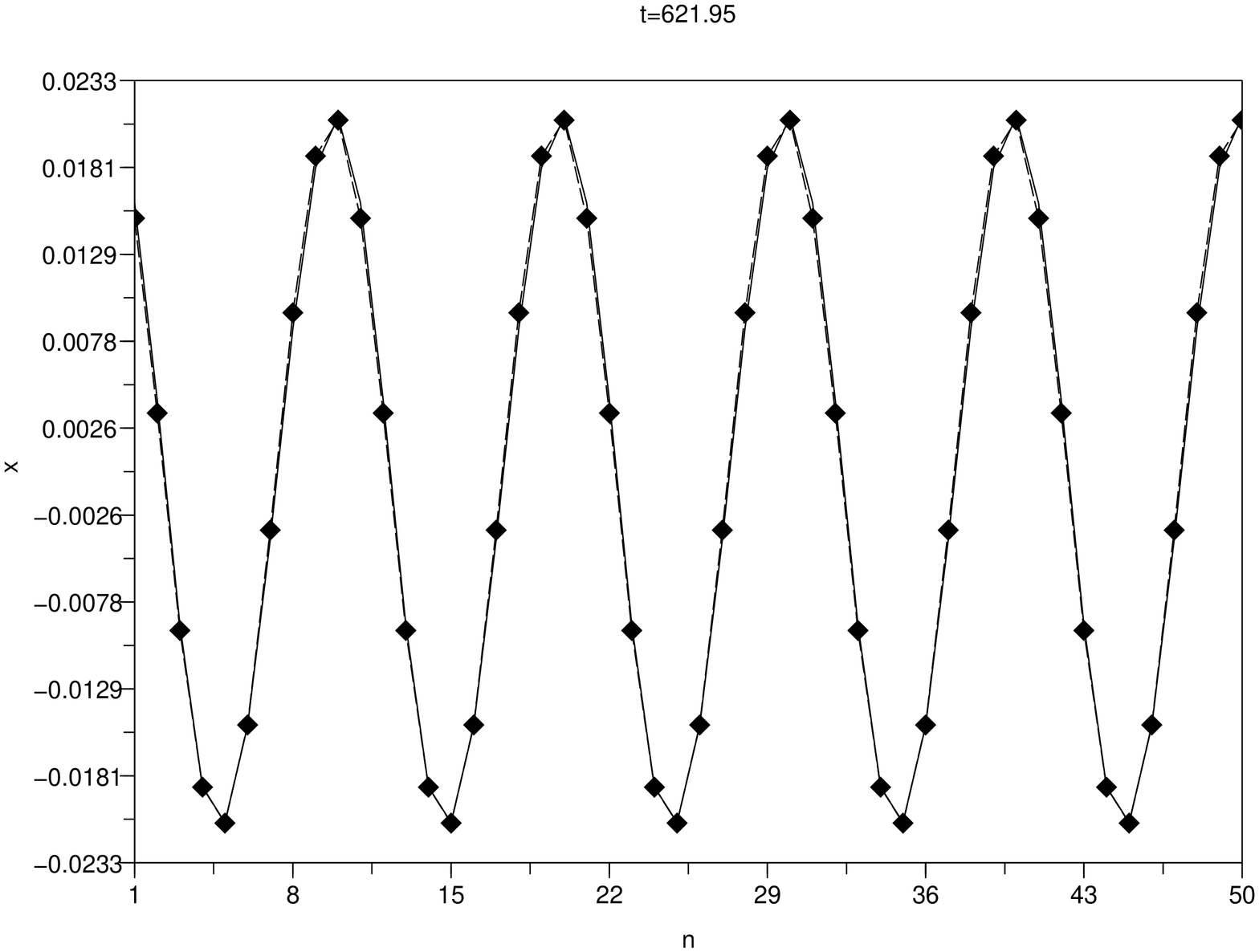}
\includegraphics[angle=-90,scale=0.22]{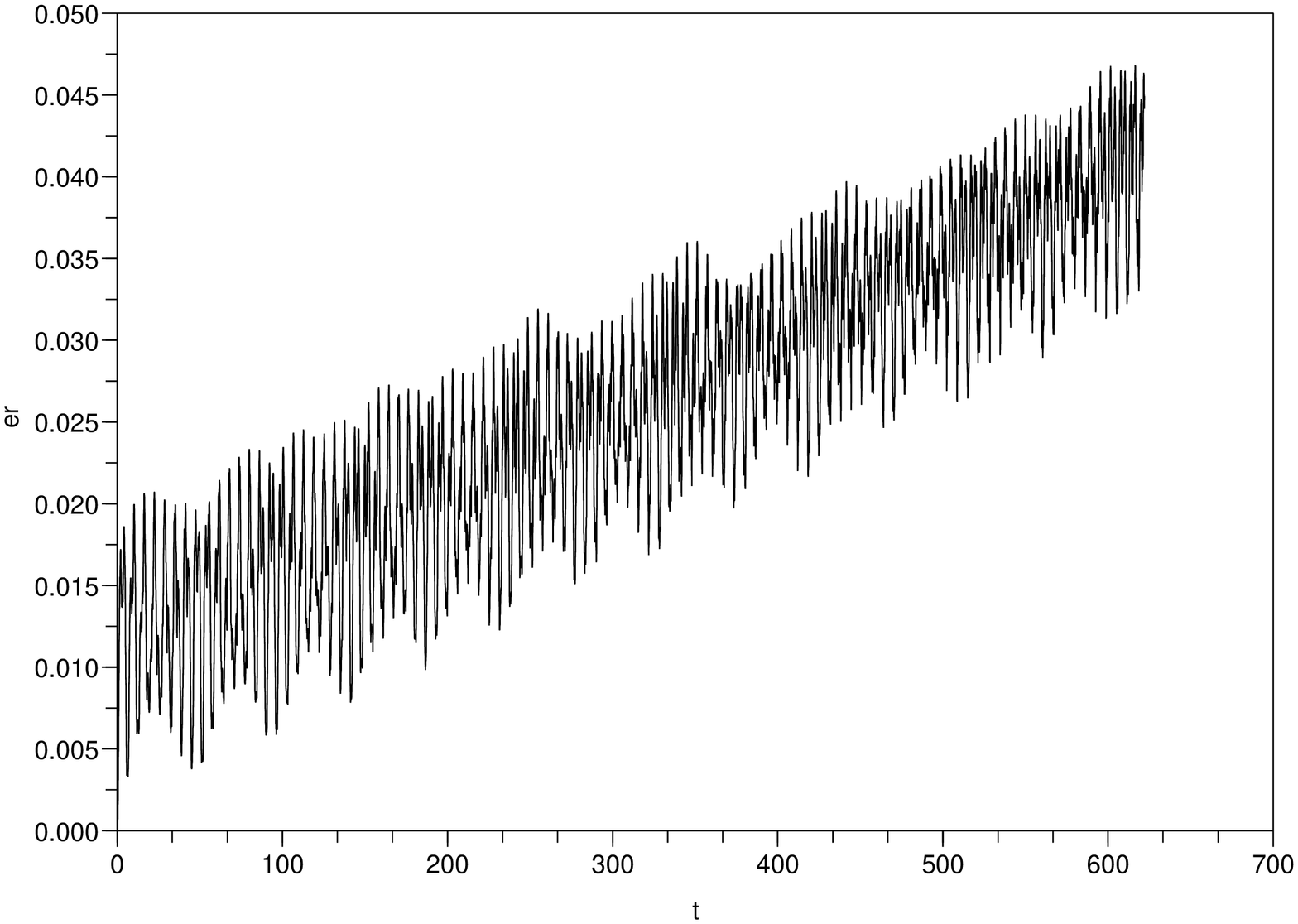}
\end{center} 
\caption{\label{plot4-5} 
Left plot~: small amplitude exact solution (continuous line) and approximate solution (dash-dot line)
for an initial condition corresponding to ansatz (\ref{ansatzapprox}),
with $\phi=0$, $q=\pi /5$ and $\omega_{\rm{tw}}=1.01$. Bead displacements are plotted
at a time $t=621.95$ of the order of $100$ periods of the linear oscillators.
The agreement is excellent since both graphs are almost indistinguishable.
Right plot~: relative error $E(t)=a^{-1}\, \| x_n^{\rm{tw}} (t) - x_n(t) \|_{\infty}$ between
exact and approximate solutions.}
\end{figure}

\begin{figure}[!h]
\psfrag{n}[0.9]{\huge $n$}
\psfrag{x}[1][Bl]{\huge $x_n(t)$}
\begin{center}
\includegraphics[angle=-90,scale=0.22]{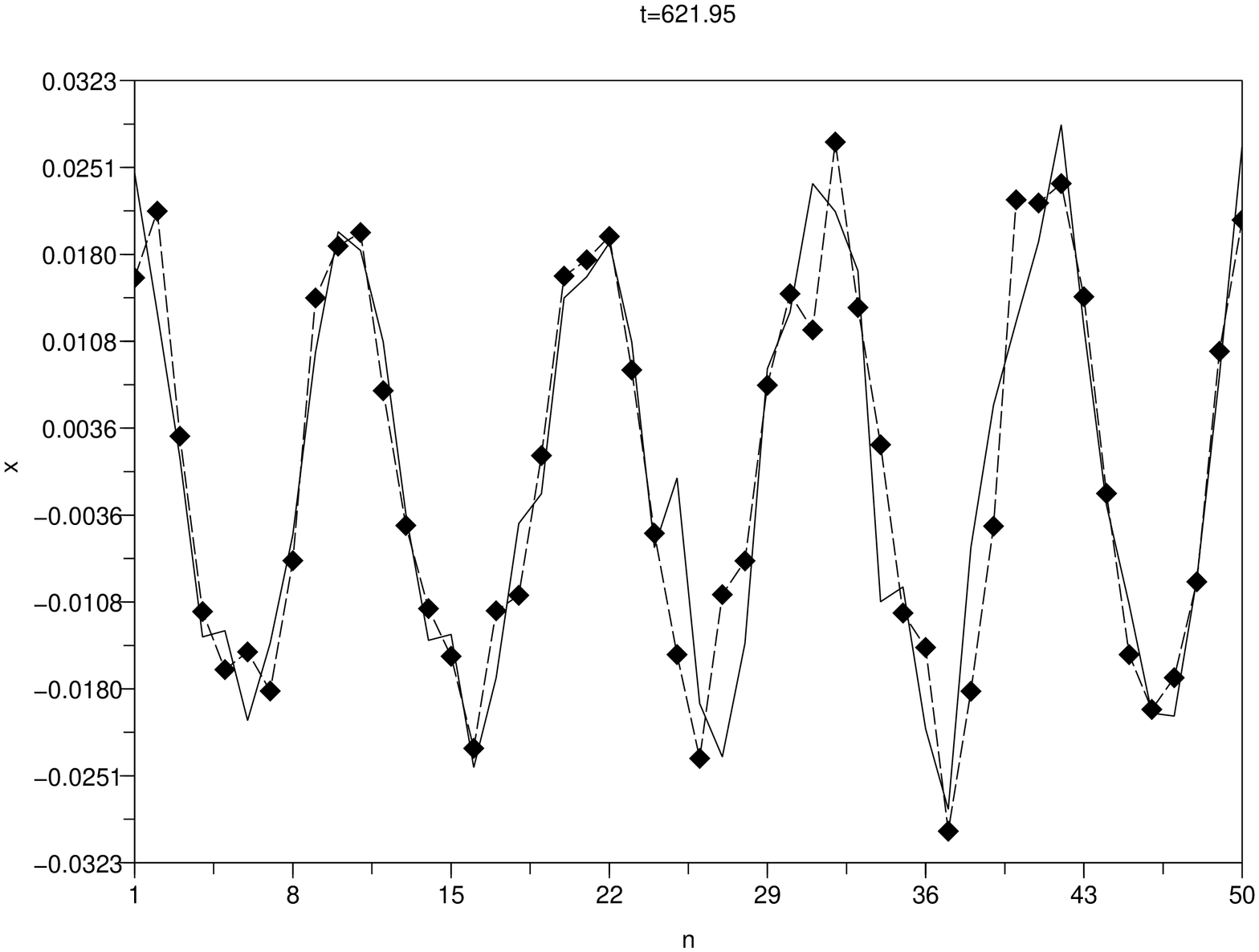}
\includegraphics[angle=-90,scale=0.22]{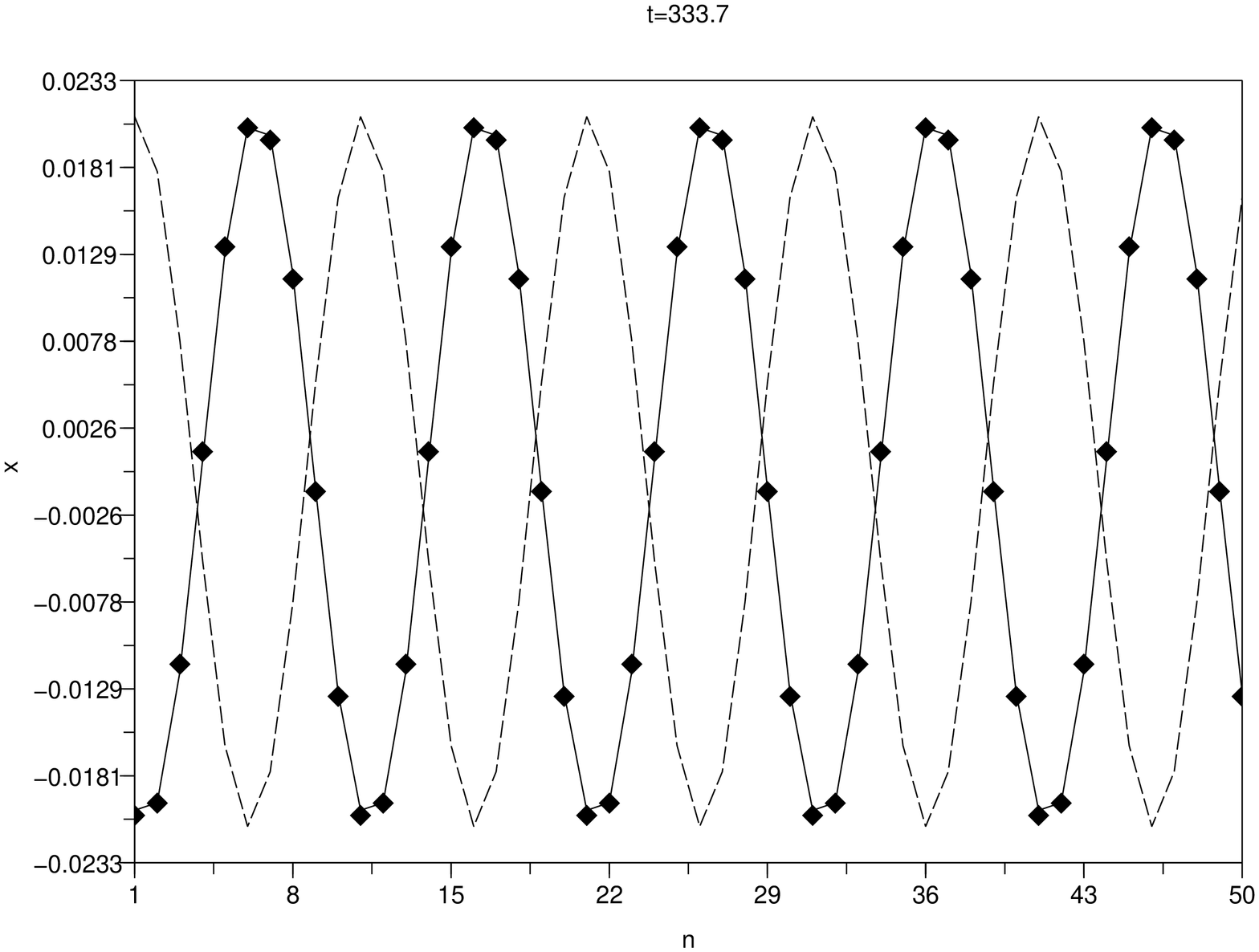}
\end{center} 
\caption{\label{plotcomplin}
In the left plot, the
same experiment as in figure \ref{plot4-5} is performed
with a random noise added to the initial
condition (equation (\ref{perturb})). 
The right plot compares the exact and approximate solutions to the linear wave 
$x_n (t)=a\, \cos{(q\, n - t)}$ solution of equation (\ref{nc}) linearized at $x_n =0$. 
Comparison is made at time $t_0=333.7$.
This wave has a nearly identical profile (dashed line) but a lower velocity, so that it 
becomes out of phase with the exact solution at $t=t_0$. On the contrary,
the approximate solution deduced from the nonlinear DpS equation captures the correct wave velocity
with excellent precision, and the graphs of the exact and approximate solutions
(continuous and dash-dot lines) are almost superposed.}
\end{figure}

Next we again choose 
ansatz (\ref{ansatzapprox}) with $q=\pi /5$ as an initial condition, but increase the frequency up to
$\omega_{\rm{tw}}=1.1$, which yields $a\approx 2.11$ in (\ref{freq}). In that case,
the maximal values of the nonlinear interaction force and the linear restoring force 
estimated with (\ref{ansatzapprox}) are of the same order, i.e. this initial
condition should correspond to a fully nonlinear regime. 
As shown by figure \ref{plot6-7}, the initial condition still generates a periodic travelling wave
that is qualitatively well described by the DpS equation. 
Quite surprisingly, the waveform remains nearly sinusoidal and close to the ansatz (\ref{ansatzapprox})
(see figure \ref{plot2-3}). However the ansatz doesn't accurately describe the wave velocity,
as revealed by the shift between the exact and approximate solutions.

\begin{figure}[!h]
\begin{center}
\includegraphics[scale=0.30]{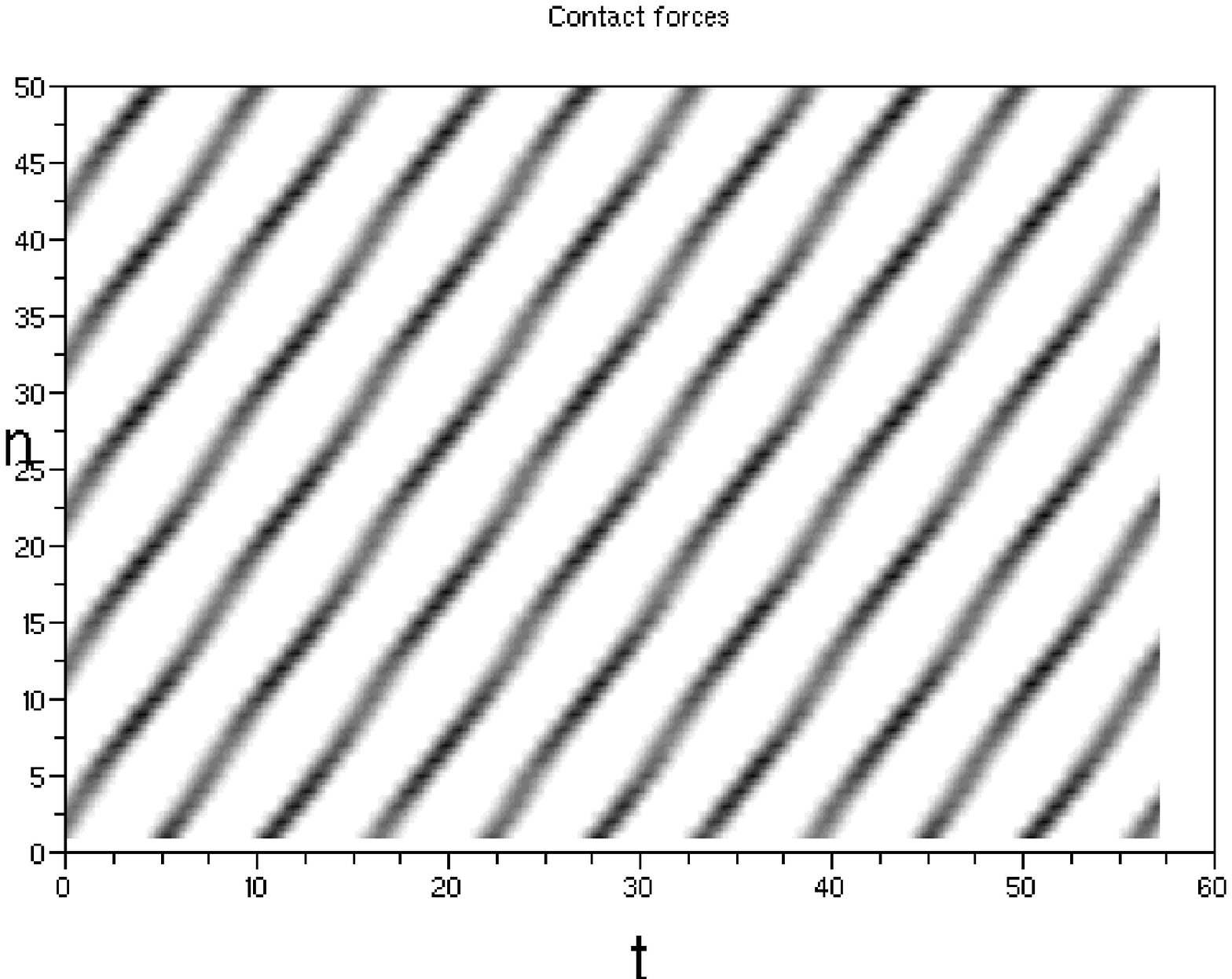}
\includegraphics[scale=0.30]{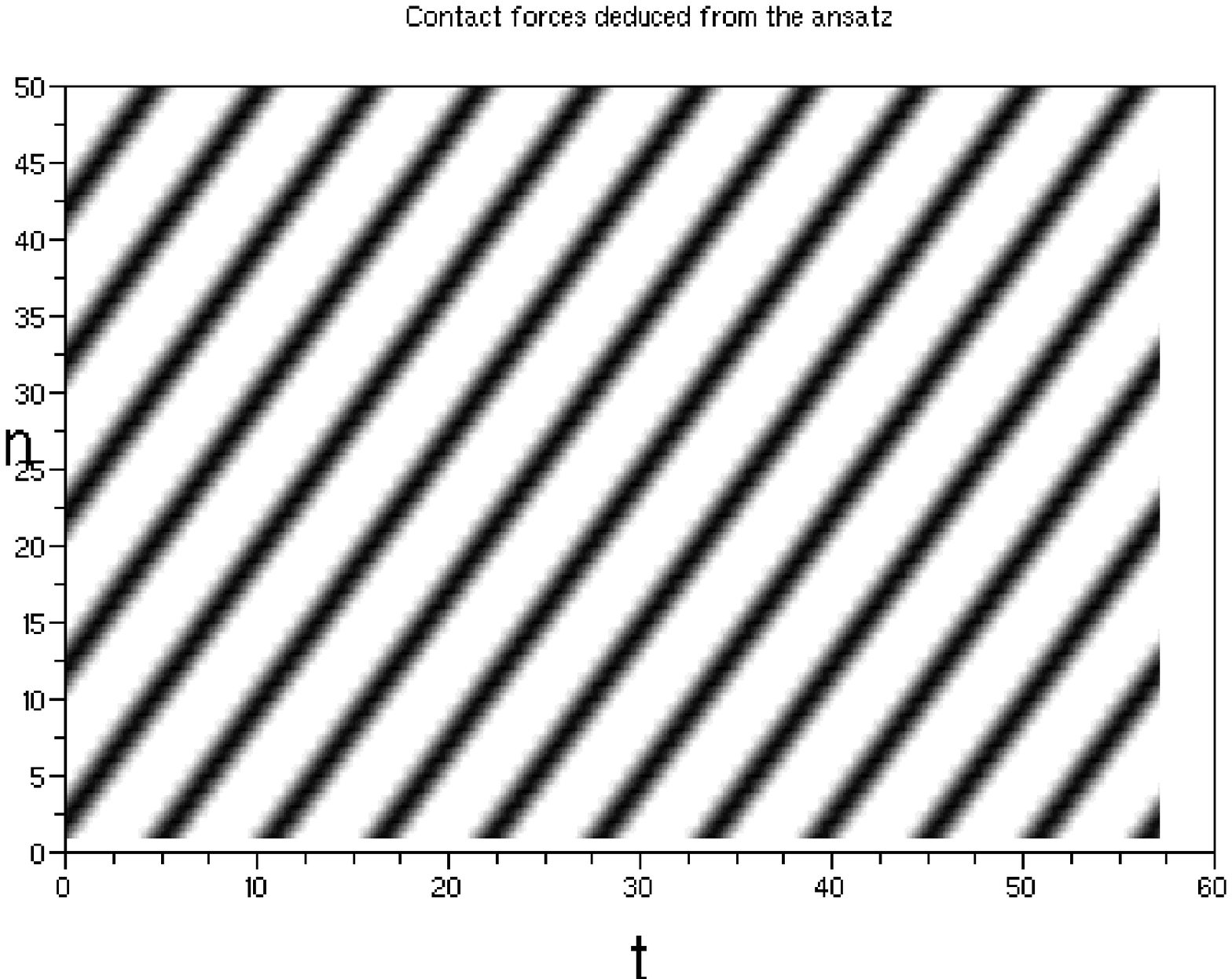}
\end{center} 
\caption{\label{plot6-7} 
Space-time diagram showing the interaction forces $V_0^\prime (x_{n+1}-x_n)$, for
an initial condition given by ansatz (\ref{ansatzapprox}) with $q=\pi /5$ and
$\omega_{\rm{tw}}=1.1$. Forces are represented 
in grey levels, white corresponding to vanishing interactions (i.e. beads not in contact)
and black to a minimal negative value of the contact force. The left plot corresponds to
the exact solution, taking the form of a periodic travelling wave with  almost constant velocity.
The right plot shows the approximate solution.}
\end{figure}

\begin{figure}[!h]
\psfrag{n}[0.9]{\huge $n$}
\psfrag{x}[1][Bl]{\huge $x_n(t)$}
\begin{center}
\includegraphics[angle=-90,scale=0.22]{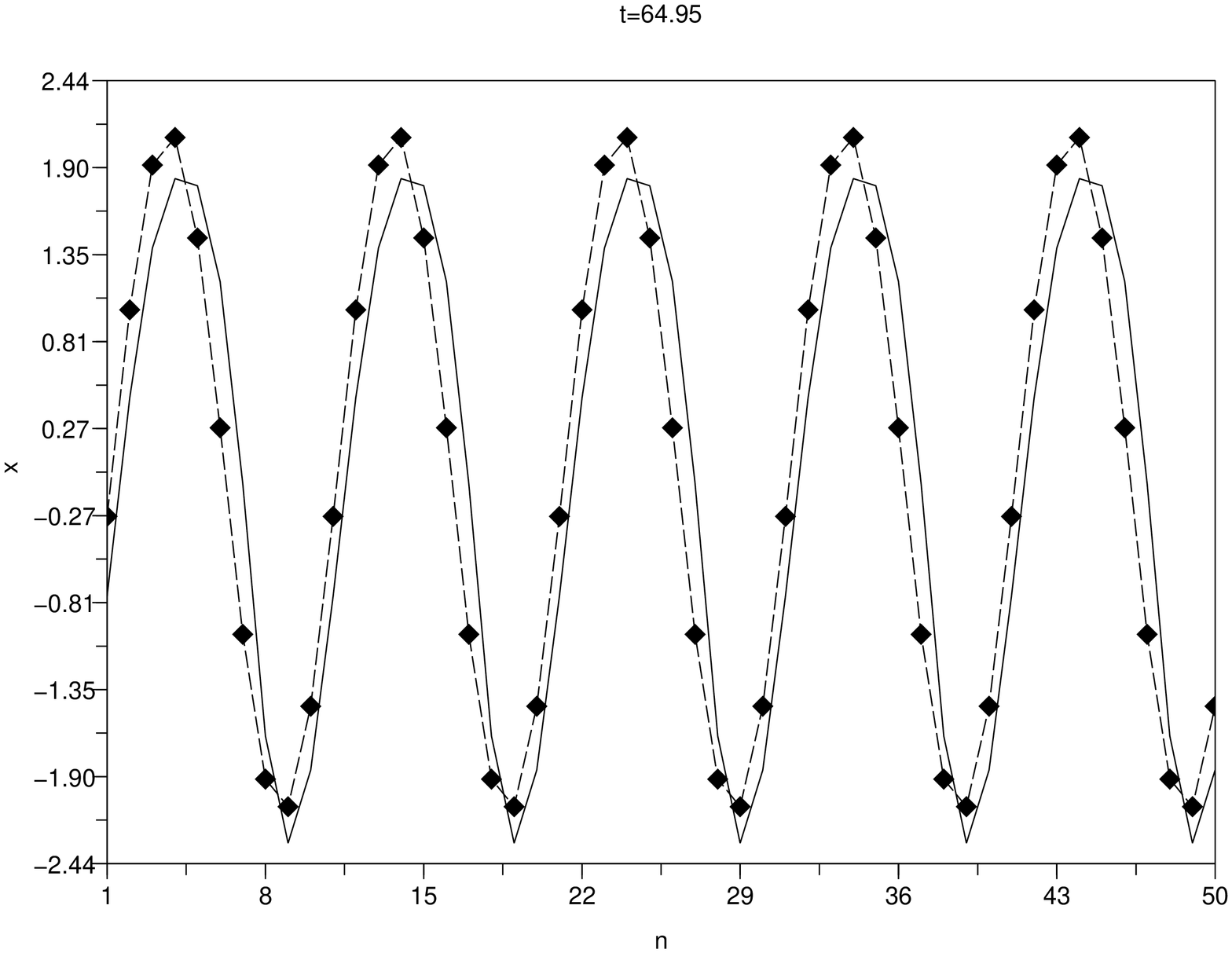}
\includegraphics[angle=-90,scale=0.22]{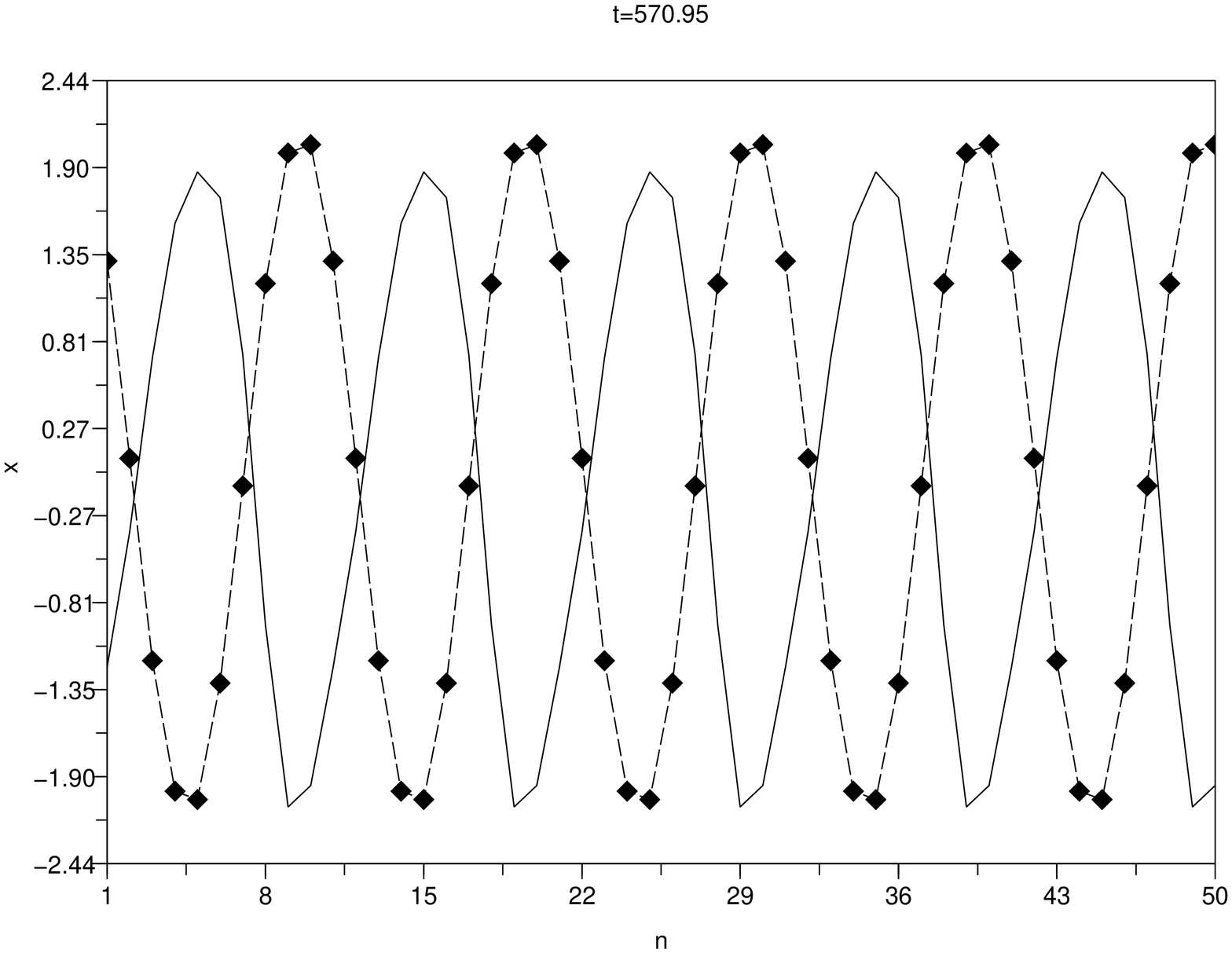}
\end{center} 
\caption{\label{plot2-3} 
Wave profiles of the exact (continuous line) and approximate (dash-dot line) solutions
in the numerical experiment of figure \ref{plot6-7}, shown at two different times
$t=64.95$ (left) and $t=570.95$ (right).
For the exact solution, the travelling wave motion is 
coupled to small collective in-phase oscillations, which explains the vertical shift of
the profile. The wave velocity is less accurately described by the DpS equation
in this regime, so that a small shift of the two profiles is already present at 
$t=64.95$, and the exact and approximate solutions are out of phase at $t=570.95$.}
\end{figure}

In the next computation we 
keep the same parameter values except we increase the frequency up to 
$\omega_{\rm{tw}}=1.5$, so that $a\approx 52.9$. This 
correspond to a highly nonlinear regime, where one can expect
that the DpS equation will not properly describe the dynamics
of (\ref{nc}). The two plots of figure \ref{plot1v5st} reveal indeed
notable differences between 
the travelling wave patterns of the exact and approximate solutions.
The initial condition generates in fact a pulsating travelling
wave, showing internal oscillations that appear as spots
along the wavefronts of figures \ref{plot1v5st} and \ref{plot1v5}.
This internal breathing appears at the beginning of the simulation
(around $t=1$) and releases some energy in the form of small
amplitude travelling waves with lower velocity. These waves
subsequently collide with the main travelling wave, generating
travelling waves with negative velocities. These phenomena
are absent in the DpS equation as shown by the right plot of
figure \ref{plot1v5st}. In addition the wave profile is not any more sinusoidal
(see figure \ref{plot1v5}, right plot).

\begin{figure}[!h]
\begin{center}
\includegraphics[scale=0.30]{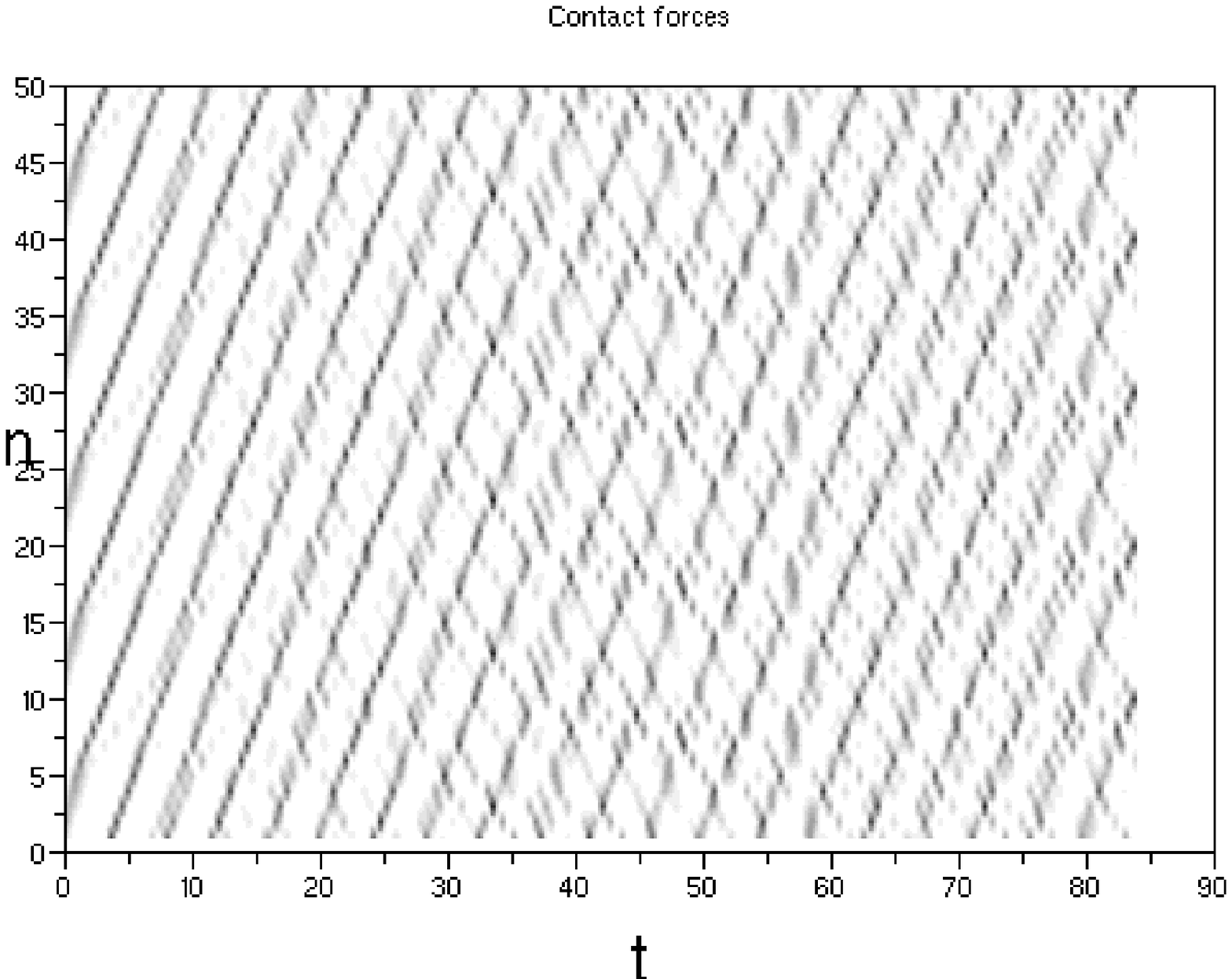}
\includegraphics[scale=0.30]{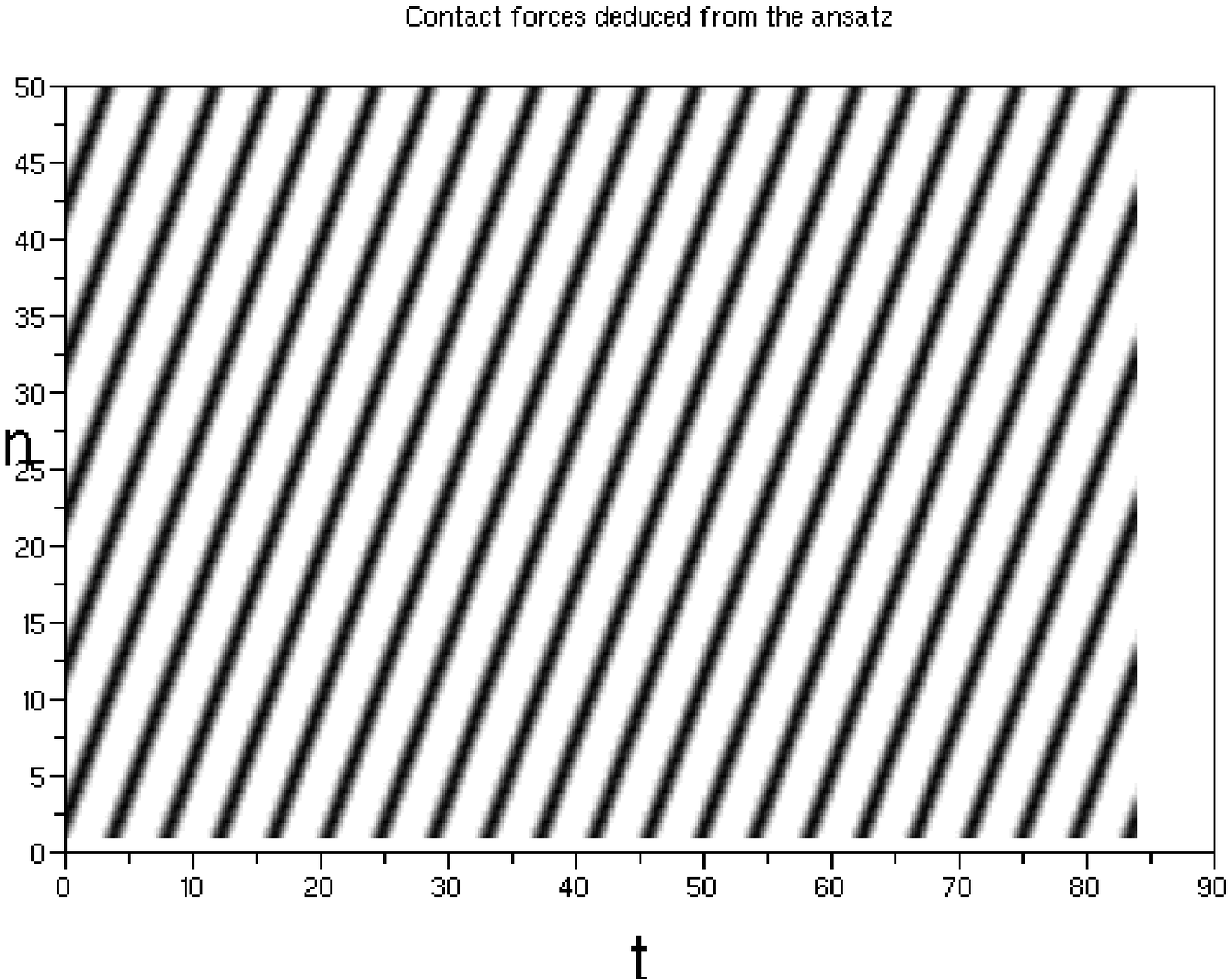}
\end{center} 
\caption{\label{plot1v5st} 
Same computation as in figure \ref{plot6-7}, but this time
in a highly nonlinear regime where the
frequency parameter of the initial ansatz is increased to 
$\omega_{\rm{tw}} = 1.5$. In that case the travelling wave pattern
of the exact solution (left plot) develops a complex 
small-scale structure that is absent in the approximate solution (right plot).}
\end{figure}

\begin{figure}[!h]
\begin{center}
\psfrag{n}[0.9]{\small $n$}
\psfrag{x}[1][Bl]{\small $x_n(t)$}
\includegraphics[scale=0.30]{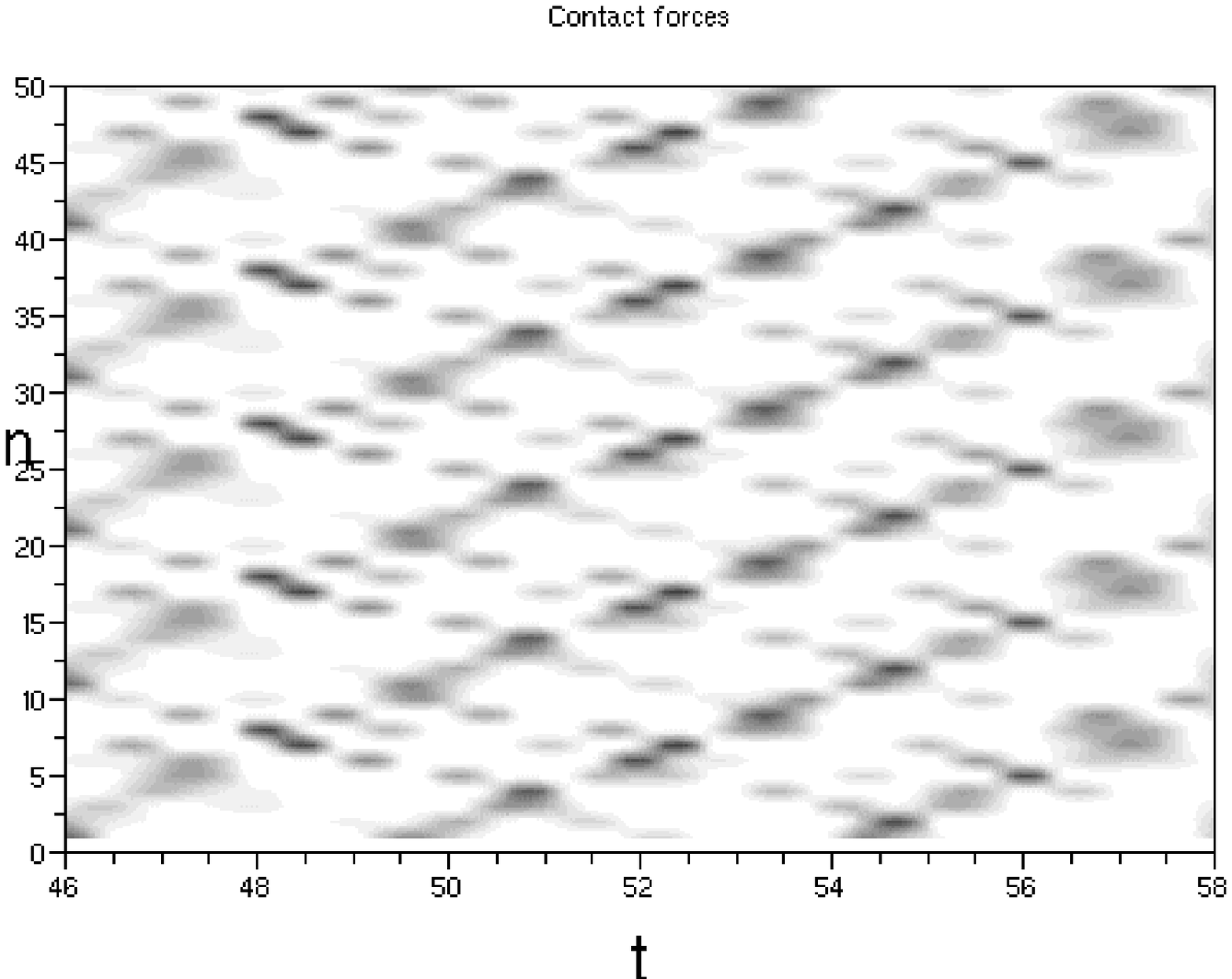}
\includegraphics[scale=0.30]{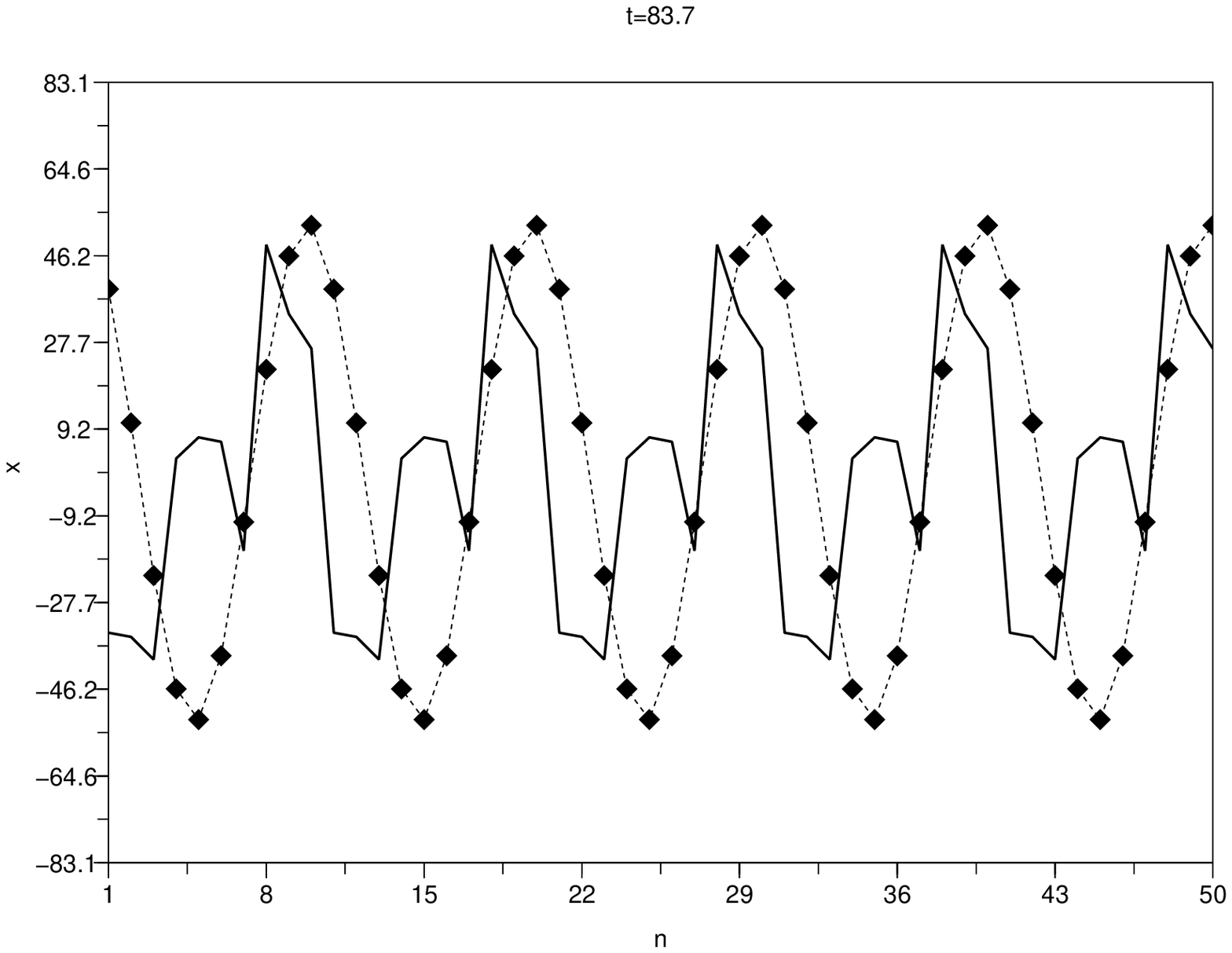}
\end{center} 
\caption{\label{plot1v5}
Left plot~: zoom on the left-side space-time diagram of 
figure \ref{plot1v5st}, showing the small-scale structure
superposed on the travelling wave pattern
of the exact solution. Right plot~: comparison between bead displacements
for the exact (continuous line) and approximate solution (dots)
at $t=83.7$.}
\end{figure}

As a conclusion, we have seen that the 
initial condition obtained with ansatz (\ref{ansatzapprox}) and $q=\pi /5$
generates a stable travelling wave that is qualitatively
well described by the DpS equation if one avoids the
highly nonlinear regime. Moreover, in the weakly nonlinear regime
the DpS equation describes the travelling wave with
excellent precision. 

\subsection{\label{unstable}Modulational instability and breathers}

In what follows we perform the same kind of simulations
but increase the wavenumber $q$ above $\pi /2$, which yields a completely
different dynamical behaviour corresponding to modulational
instability. We first determine an unperturbed initial condition using ansatz
(\ref{ansatzapprox})-(\ref{freq}) with $q=4\pi /5$ and $\omega_{\rm{tw}}=1.1$,
which corresponds to a small amplitude $\epsilon \approx 3,\! 8.10^{-3}$
in (\ref{ansatzapprox1}) (as above the travelling wave solution (\ref{ansatz3})
is normalized by fixing $R=1$). Then we perturb this
initial condition as indicated in (\ref{perturb}),  
where $\rho_n^{(1)}, \rho_n^{(2)}  \in [-0.1 , 0.1]$ denote
uniformly distributed random variables. 
The phenomenon of modulational instability is illustrated by figures 
\ref{plot8-11} and \ref{plot9-10}. The envelope of the initial condition
(figure \ref{plot8-11}, left plot) localizes after some transcient time, 
yielding larger oscillations of some beads 
corresponding to discrete breathers (figure \ref{plot8-11}, right plot).
The modulational instability yields 
a disordered train of travelling breathers that
propagate along the lattice as shown by figure \ref{plot9-10}.  
The DpS equation describes this modulational instability
quite well (compare the two plots of figure \ref{plot9-10}),
and even reproduces fine details like the crossing of two
breathers at $t\approx 330$.

\begin{figure}[!h]
\psfrag{n}[0.9]{\huge $n$}
\psfrag{x}[1][Bl]{\huge $x_n(t)$}
\begin{center}
\includegraphics[angle=-90,scale=0.22]{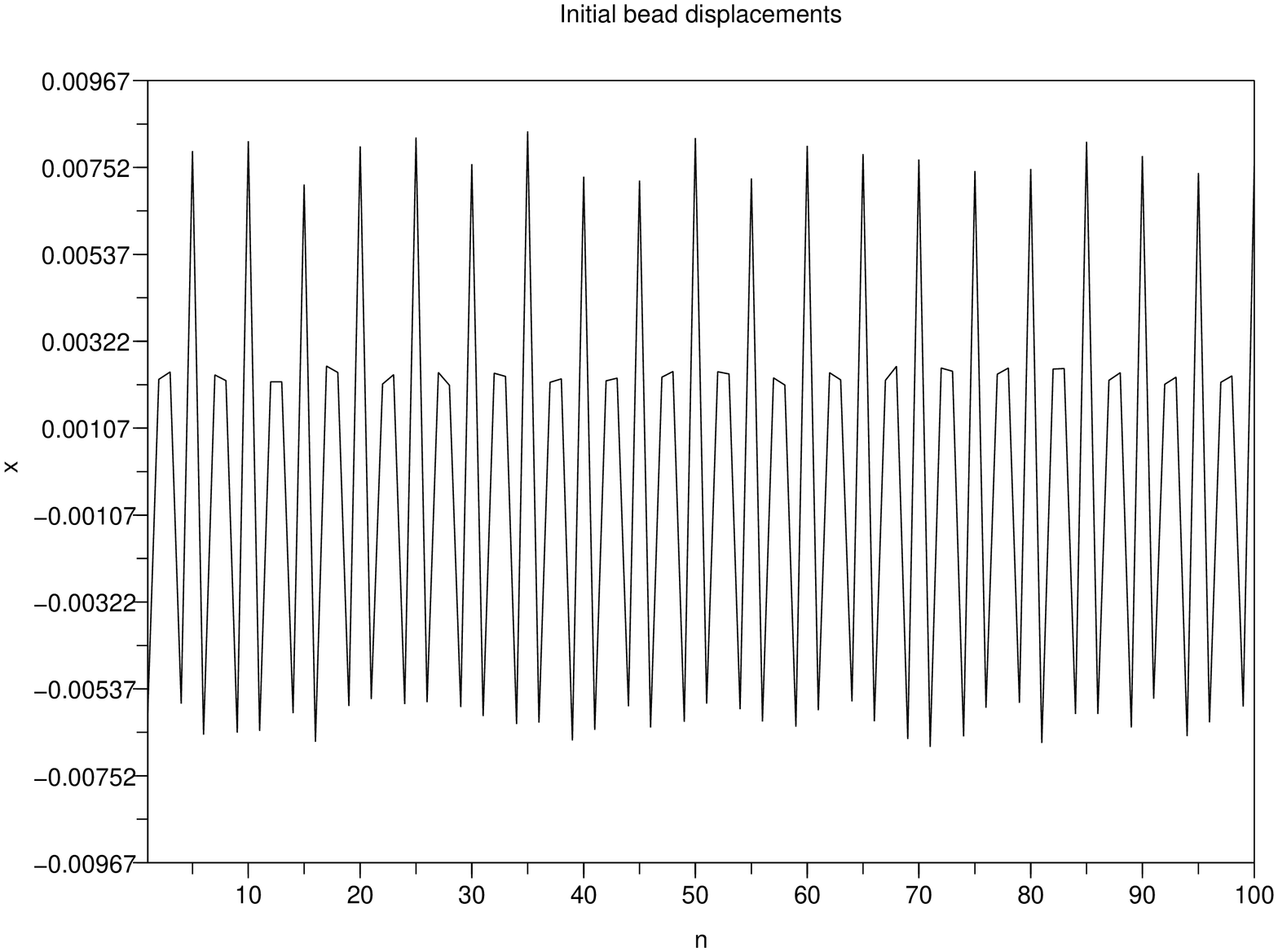}
\includegraphics[angle=-90,scale=0.22]{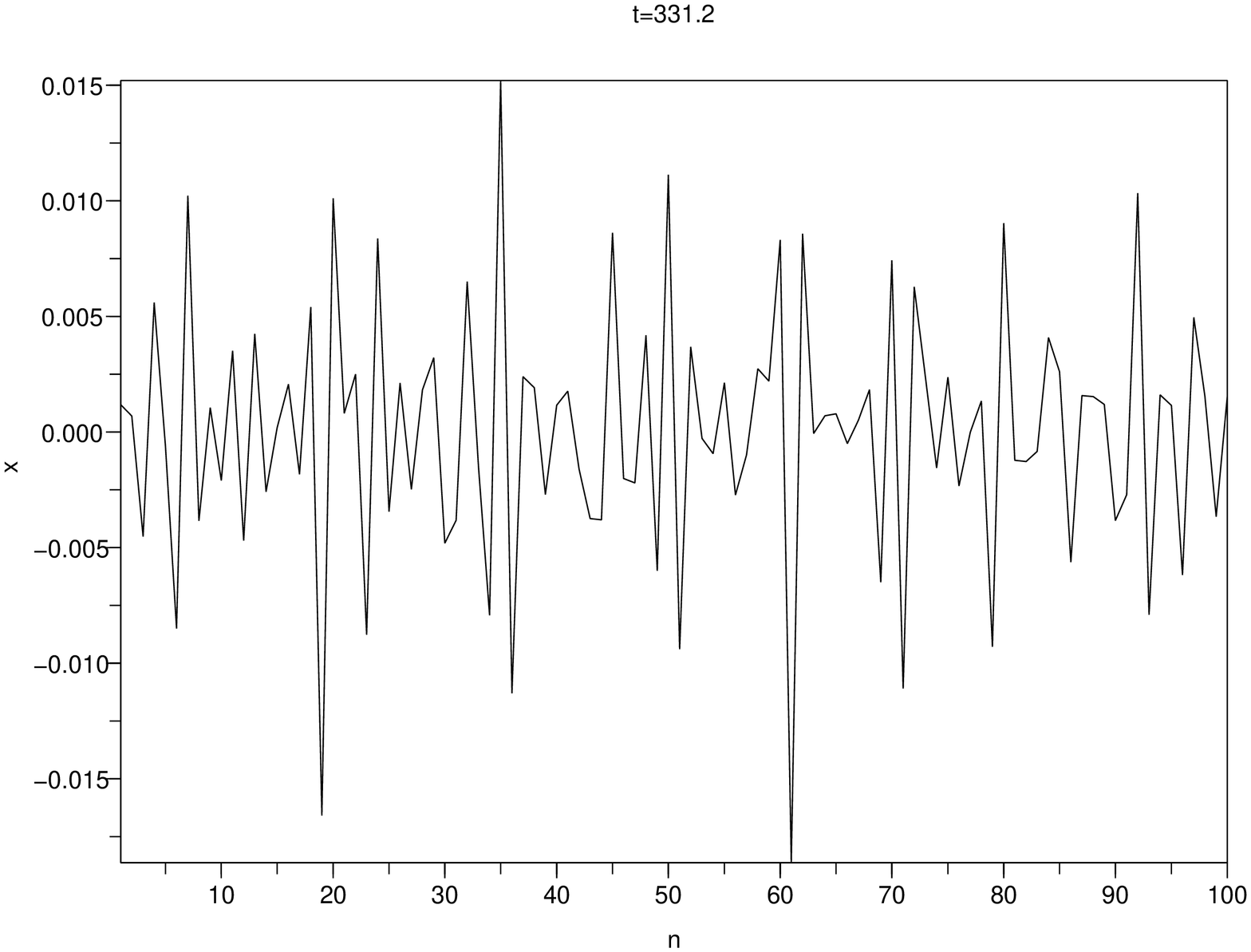}
\end{center} 
\caption{\label{plot8-11} 
Left plot~: initial bead displacements, fixed by the periodic travelling wave 
ansatz (\ref{ansatzapprox})-(\ref{freq})
with $q=4\pi /5$ and $\omega_{\rm{tw}}=1.1$,
perturbed by a small random noise. Right plot~:
bead displacements at $t=331.2$. The dynamics generates alternate regions of
large and small amplitude motion, characteristic of modulational instability.}
\end{figure}

\begin{figure}[!h]
\begin{center}
\includegraphics[scale=0.30]{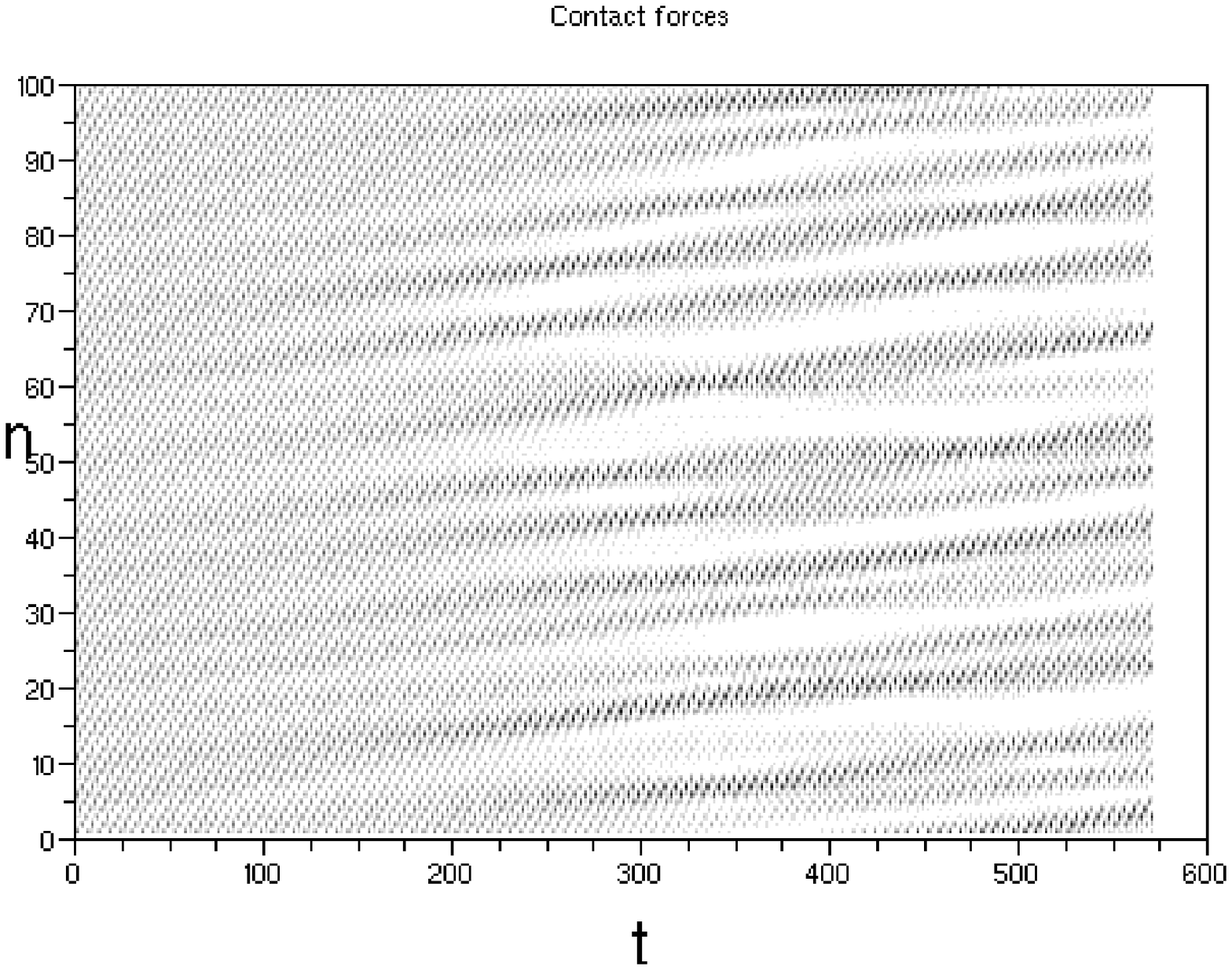}
\includegraphics[scale=0.30]{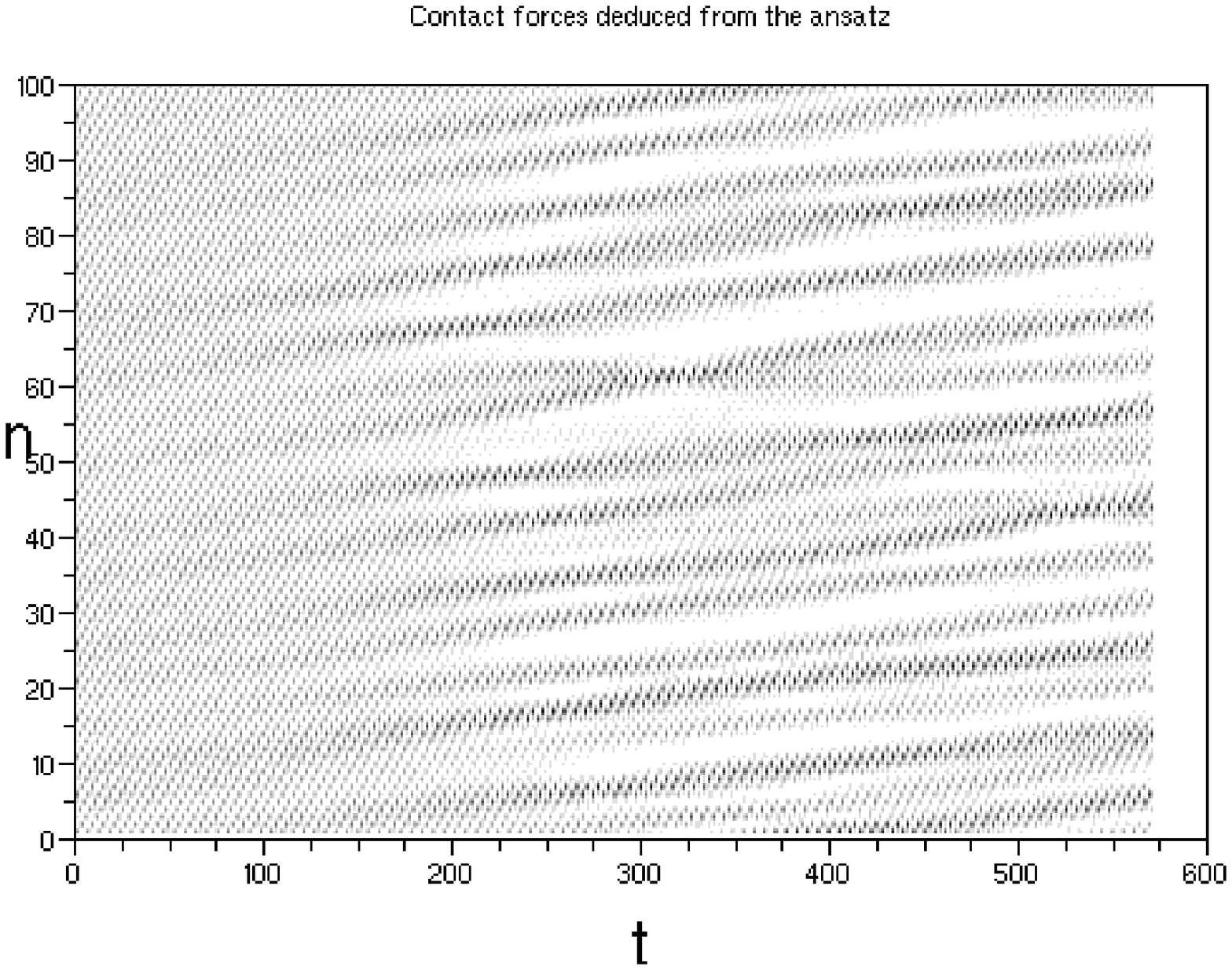}
\end{center} 
\caption{\label{plot9-10}
Left plot~: spatiotemporal evolution of the interaction forces for the
computation of figure \ref{plot8-11}. Forces are represented in grey levels 
as in figure \ref{plot6-7}. The grey region at early times of the simulation
corresponds to a travelling wave profile, and
the dark lines to a train of travelling breathers generated by the modulational 
instability. Right plot~: approximate solution deduced from the DpS equation.}
\end{figure}

The average velocity of the travelling breathers generated near the onset of instability
depends on the wavenumber of the unstable travelling wave. When $q$ is equal or
close to $\pi$  ($q=\pi$ corresponding to binary oscillations), one typically observes 
static breathers that remain pinned to some lattice sites over long times. 
This phenomenon is illustrated by figures \ref{plot29-30} and \ref{plot26-27},
which correspond to the case $q=\pi$, $\omega_{\rm{tw}} =1.1$ and $\rho_n^{(1)}, \rho_n^{(2)}  \in [-0.5 , 0.5]$. 
The DpS equation provides a good
qualitative picture of the instability (compare the two plots of figure \ref{plot26-27}).
It is able to reproduce some complex phenomena like the interaction and subsequent
merging of two close static breathers, at $t\approx 223$. However the local
dynamics after the two breathers have merged is not well described by
the DpS equation, since more energy is released at the collision site
for the exact solution than with the approximate solution. This may be
the sign that higher harmonics neglected by ansatz (\ref{ansatz}) come into play
during the merging phase, or may simply originate from sensitiveness
of the collision result to relative phases and breather positions.

In the above space-time diagrams, travelling and static breathers appear 
as grey lines due to their internal breathing motion. These oscillations are more visible on figure \ref{plot28},
where static breathers with different amplitudes and spatial extentions are shown.
  
\begin{figure}[!h]
\psfrag{n}[0.9]{\huge $n$}
\psfrag{x}[1][Bl]{\huge $x_n(t)$}
\begin{center}
\includegraphics[angle=-90,scale=0.22]{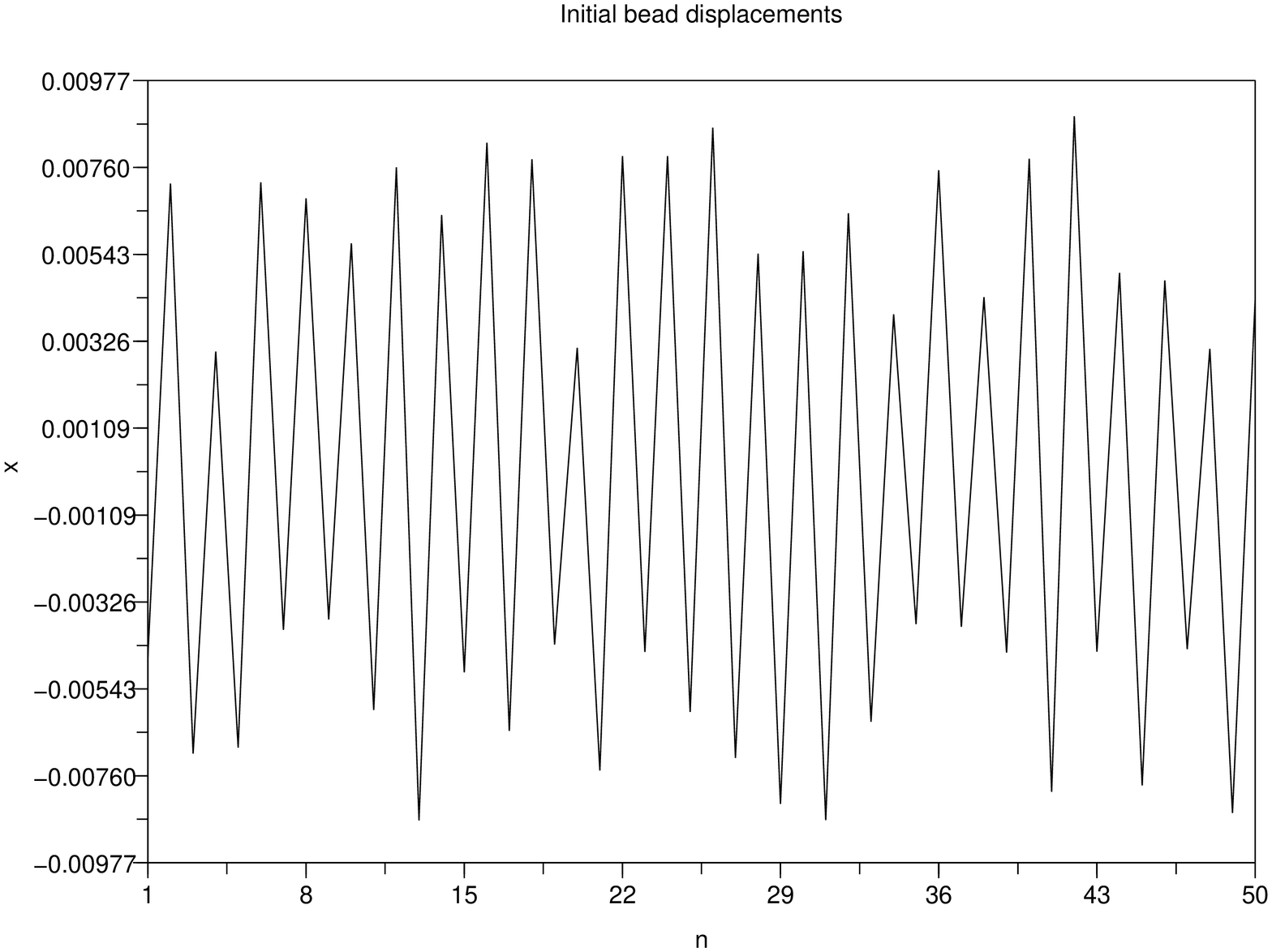}
\includegraphics[angle=-90,scale=0.22]{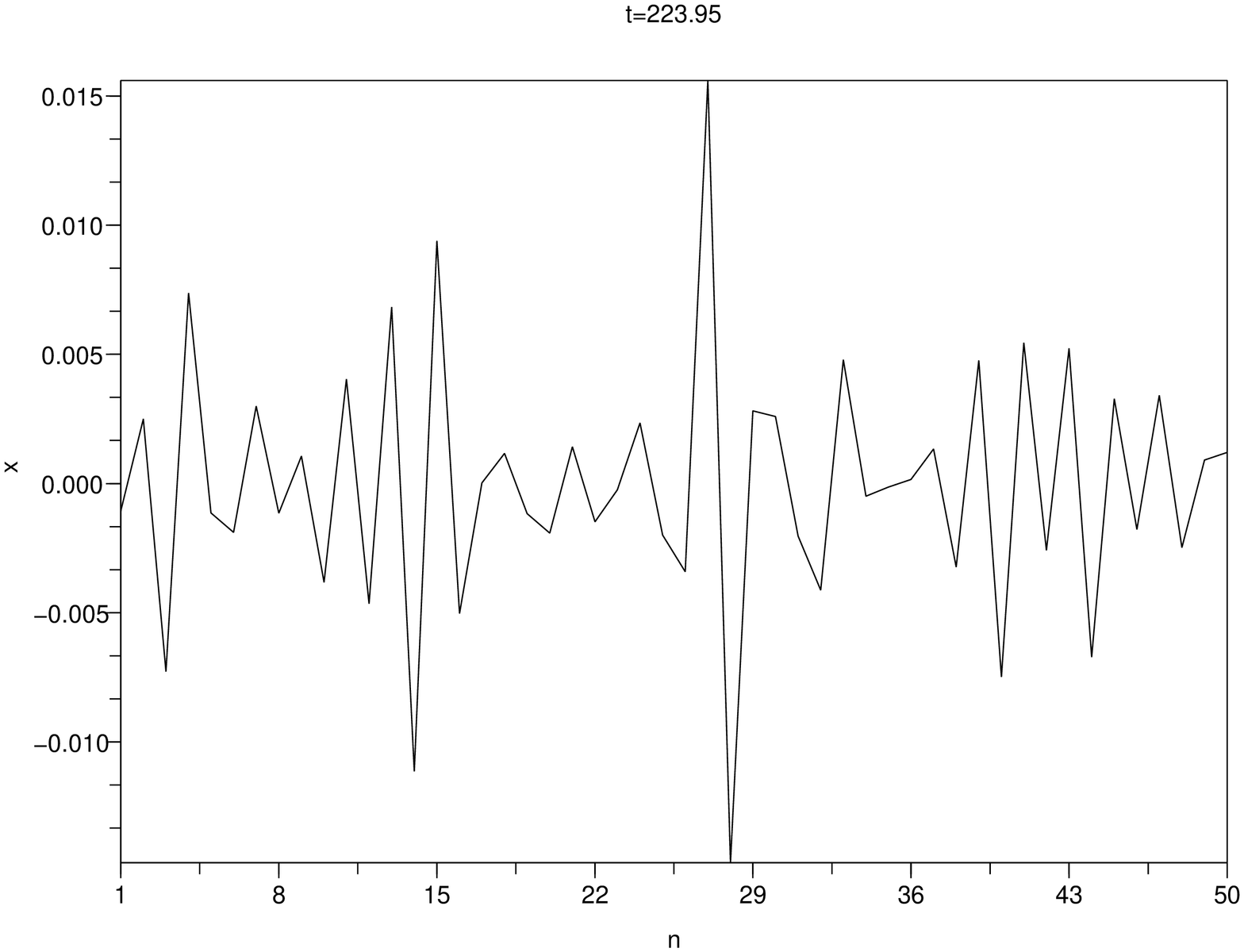}
\end{center} 
\caption{\label{plot29-30} 
Left plot~: initial bead displacements, corresponding to a binary
oscillation ($q=\pi $ and $\omega_{\rm{tw}}=1.1$ in ansatz (\ref{ansatzapprox})-(\ref{freq})),
perturbed by a random noise. Right plot~: snapshot of a static breather formed at $t\approx 223$ after
the merging of two smaller breathers.}
\end{figure}

\begin{figure}[!h]
\begin{center}
\includegraphics[scale=0.30]{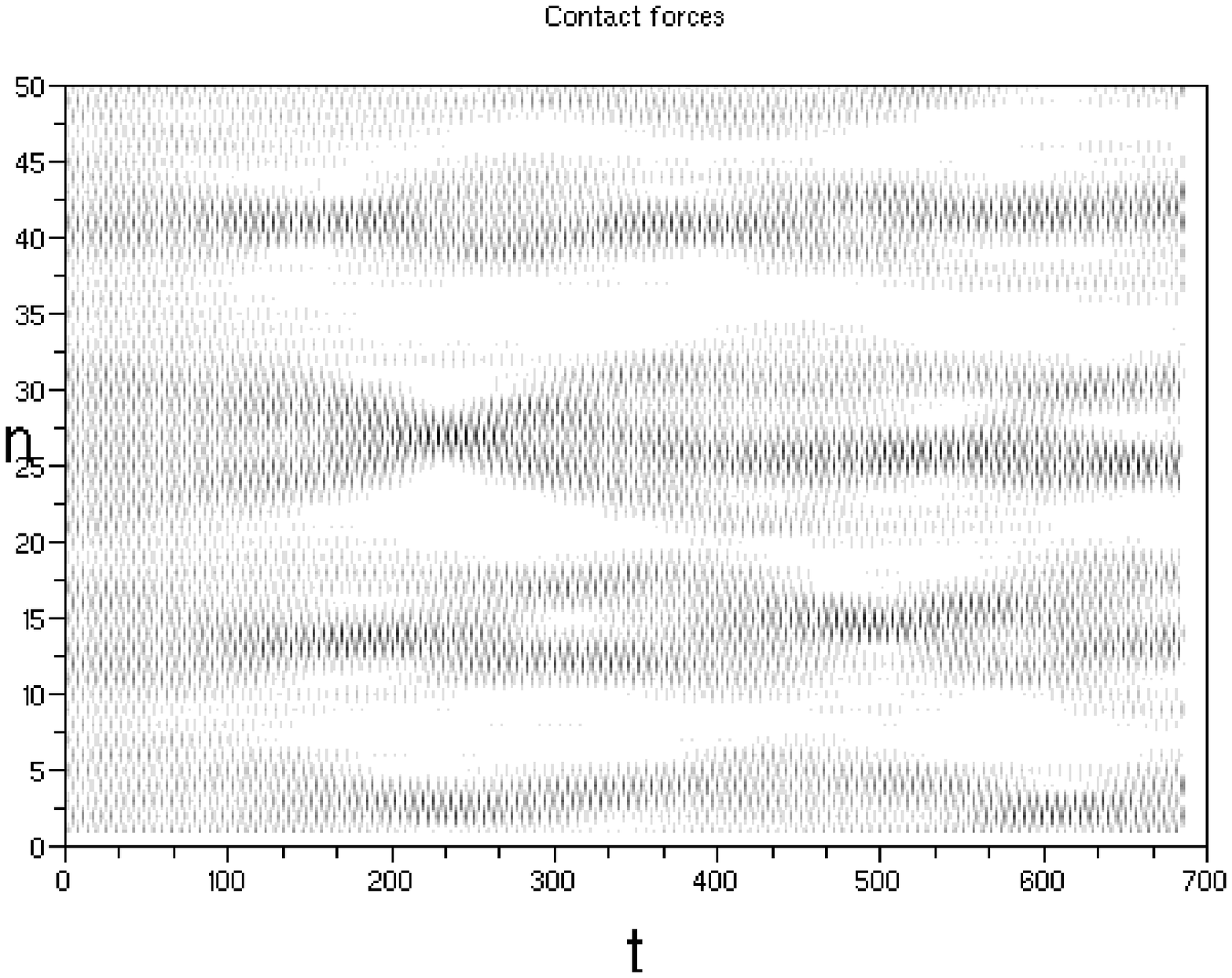}
\includegraphics[scale=0.30]{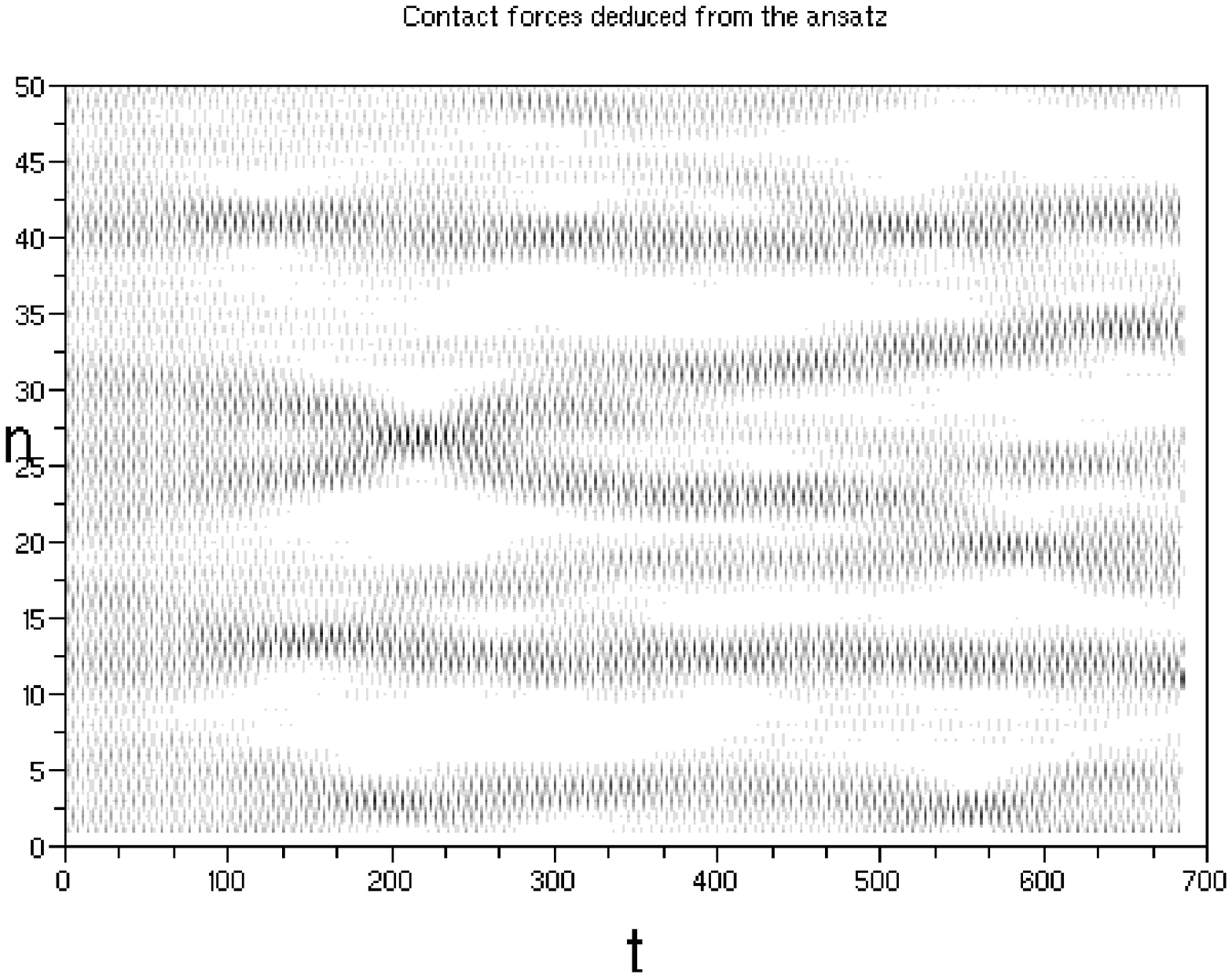}
\end{center} 
\caption{\label{plot26-27} 
Spatiotemporal evolution of the interaction forces for the
computation of figure \ref{plot29-30} (left plot), and its
approximation deduced from the DpS equation (right plot). 
The grey region at early times of the simulation
corresponds to the perturbed binary oscillation, and  
the dark regions to nearly static breathers generated by the modulational 
instability.}
\end{figure}

\begin{figure}[!h]
\begin{center}
\includegraphics[scale=0.35]{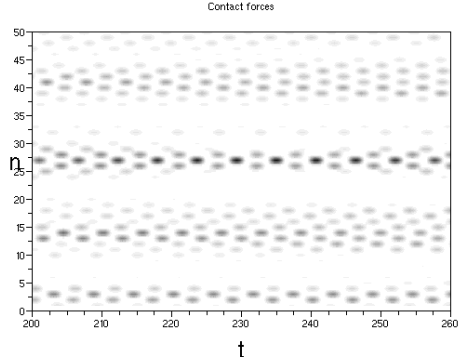}
\end{center} 
\caption{\label{plot28} 
Detail of figure \ref{plot26-27}, showing a distribution
of static breathers with different sizes and amplitudes. 
Black spots alternating with white regions correspond
to sequences of compressions and separations
of neighbouring beads.}
\end{figure}

\ve

During the modulational instability, only a part 
of the modes are initially amplified. In order to
analyze how the DpS equation is able
to reproduce this phenomenon,
we now study the time evolution of the spatial
Fourier transform of $\{x_n (t) \}$.
As previously we consider periodic boundary conditions
$x_{n+N}(t)=x_n(t)$, but increase the number of particles
to $N=200$ in order to achieve a better spectral resolution. 
Considering the discretization of the spectral band $[-\pi , \pi ]$ given by
$
\Gamma_\ast = \frac{2\pi}{N}\cdot \{ -\frac{N}{2}+1, \ldots , \frac{N}{2}  \}
$
and defining the discrete Fourier transform
$$
\hat{x}_k (t)=\sum_{n=0}^{N-1}{x_n (t)\, e^{-i\, n\, k}}, \ \ \
k\in \Gamma_\ast , 
$$
we have 
$$
x_n (t)= \frac{1}{N}\, \sum_{k\in \Gamma_\ast}{\hat{x}_k (t) \, e^{i\, n\, k}}.
$$
We compare the time evolution of the magnitude 
of the Fourier transform $|\hat{x}_k (t) |$ for the
exact and approximate solutions, for initial travelling wave profiles perturbed by random noise.
The unperturbed initial condition is determined using ansatz (\ref{ansatzapprox})-(\ref{freq}), for
$\omega_{\rm{tw}} =1.1$ and different wavenumbers $q$ that determine
the amplitude $ \epsilon =a/2$ in (\ref{ansatzapprox1}).
We perturb this initial condition as indicated in 
(\ref{perturb}), where $\rho_n^{(1)}, \rho_n^{(2)} $ are random variables generated from
the Gaussian distribution with mean $0$ and standard deviation $0,\! 5.10^{-2}$.

Figure \ref{fft106} compares the spectra of 
the exact and approximate solutions for $q=4\pi /5$.
At $t=106.7$ (first row), the two peaks at $k=\pm q$
correspond to the periodic travelling wave generated
by the initial condition. Two second harmonics at $k=\pm 2(q-\pi)$ 
are visible
for the exact solution and are absent for the approximate one
(these harmonics are generated at early times of the simulation
even in the absence of noise). 
As time further increases, a band of unstable modes grows in amplitude 
at both sides of $k=\pm q$.  
The unstable bandwidths and most unstable wavenumbers 
of the exact and approximate solutions are quite close
(compare the plots at $t=399.95$, third row), 
but the approximate solution overestimates the
initial growth rate (compare the plots at $t=253.45$, second row). 
Both for the exact and approximate solutions, a breathing
of the spectrum near $k=\pm \pi$ settles progressively. 
The plots at the fourth row show quite similar spectral distributions 
of the exact solution at $t=541.95$ and of the approximate one
at $t=543.7$ (we have slightly shifted time to compensate a 
phase-shift in the spectral breathing).
As a conclusion, the DpS equation describes quite
accurately the modulational instability of the travelling wave with
$q=4\pi /5$ in the weakly nonlinear regime.

\begin{figure}[!h]
\psfrag{k}[0.9]{\huge $k$}
\psfrag{m}[1][Bl]{\huge $|\hat{x}_k(t) |$}
\begin{center}
\includegraphics[angle=-90,scale=0.22]{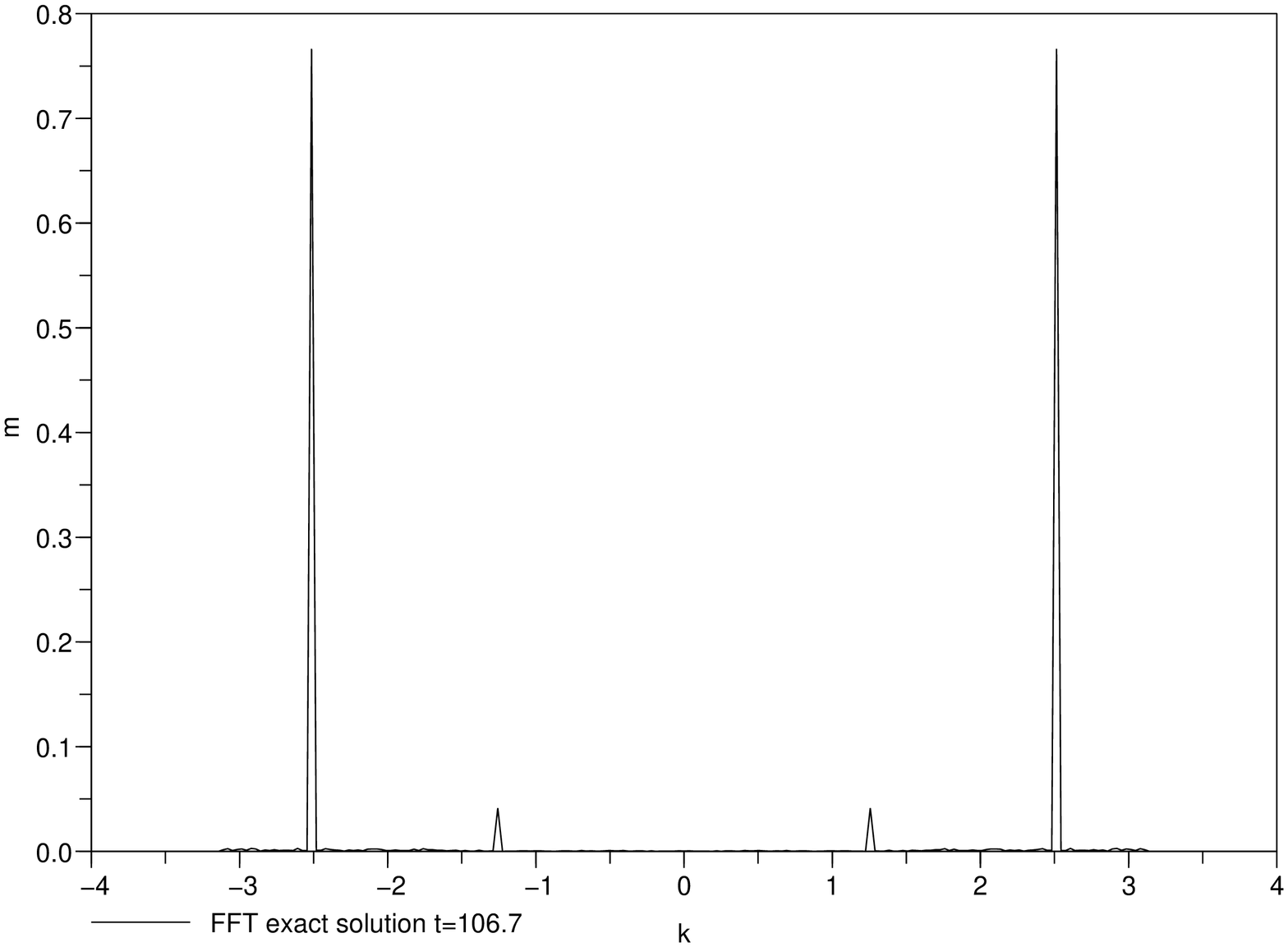}
\includegraphics[angle=-90,scale=0.22]{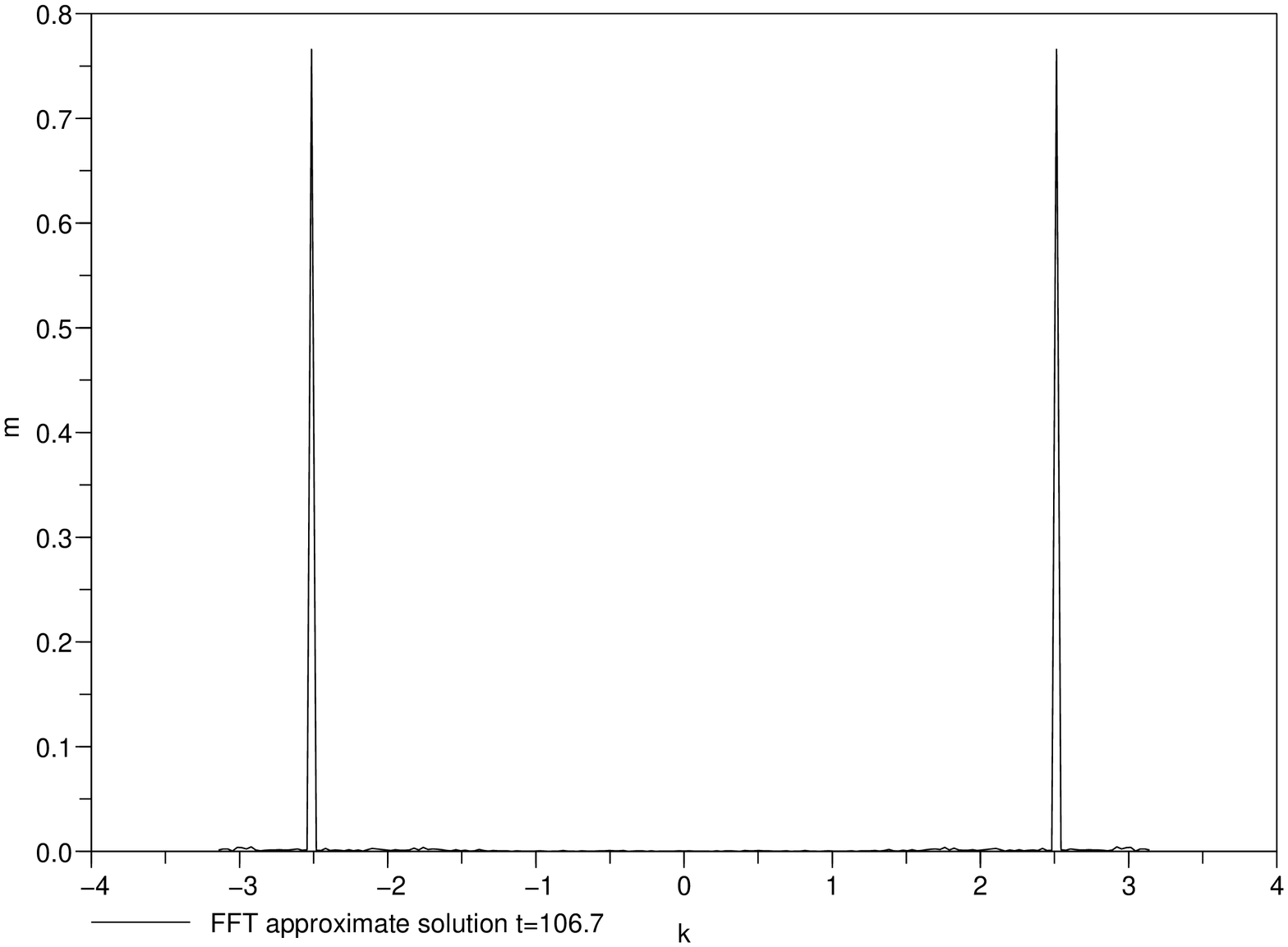}
\includegraphics[angle=-90,scale=0.22]{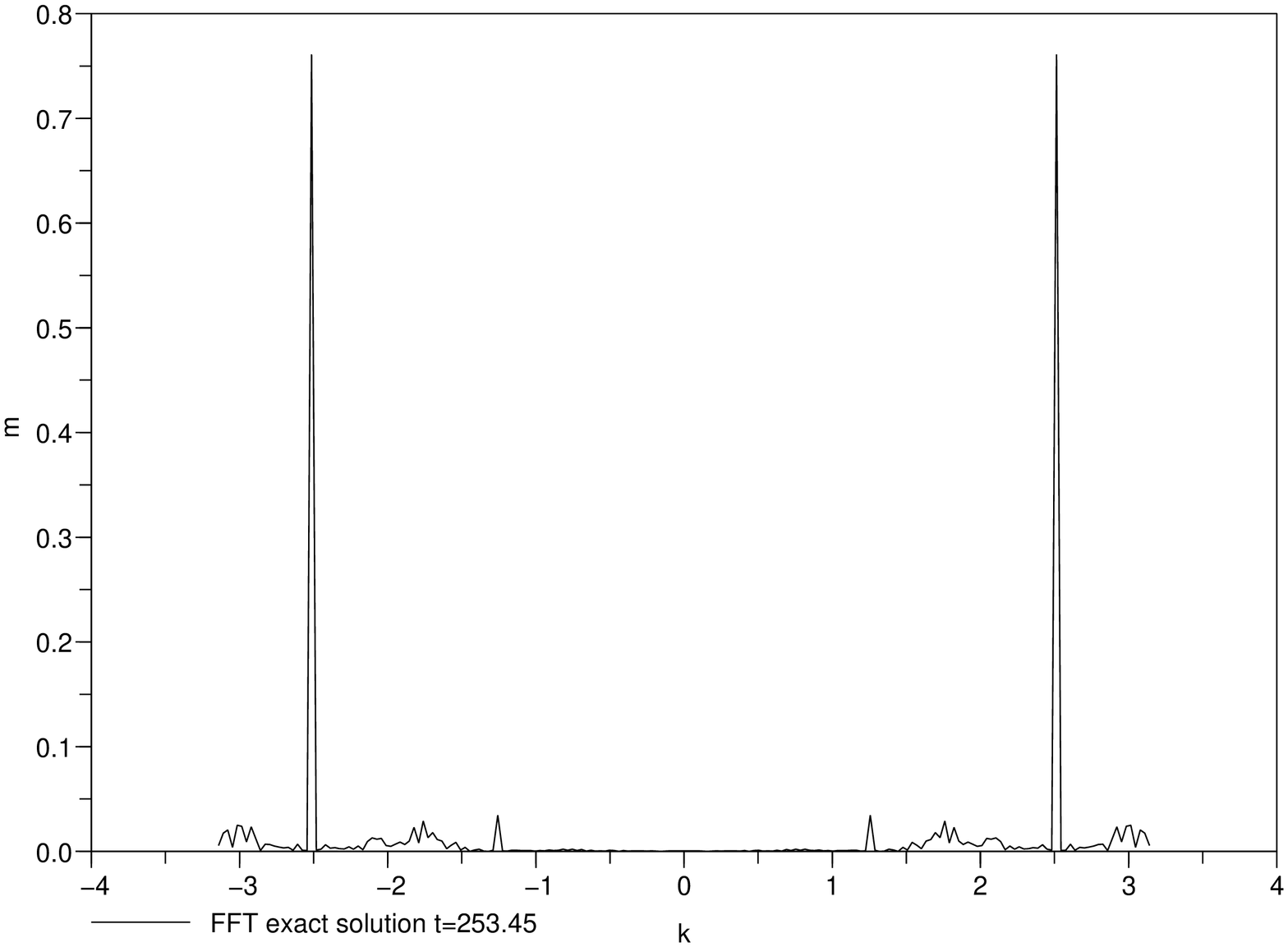}
\includegraphics[angle=-90,scale=0.22]{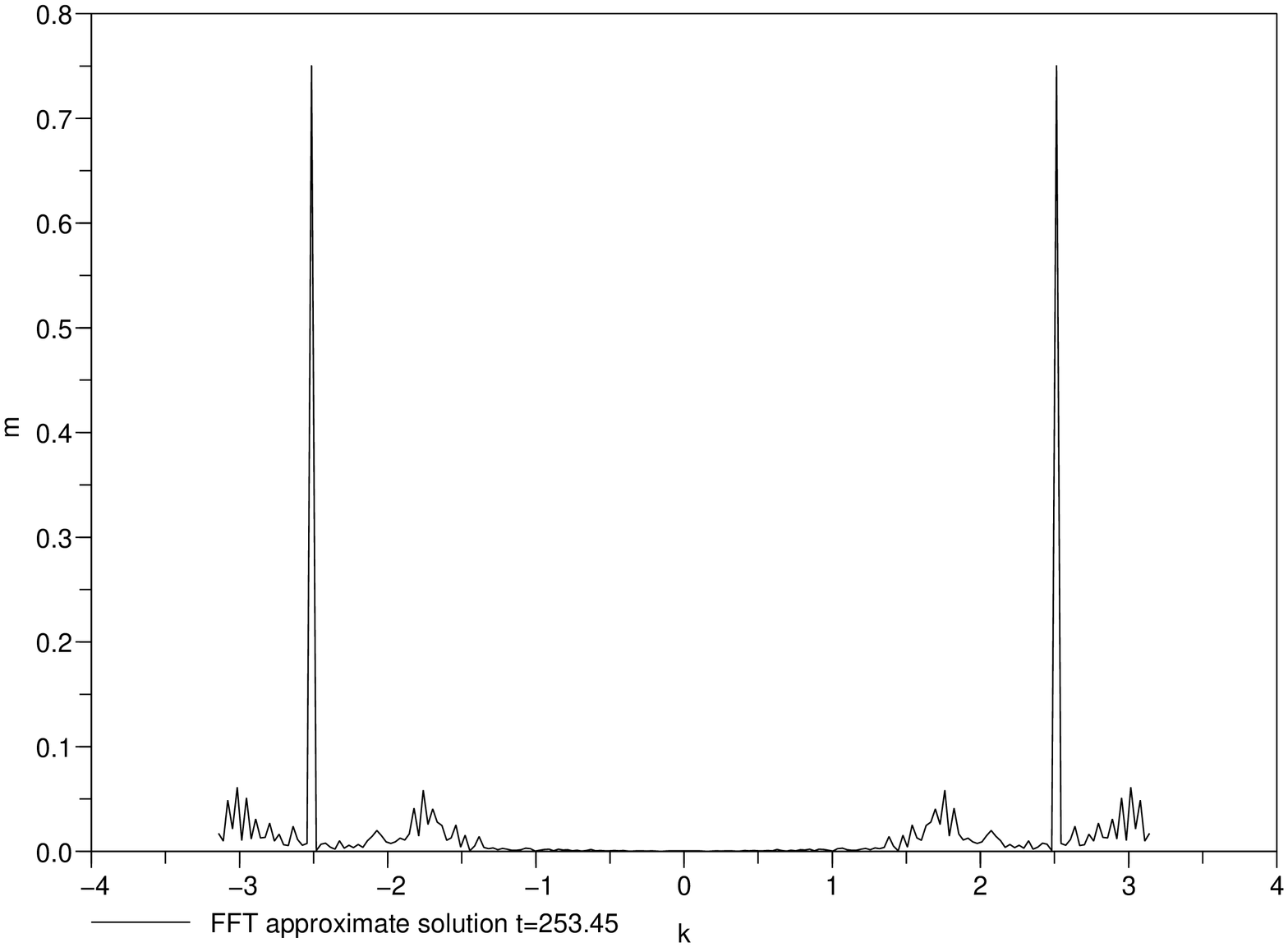}
\includegraphics[angle=-90,scale=0.22]{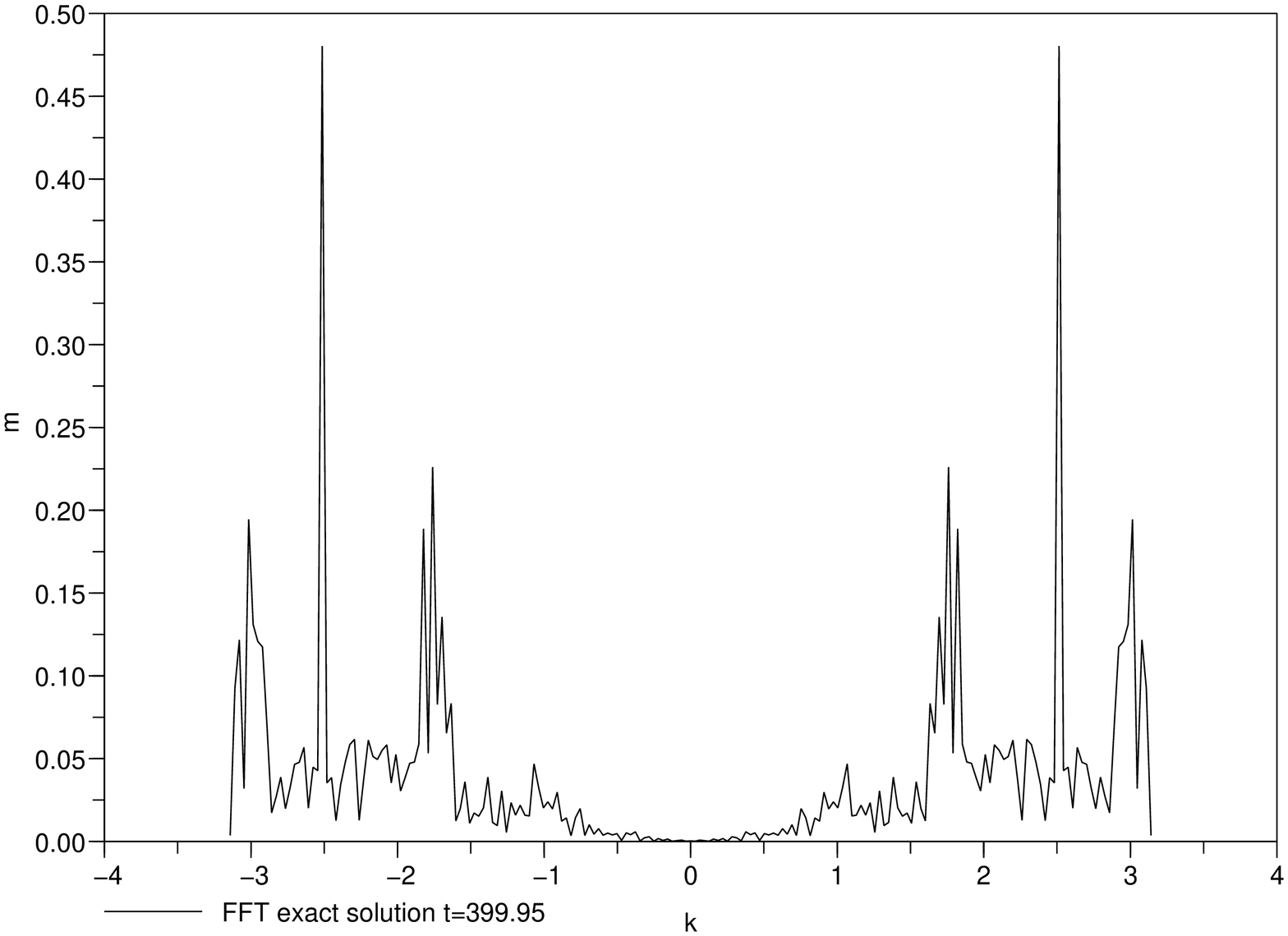}
\includegraphics[angle=-90,scale=0.22]{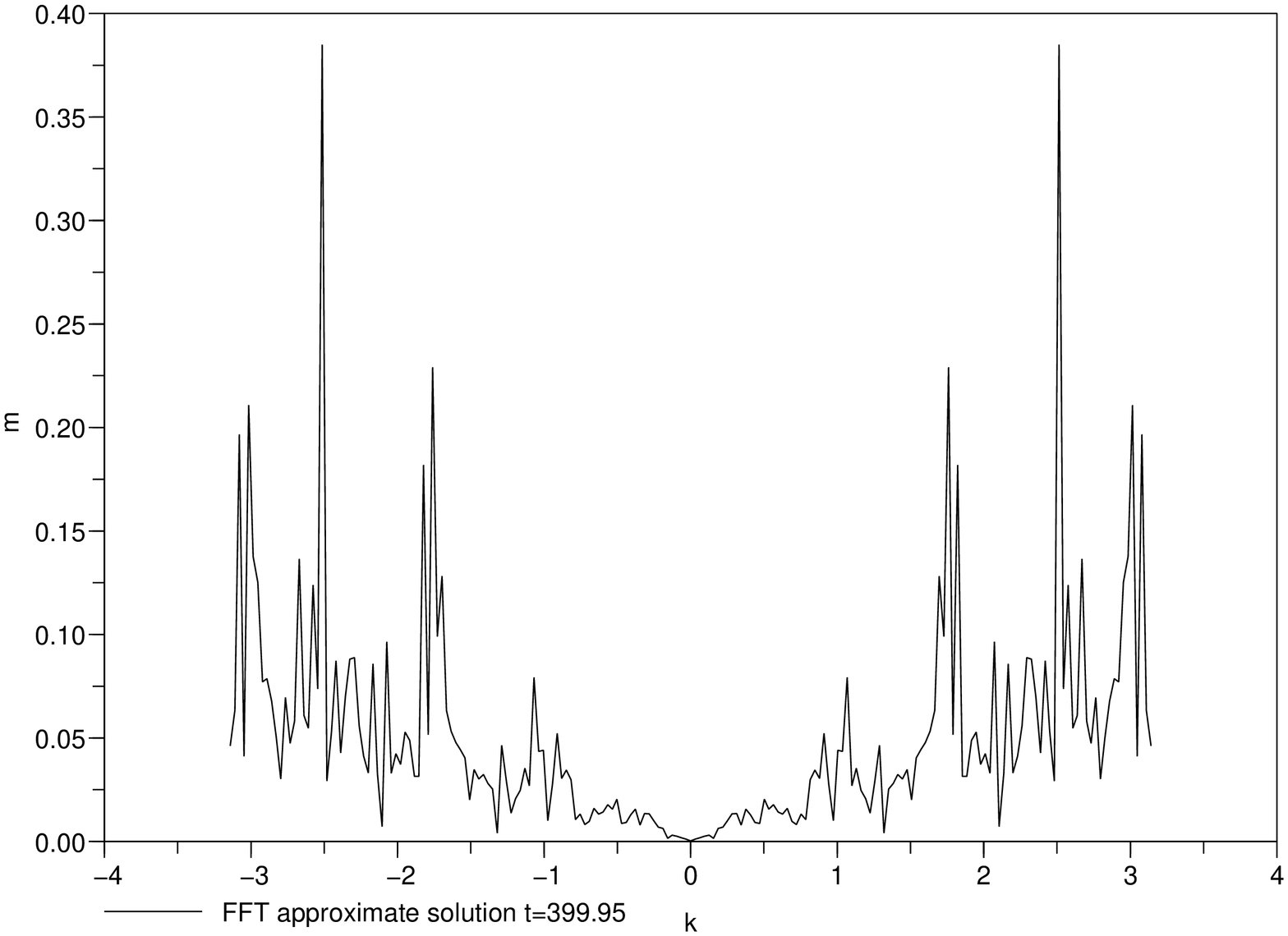}
\includegraphics[angle=-90,scale=0.22]{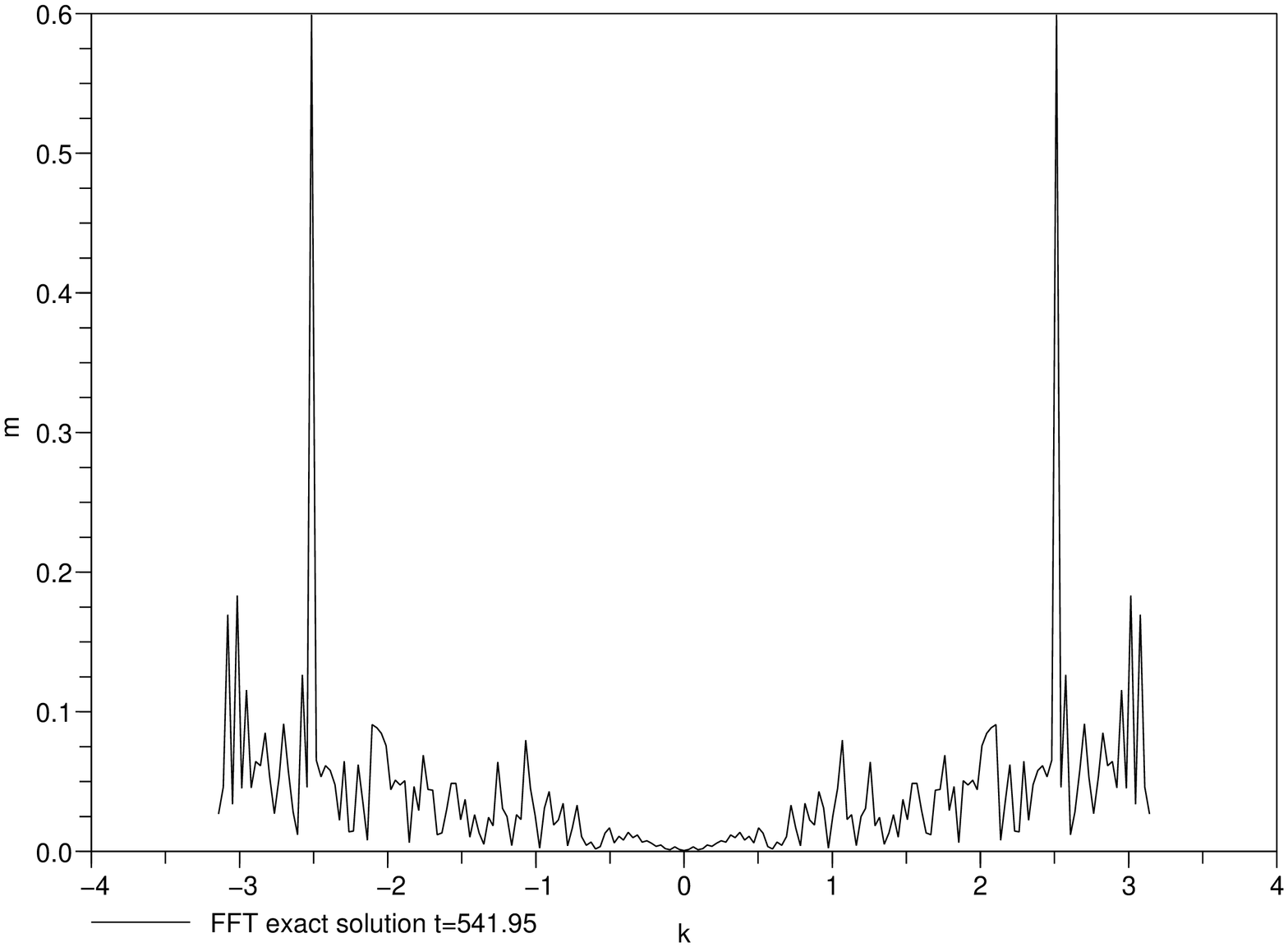}
\includegraphics[angle=-90,scale=0.22]{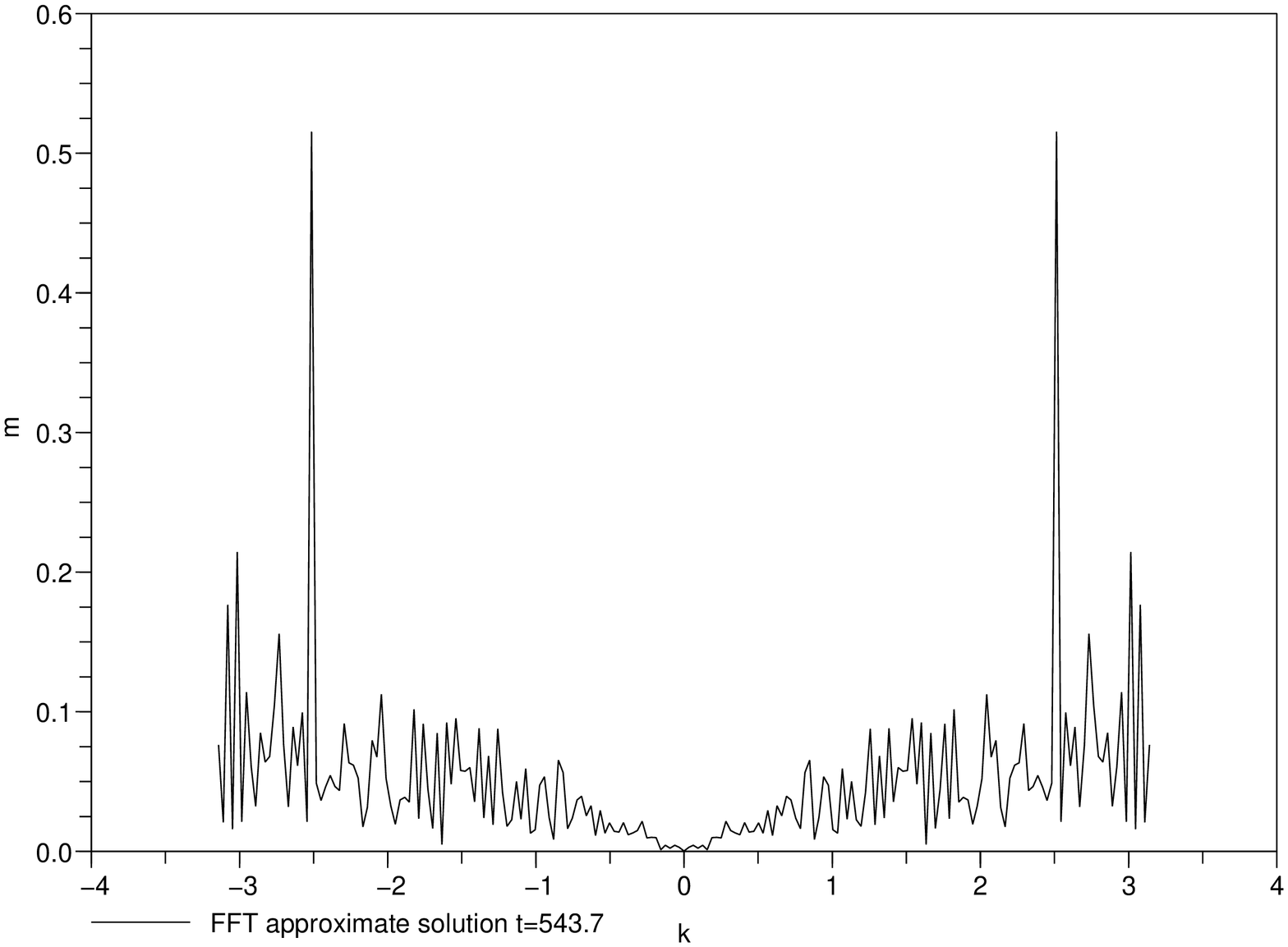}
\end{center} 
\caption{\label{fft106} 
Magnitude of the discrete spatial Fourier transforms of the
exact solution (left column) and approximate solution
(right column), for a random perturbation of a small amplitude
travelling wave with $q=4\pi /5$.  Plots correspond to
different times $t=106.7$ (first row), $t=253.45$ (second row),
$t=399.95$ (third row), $t=541.95$ (fourth row, left),
$t=543.7$ (fourth row, right).}
\end{figure}

\ve

In contrast with the above results, we describe below two situations in which the DpS equation gives
incorrect results on the stability or dynamical properties of travelling waves.

\ve

The first one corresponds to a simple instability threshold effect.
For $q=\pi /2$ and below, we haven't observed modulational
instabilities both for the DpS equation and
system (\ref{nc}) in the weakly nonlinear regime. 
However, in both models 
we have found thresholds for the modulational instability
at slightly different critical values of $q$ near $\pi /2$.
This is illustrated by figure \ref{fftseuil}, where 
the magnitudes of the spatial Fourier transforms of the
exact and approximate solutions are plotted at $t=913.7$ 
($160$ time periods of the initial periodic travelling wave)
for different values of $q$ near $\pi /2$.
As shown by the two plots at the first row, small amplitude travelling waves
with $q=52\pi /100$ are unstable both for system (\ref{nc}) and the DpS
approximation. When $q$ decreases to $q=51 \pi /100$, the modulational
instability persists for the DpS approximation, but it does not occur
(or may occur extremely slowly) for system (\ref{nc})
(see plots at the second row).
The third row of figure \ref{fftseuil} provides the
spectra for $q=\pi /2$, showing no trace of modulational instability.
In addition to the main peaks, one can see
second harmonics
at $k=\pm\pi$ in the exact solution spectrum,
and a small noisy background 
(more important for the exact solution) that 
follows from an energy cascade between modes.  
As a conclusion, these computations reveal that 
system (\ref{nc}) and the DpS approximation yield slightly
different instability thresholds for the wavenumber $q$,
and consequently the DpS approximation does not correctly
describe the dynamics when $q$ is chosen between the two 
critical values.

\begin{figure}[!h]
\psfrag{k}[0.9]{\huge $k$}
\psfrag{m}[1][Bl]{\huge $|\hat{x}_k(t) |$}
\begin{center}
\includegraphics[angle=-90,scale=0.22]{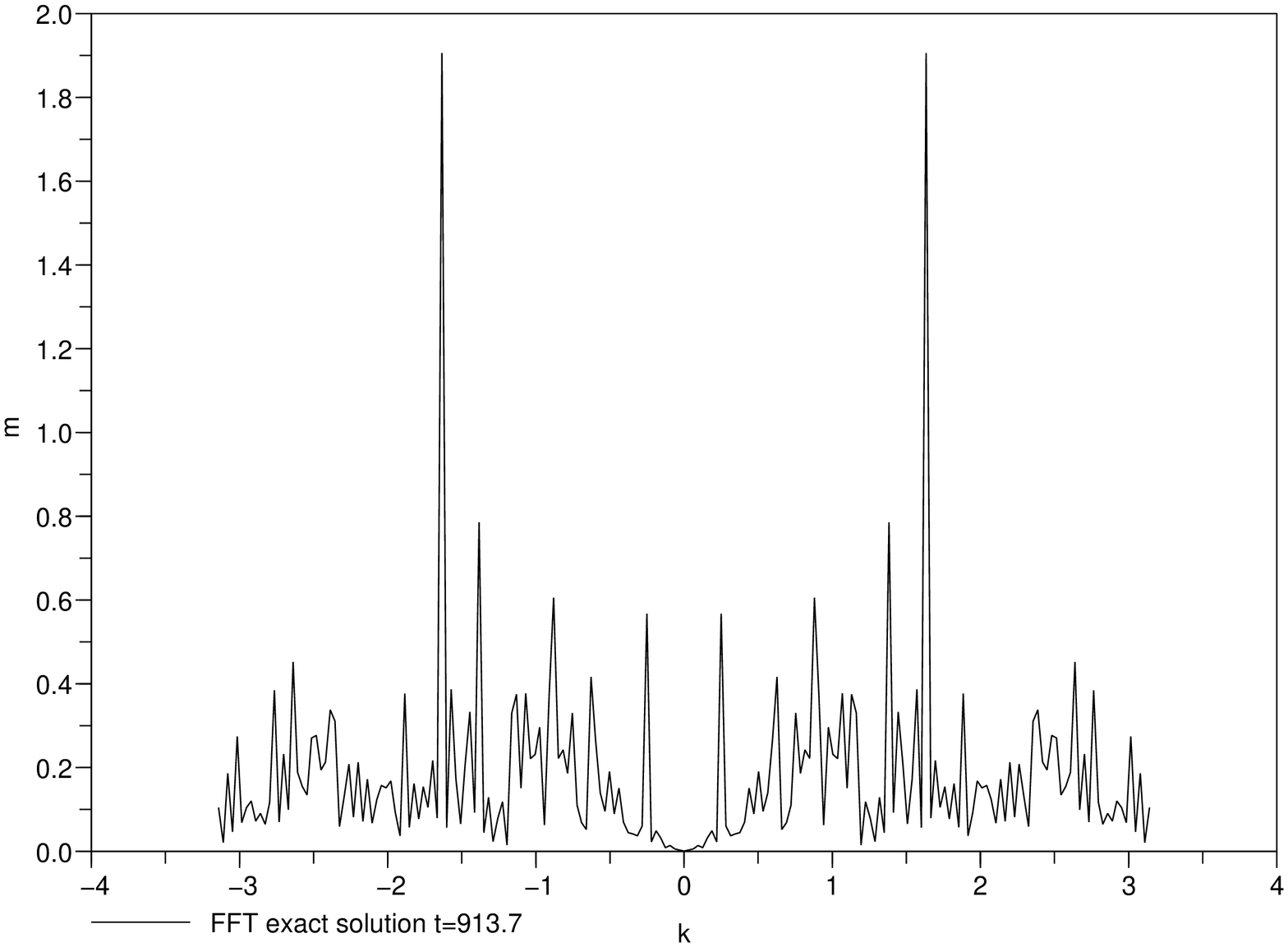}
\includegraphics[angle=-90,scale=0.22]{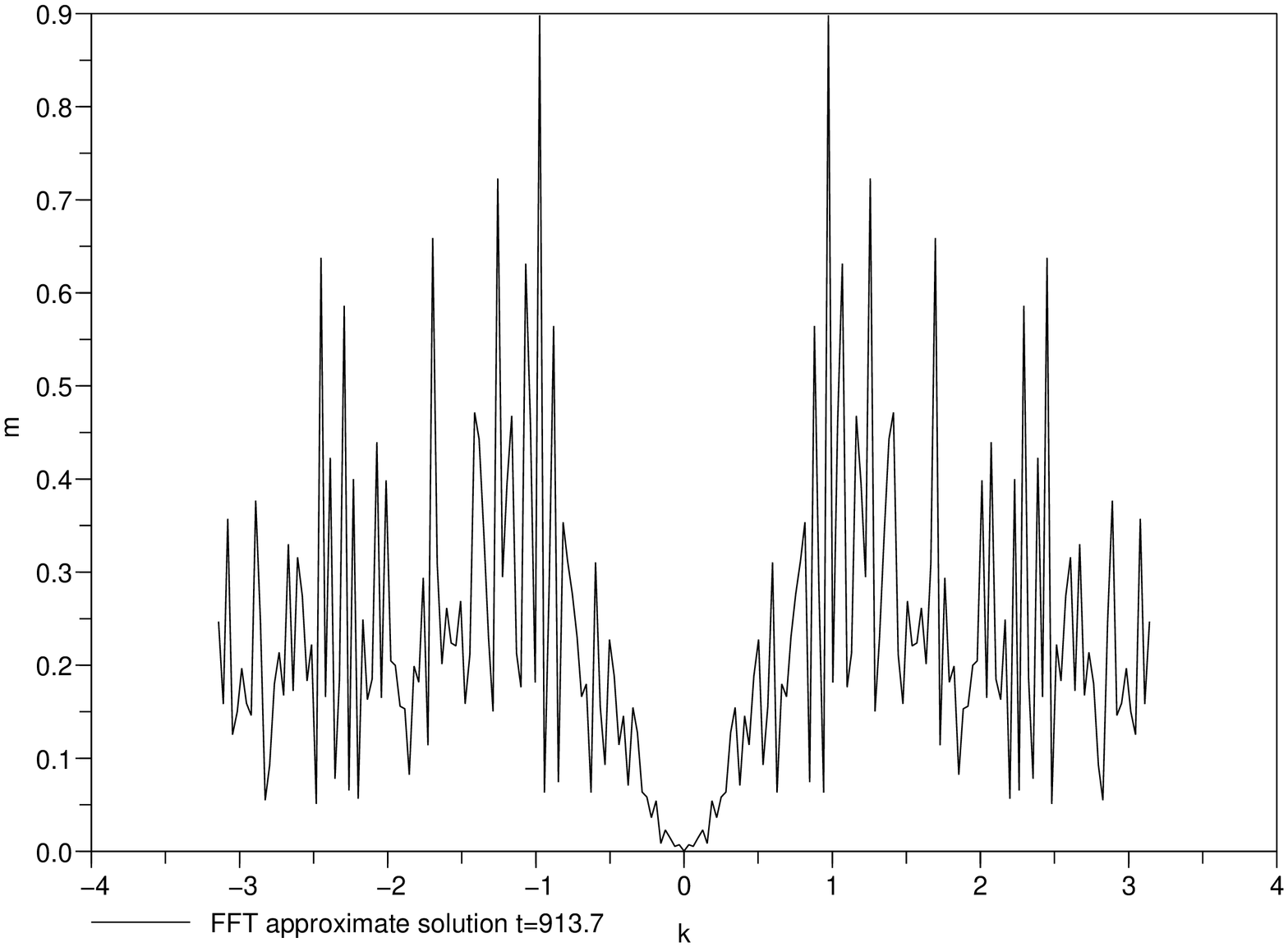}
\includegraphics[angle=-90,scale=0.22]{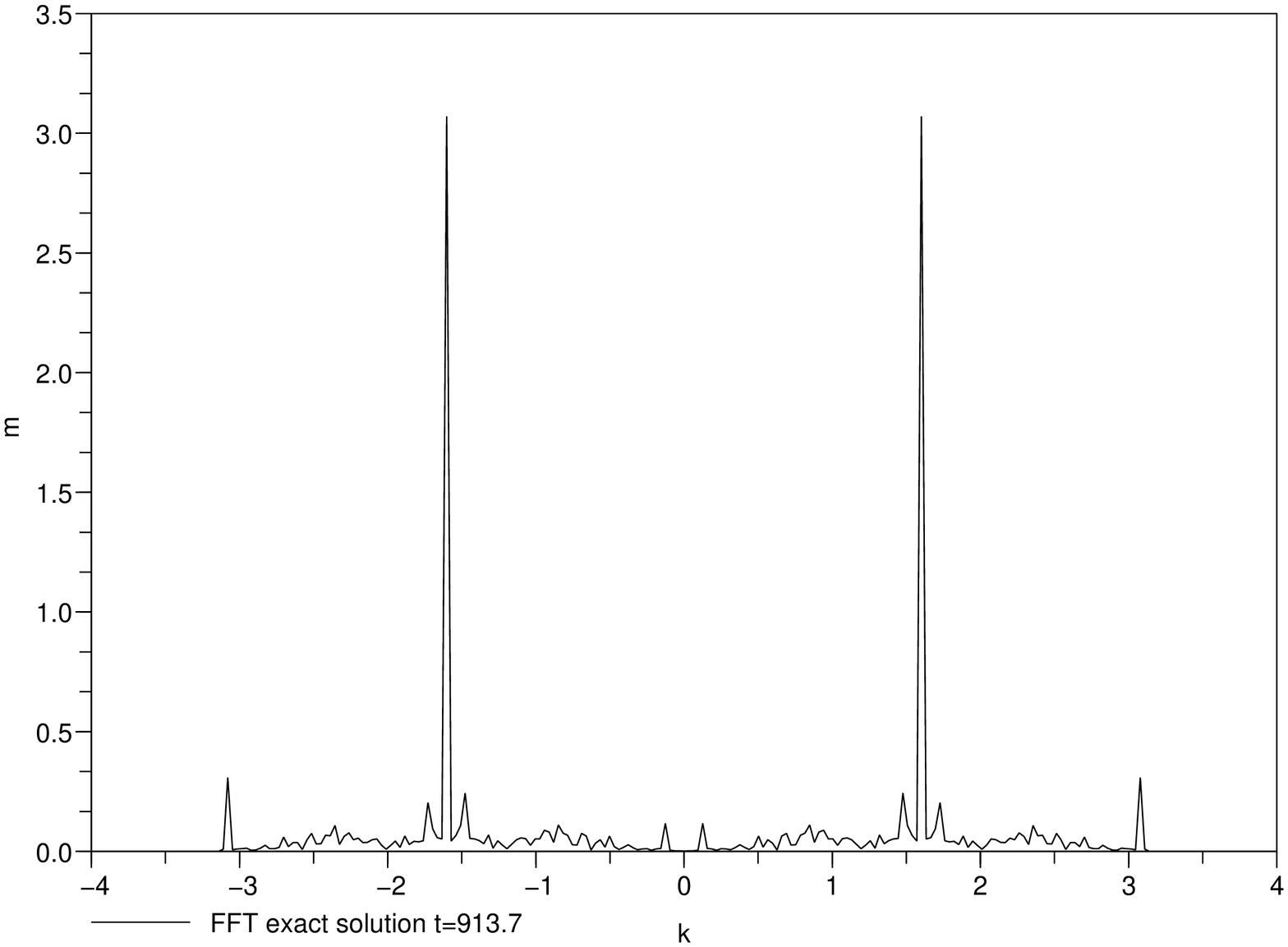}
\includegraphics[angle=-90,scale=0.22]{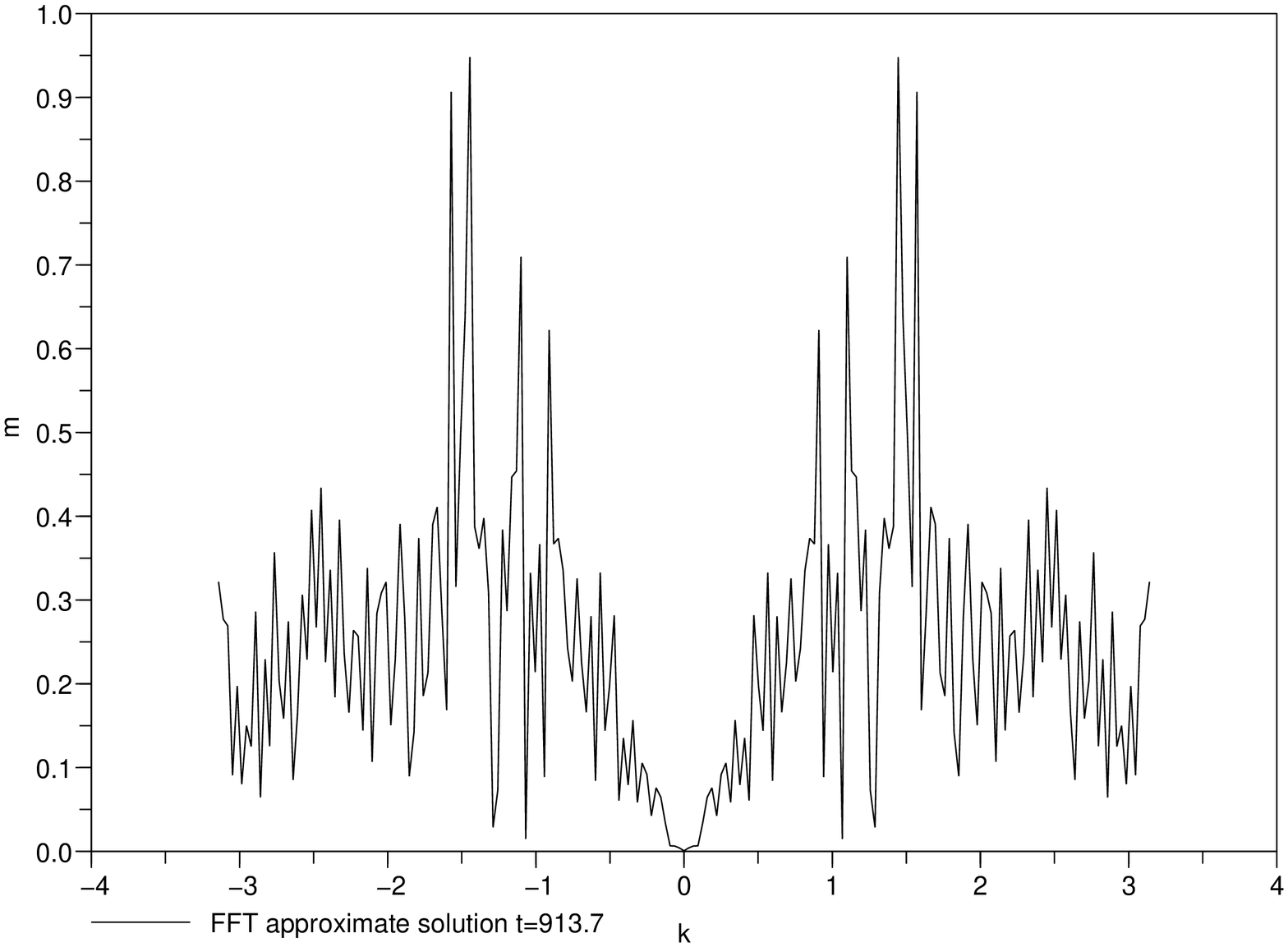}
\includegraphics[angle=-90,scale=0.22]{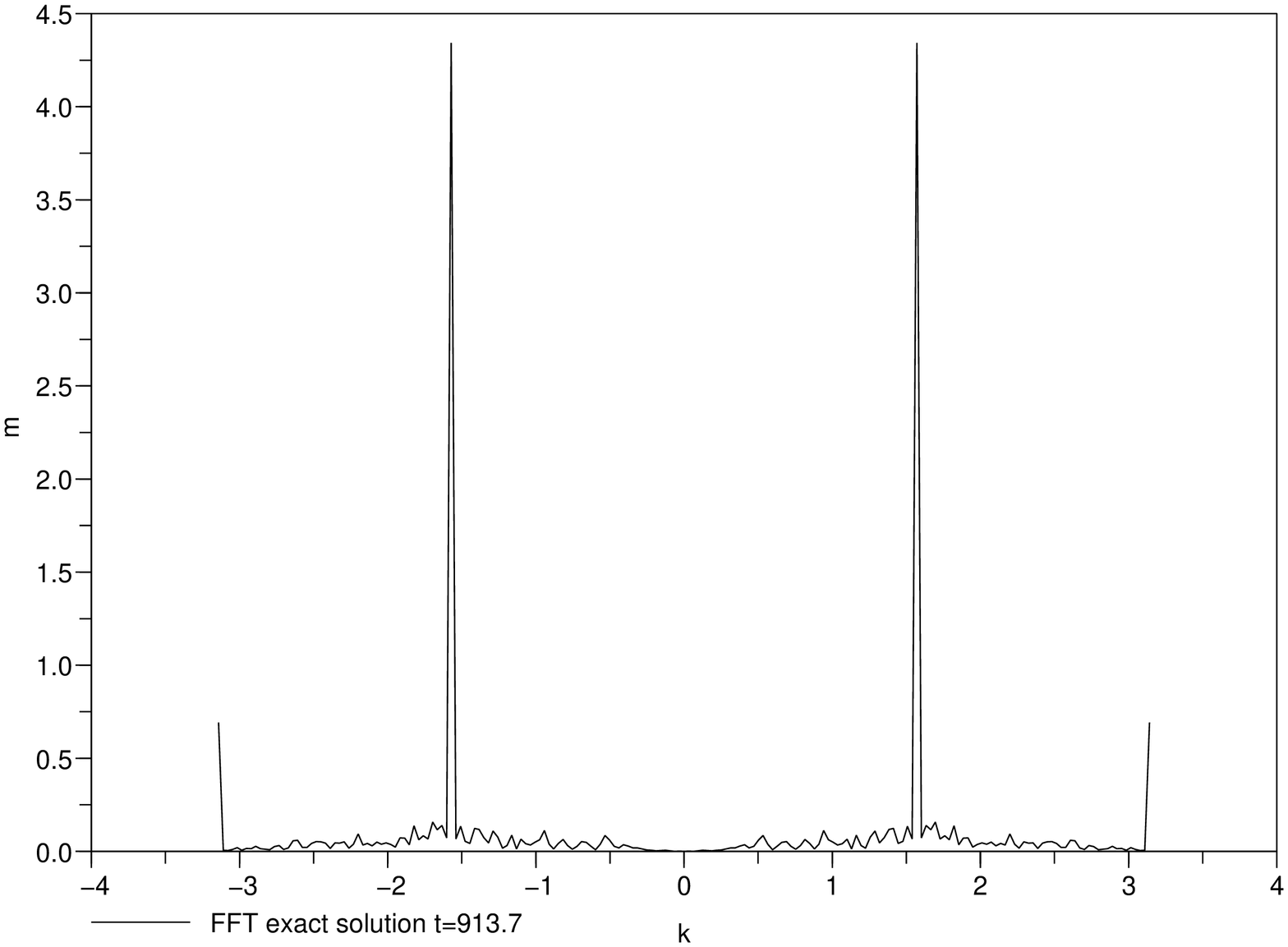}
\includegraphics[angle=-90,scale=0.22]{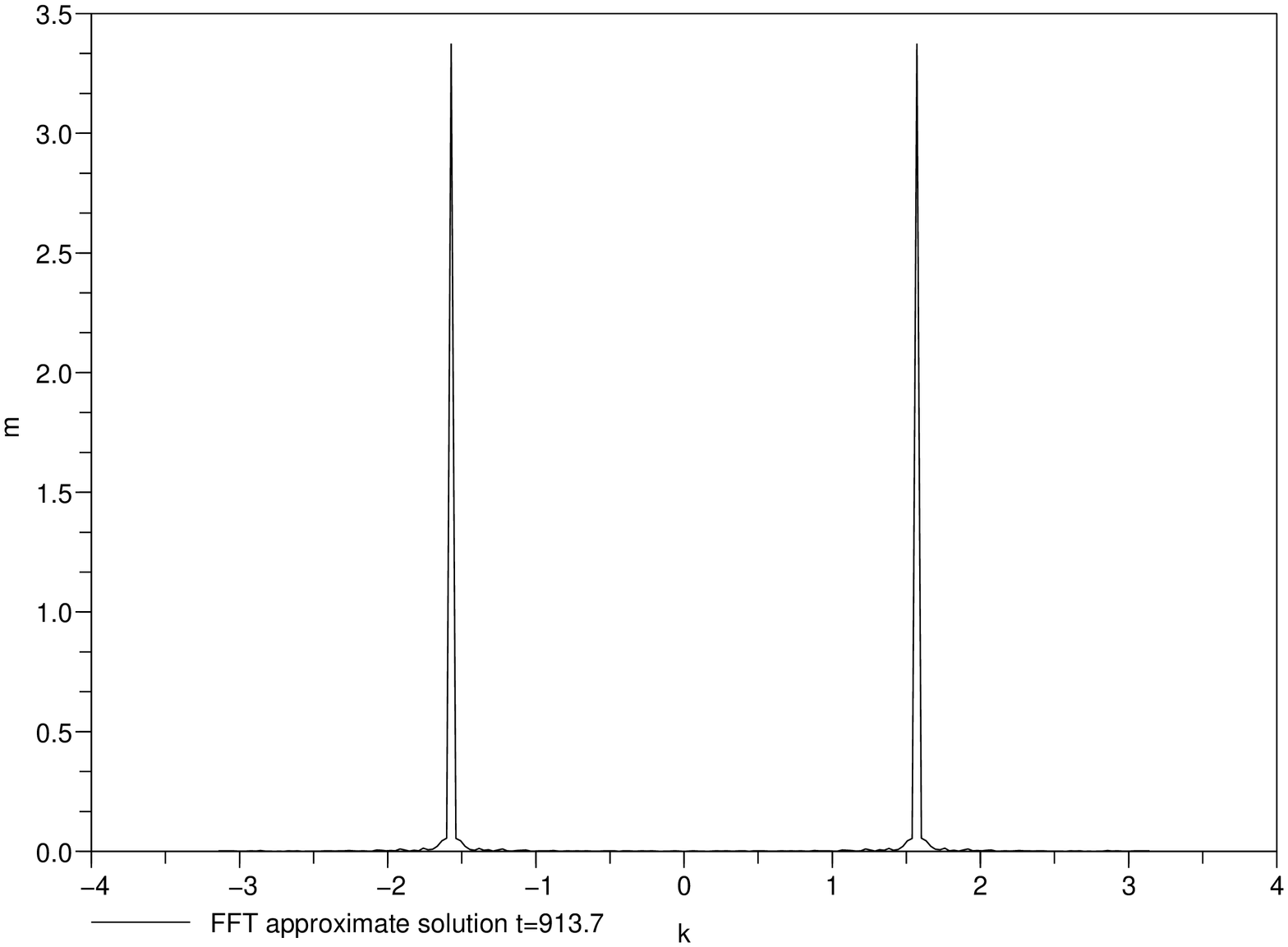}
\end{center} 
\caption{\label{fftseuil} 
Magnitude of the discrete spatial Fourier transforms of the
exact solution (left column) and approximate solution
(right column), for a random perturbation of small amplitude
travelling waves with different wavenumbers $q$ close to $\pi /2$
(top plots~: $q=52 \pi /100$, middle~: $q=51 \pi /100$, bottom $q=\pi /2$).
Plots correspond to time $t=913.7$.}
\end{figure}

\ve

As second situation in which the DpS approximation breaks down
is described in figure \ref{fft344}. For $q=57\pi /100$, the travelling wave
develops a modulational instability both in system (\ref{nc}) and the DpS
equation, but the two models display different transitory
dynamical behaviours. 
This originates from two small 
harmonics appearing in the spectrum of the exact solution
at early times of the simulation, more precisely
a second harmonic at $k=\pm 2(q-\pi )$ and a third harmonic
at $k=\pm (3q-2\pi )$. As previously
these peaks are absent in the spectrum of the approximate solution.
Here the difference with respect to figures \ref{fft106} and \ref{fftseuil}
is that these small peaks fall inside the bands of modes amplified by
modulational instability, and their amplitude grows with time
(the second harmonic dominates the third one at the beginning of the
simulation, but their amplitudes become comparable arount $t=150$).
This results in the spectra shown in figure
\ref{fft344}, where two additional large amplitude oscillating peaks are   
present in the spectrum of the exact solution at $t=344.95$ (upper left plot) and absent for the approximate solution (upper right plot). The two middle plots show
the bead displacements for the exact solution
at two different times, revealing a large scale
standing wave that strongly modulates the basic pattern.
The lower plot corresponds to the approximate 
solution given by the DpS equation, which is 
very weakly modulated and fails to reproduce the
correct transitory dynamics of (\ref{nc}) in the present situation.
However, in both models broad bands of unstable modes 
become slowly amplified at larger times by modulational instability
(this appears clearly after $t=600$ for system (\ref{nc}) and $t=700$
for the DpS equation), and around $t=900$ the
bead displacements show rather similar features in both cases. 

\begin{figure}[!h]
\psfrag{k}[0.9]{\huge $k$}
\psfrag{m}[1][Bl]{\huge $|\hat{x}_k(t) |$}
\psfrag{n}[0.9]{\huge $n$}
\psfrag{x}[1][Bl]{\huge $x_n(t)$}
\begin{center}
\includegraphics[angle=-90,scale=0.22]{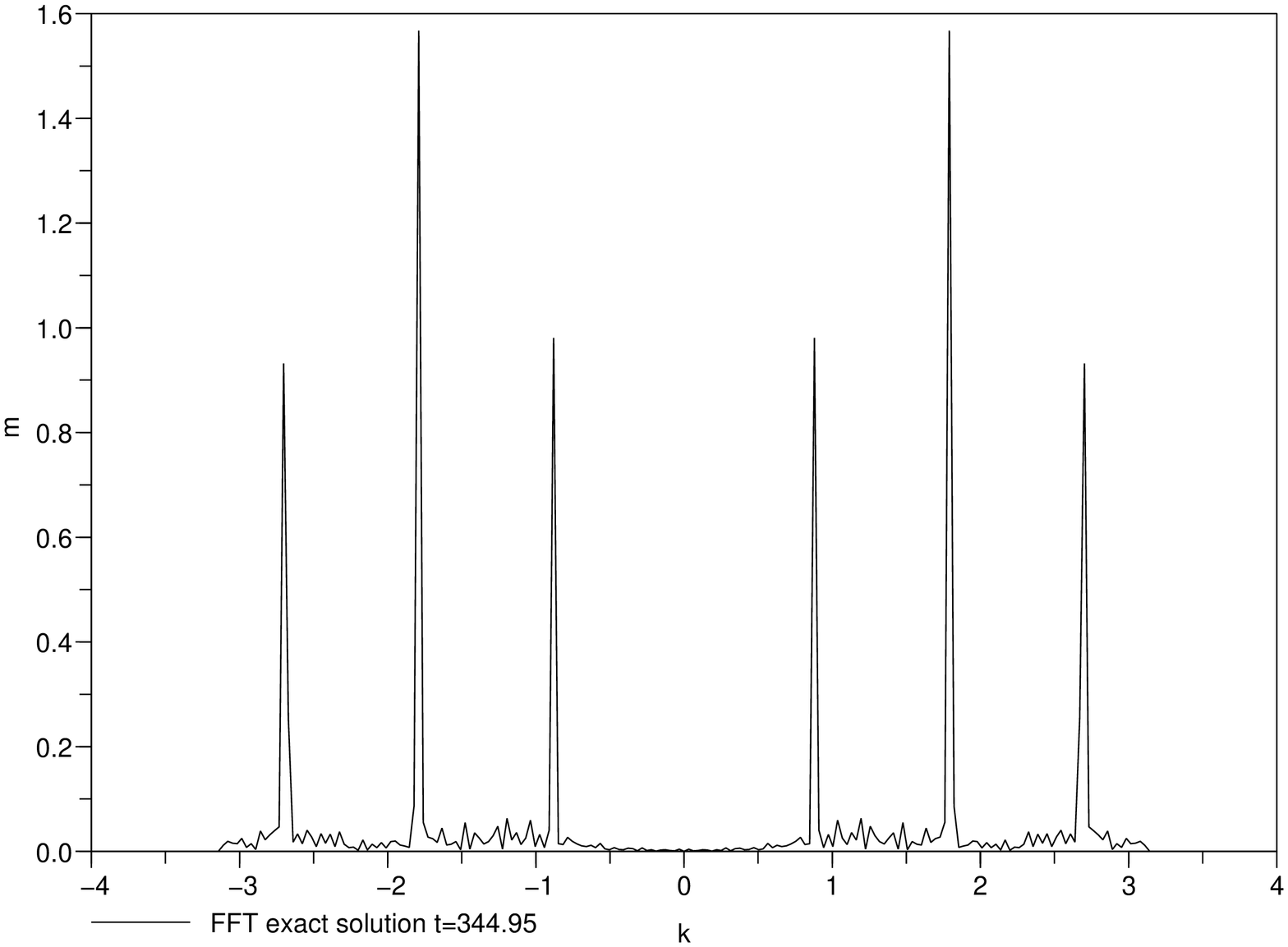}
\includegraphics[angle=-90,scale=0.22]{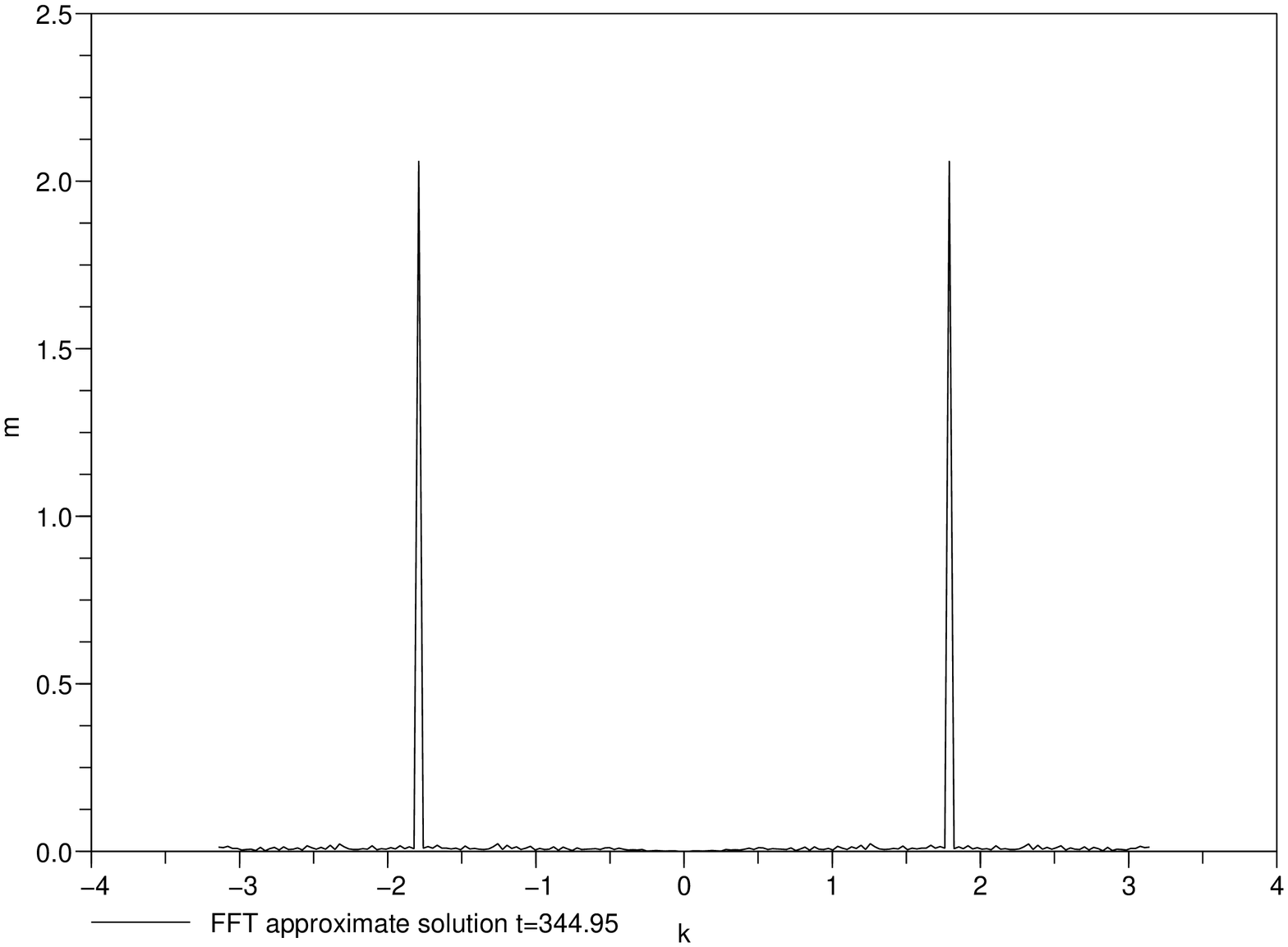}
\includegraphics[angle=-90,scale=0.22]{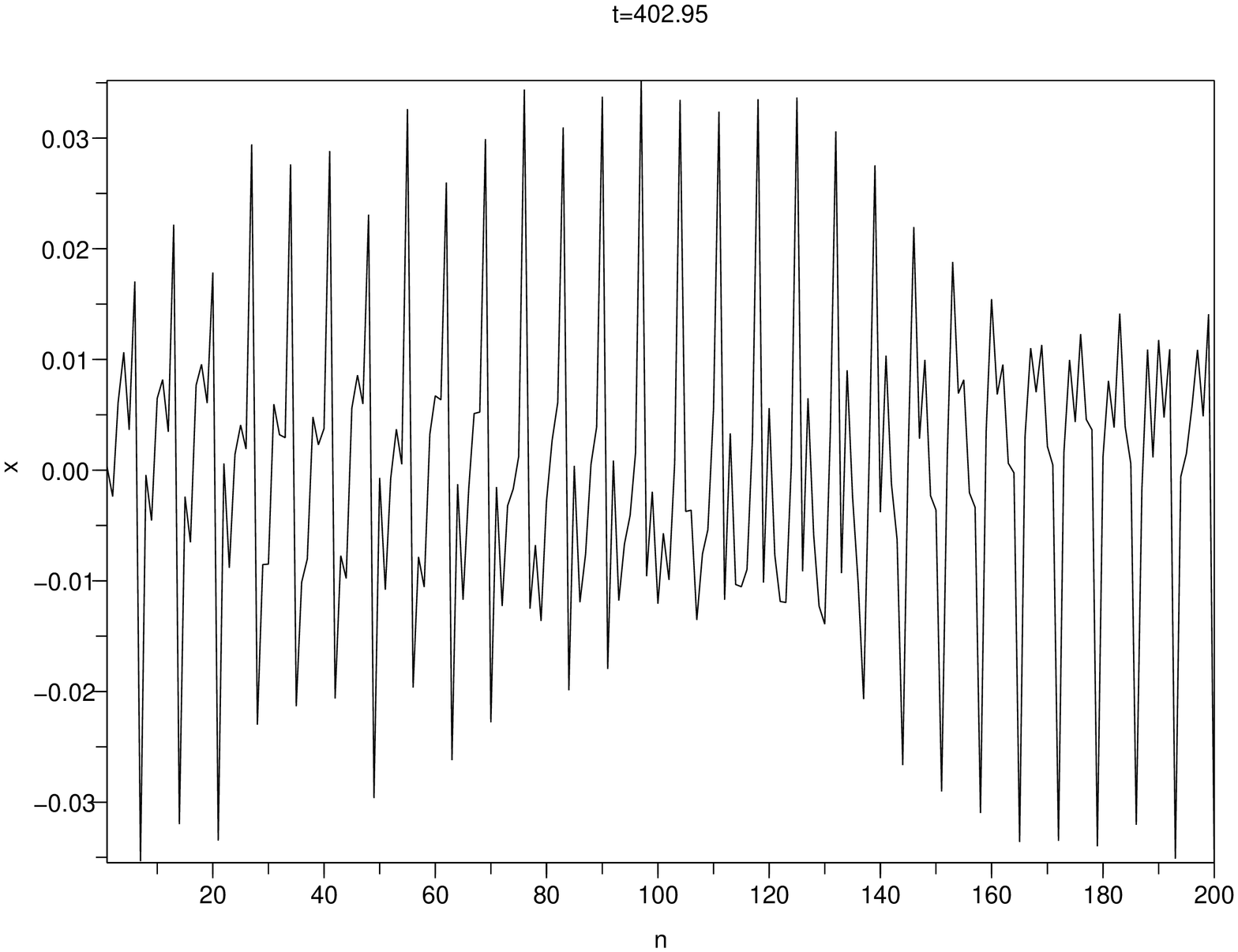}
\includegraphics[angle=-90,scale=0.22]{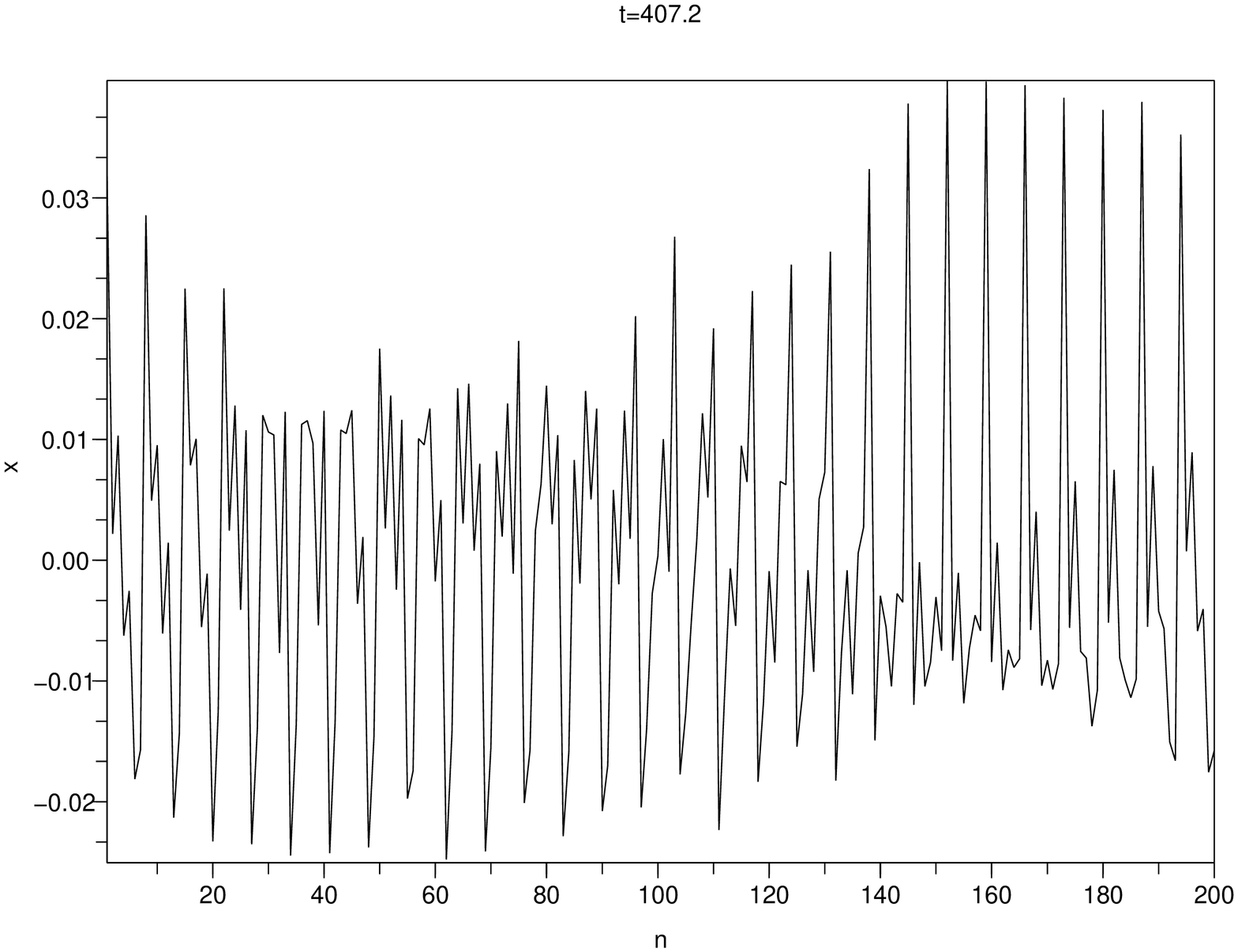}
\includegraphics[angle=-90,scale=0.22]{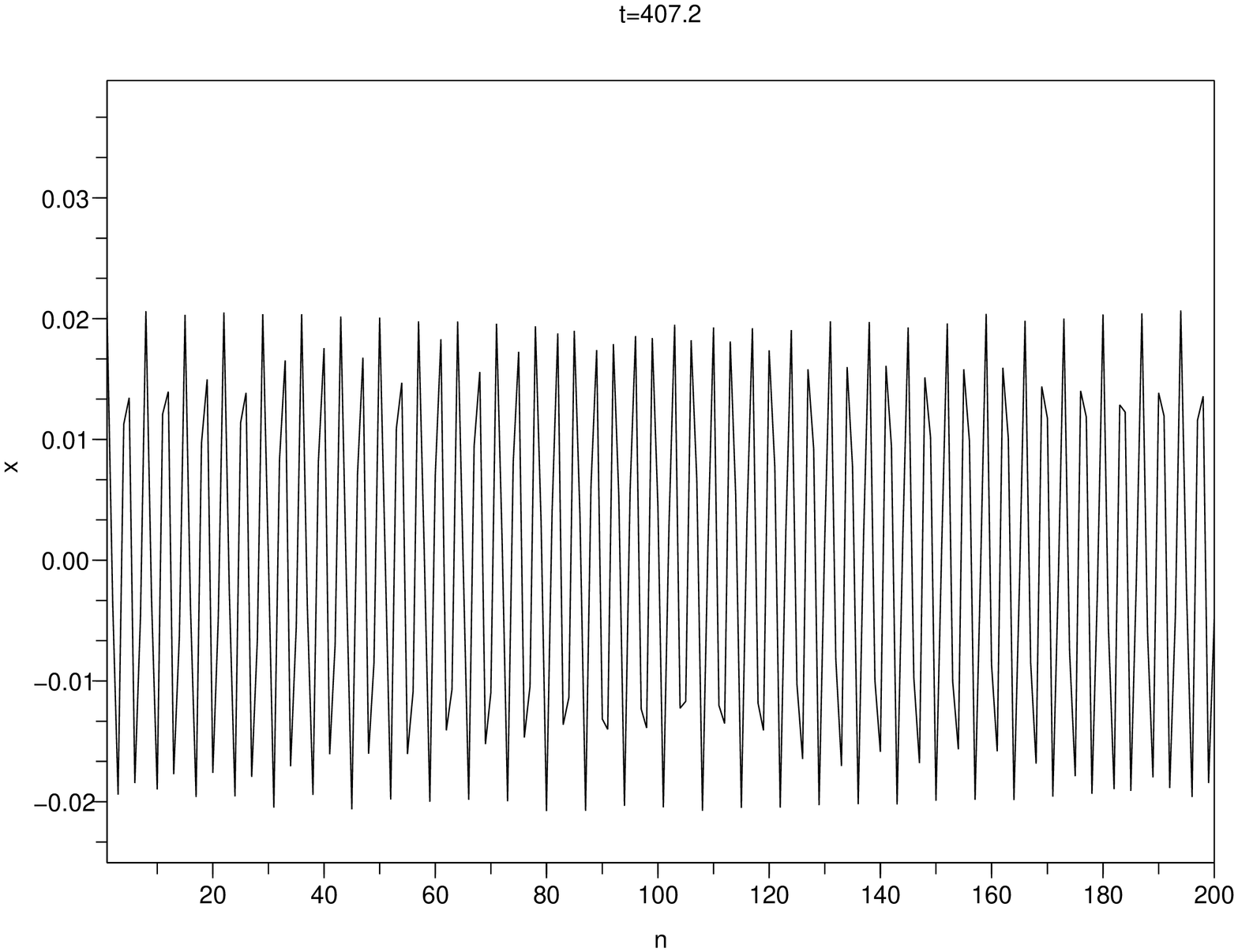}
\end{center} 
\caption{\label{fft344} 
Upper plots~: 
magnitude of the discrete spatial Fourier transforms of the
exact solution (left column) and approximate solution
(right column), for a random perturbation of a small amplitude
travelling wave with wavenumber $q=57\pi /100$,
at time $t=344.95$. Middle plots~: bead displacements for
the exact solution at $t=402.95$ and $t=407.2$.
Lower plot~: bead displacements for
the approximate solution at $t=407.2$.}
\end{figure}

\subsection{\label{loc}Localized initial conditions}

In section \ref{unstable} we have seen that the modulational instability of
certain travelling waves generates travelling or static breathers,
a phenomenon which can be captured by the DpS equation.
In what follows we illustrate with a few examples 
how the DpS equation approximates the
evolution and interaction of these localized structures.

\ve 

We start by generating
single breathers similar to the ones that result from the modulational instability
of the $q=\pi$ mode. For this purpose, we consider the spatially antisymmetric
homoclinic solution $\{ a_n \}$ of (\ref{dpsstn}) represented in figure
\ref{bcsc} (bottom left plot), the associated breather solution of (\ref{dps}) 
\begin{equation}
\label{breatherdps}
A_n (\tau ) = \frac{1}{2}\, {\mu}^{\frac{1}{\alpha -1}}\, {a}_n\, e^{i\, (\omega_{\rm{sw}} -1)\, \tau },
\ \ \
\mu = (\omega_{\rm{sw}} -1)\, 2^\alpha \tau_0 >0,
\end{equation}
and the corresponding breather ansatz (\ref{ansatzapproxst2}).
Using the initial condition determined by (\ref{ansatzapproxst2}), we integrate (\ref{nc}) numerically.
Computations are performed for a chain of $N=50$ particles with periodic boundary conditions.

The results are shown in figure \ref{plot18} for different amplitudes of the
initial breather profile obtained by varying $\mu$. 
The first case corresponds to
$|x_{N/2 +1}(0)|\approx 9,\! 7.10^{-3}$, i.e. $\omega_{\rm{sw}} \approx 1.1$
(initial condition shown in the top left plot).
By construction,
the profile of the approximate solution given by the DpS equation 
remains similar to the top left plot for all times,
just oscillating periodically in time. One observes the same qualitative behaviour
for the exact solution. As shown by the top right plot, the profiles of
the exact solution at $t=569.4$ (continuous line) and
the approximate solution at $t=565.95$ (dots) are very close
(we compare both profiles with a small time-shift, because 
the exact and approximate solutions develop a half-period
phase-shift after approximately $83$ breather periods).
These results
show the efficiency of
approximation (\ref{ansatzapproxst2}) on large time intervals of the order
of $100$ multiples of the breather period.  As shown by the bottom left plot, this localized
$3$-sites breather solution qualitatively corresponds to the central breather
forming in figure \ref{plot28}.
In the second case, the initial profile has a larger amplitude
$|x_{N/2 +1}(0)|\approx 0.34$, which corresponds to $\omega_{\rm{sw}} \approx 1.6$.
The bottom right plot (bead displacements at $t=293.2$) show that small
dispersive waves are emitted by the exact solution. 
On the contrary, the DpS equation
does not yield any dispersion in that case, only
time-periodic oscillations of the profile given by the initial condition.

\begin{remark}
We obtain qualitatively similar results with the spatially symmetric homoclinic
solution $\{ a_n \}$ of (\ref{dpsstn}) (figure \ref{bcsc}, top left plot), except the exact
solution rapidly develops a slight asymmetry which is absent in the approximate solution.
This is due to the fact that (\ref{nc}) does not possess the invariance $x_n \rightarrow x_{-n}$,
whereas (\ref{dfnls}) has the invariance $A_n \rightarrow A_{-n}$. The asymmetry
of the approximate solution could be recovered by taking into account
the higher order term $R_n$ present in (\ref{calculapp}).
\end{remark}

\begin{figure}[!h]
\psfrag{n}[0.9]{\huge $n$}
\psfrag{x}[1][Bl]{\huge $x_n(t)$}
\psfrag{m}[0.9]{\tiny $n$}
\psfrag{y}[1][Bl]{\tiny $x_n(t)$}
\begin{center}
\includegraphics[angle=-90,scale=0.22]{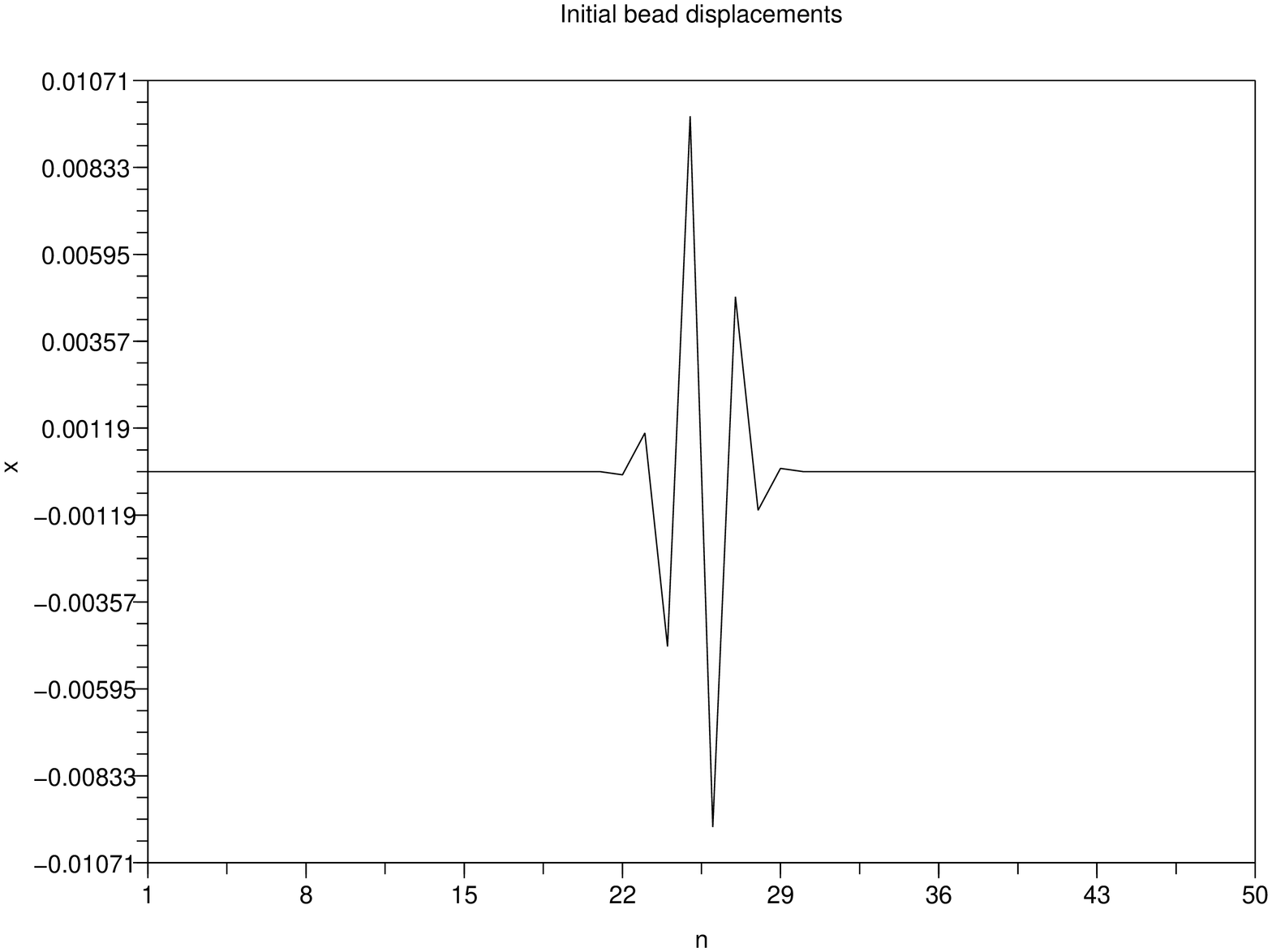}
\includegraphics[angle=-90,scale=0.22]{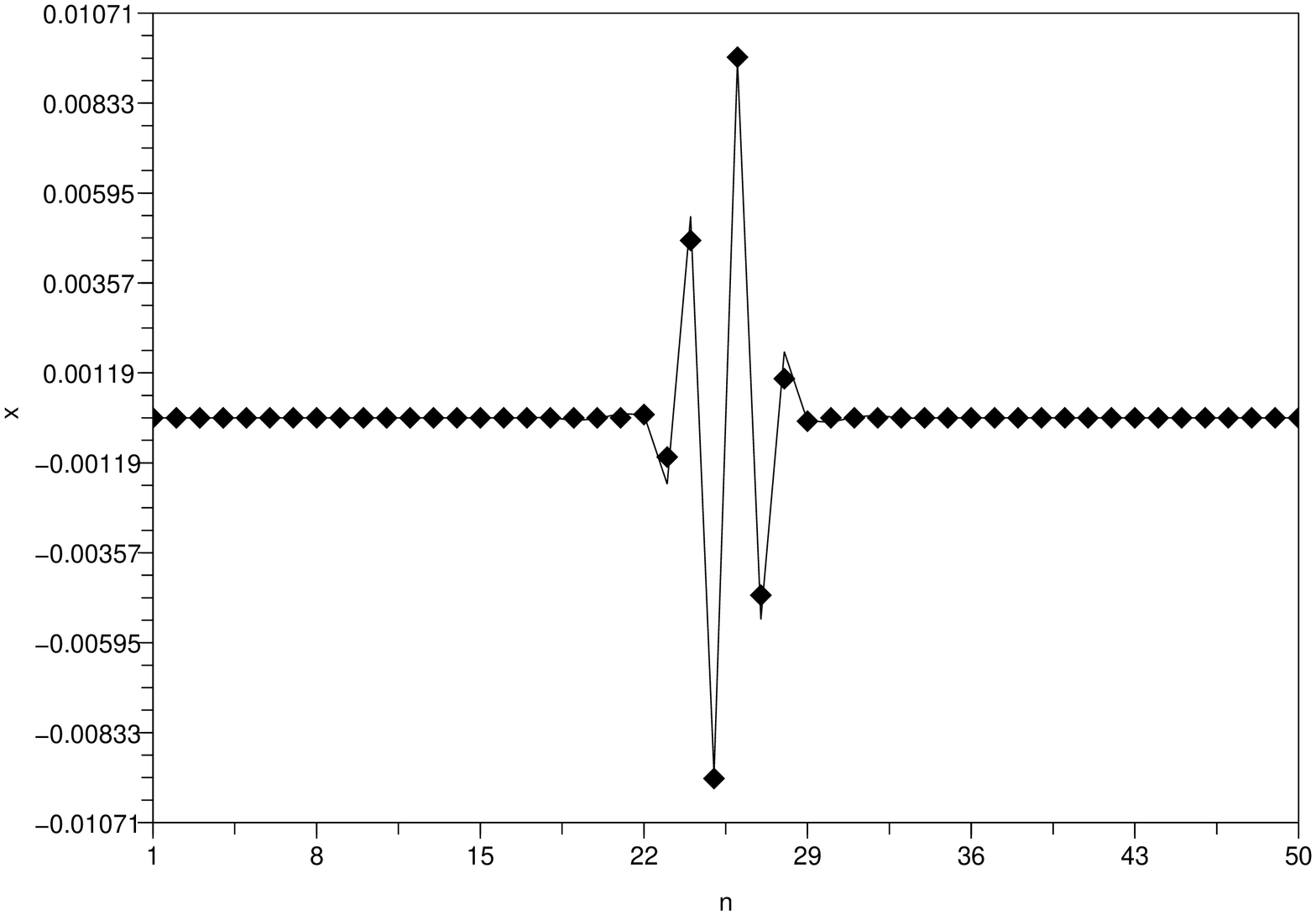}
\includegraphics[scale=0.31]{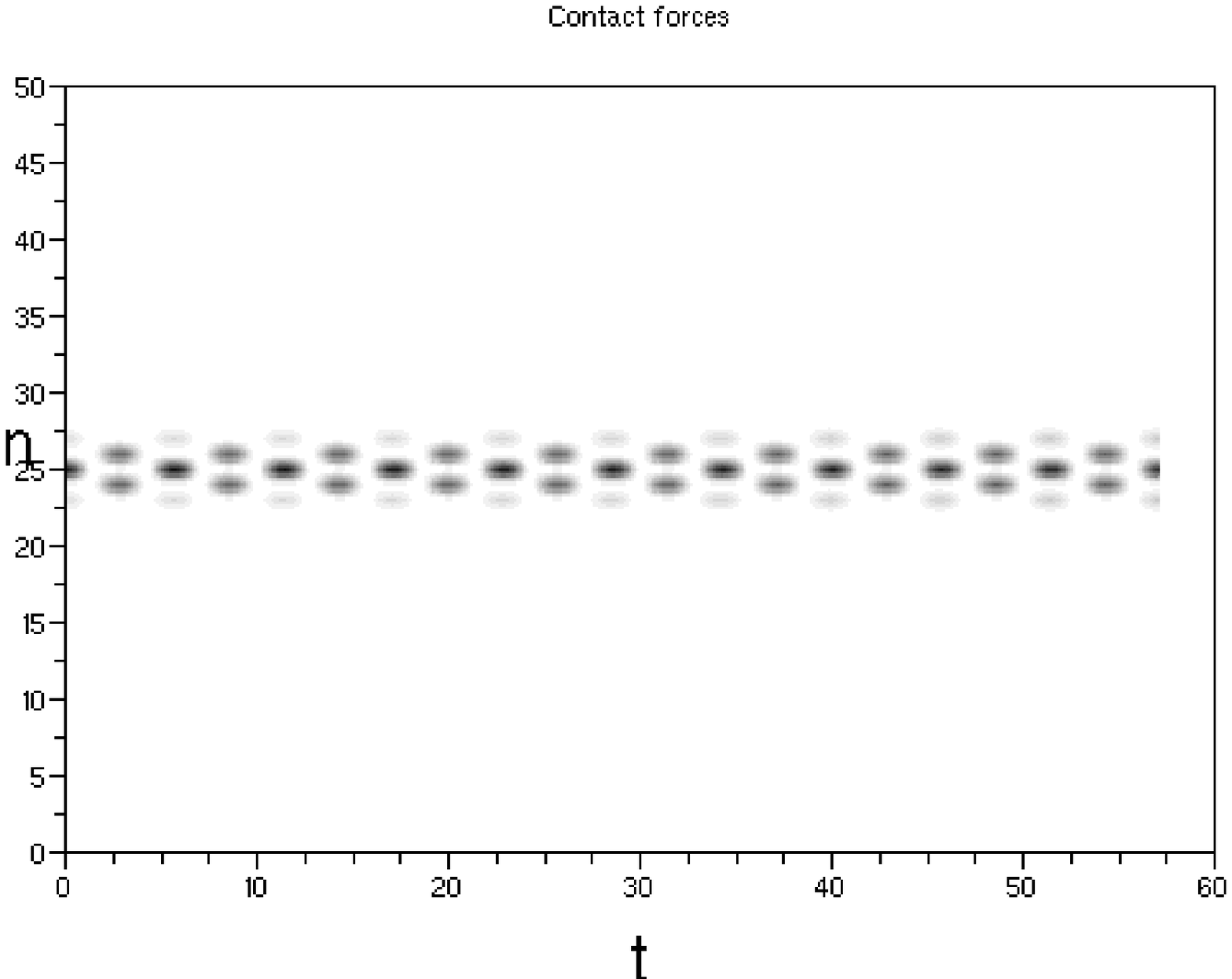}
\includegraphics[scale=0.31]{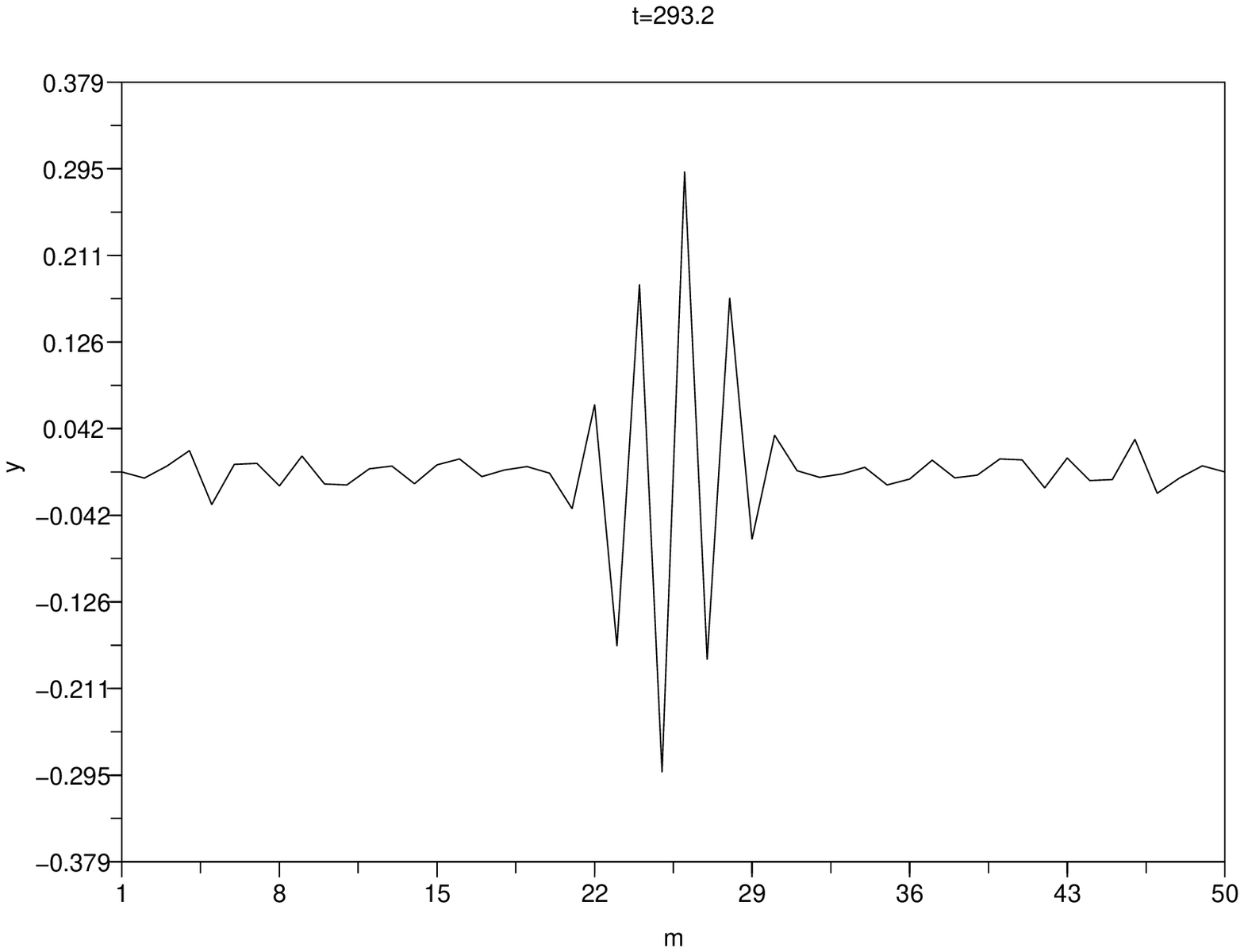}
\end{center} 
\caption{\label{plot18} 
Top left plot~: initial condition $x_n (0)$ determined by (\ref{ansatzapproxst2}),
where the amplitudes $a_n$ correspond to a spatially localized
solution of (\ref{dpsstn}) and initial velocities are set to $0$. This initial condition defines
an approximate breather solution with frequency $\omega_{\rm{sw}}\approx 1.1$. 
Top right plot~: comparison of
the exact solution at $t=569.4$ (continuous line) with
the approximate solution at $t=565.95$ (dots).
Bottom left plot~: spatiotemporal evolution of the interaction forces for the exact solution (grey levels), where
a succession of black spots reveals the breathing dynamics.
Bottom right~: exact solution at $t=293.2$, for the same type of initial condition 
with amplitude approximately $30$ times larger than in the previous case
($\omega_{\rm{sw}}\approx 1.6$).}
\end{figure}

\ve

The above computation shows that the DpS equation can approximate single breather
solutions of (\ref{nc}) provided their amplitude is sufficiently small, i.e. their
frequency sufficiently close to $1$. Now let us illustrate how the DpS equation
accounts for the interaction of small amplitude breathers, on long times scales up to
$600$ breather periods. 
Initial conditions for (\ref{nc}) are generated
using the ansatz (\ref{ansatzapproxst2}) with $\omega_{\rm{sw}} \approx 1.1$,
for initial envelopes $\{ a_n \}$ that do not correspond to exact solutions of
equation (\ref{dpsstn}). The DpS equation is initialized by setting $\tau =0$
in (\ref{breatherdps}). 

The left side plots of figure \ref{loc1} correspond to
initial displacements almost compactly supported on $12$ lattice sites, modulating
binary oscillation with amplitude close to $0.016$ (top plot). Initial velocities are set to $0$.
The spatiotemporal evolution of the interaction forces is shown for equation (\ref{nc})
(middle plot) and the DpS approximation deduced from ansatz (\ref{ansatz})  (bottom plot). 
In both cases, the initial packet splits into a static breather and
two pairs of breathers travelling slowly in opposite directions. Further crossings of the
pulses occur, with some interactions resulting in 
small phase shifts or velocity changes. The global sequence of collisions is different for
the exact and approximate solutions 
(because collisions result in differences that accumulate with time),
but these computations show nevertheless a good qualitative agreement between
the exact and approximate dynamics.

\ve

The right side plots of figure \ref{loc1} are obtained for
initial displacements corresponding to two small amplitude static breathers
initially separated by one lattice site (top plot), with amplitude close to $0.01$ 
and zero initial velocities.
Merging of the two breathers occurs after a long transcient,
both for the exact (middle plot) and approximate (bottom plot)
solutions. However this occurs 
at largely different times, around $175$ multiples of the breather period
for the exact solution, and twice more for the approximate solution.
In that case, the DpS approximation gives consequently a good
qualitative picture of the slow interaction of the two breathers,
but is quantitatively not satisfactory to predict the finite 
lifetime of the two-breather bound state.

\ve

\begin{figure}[!h]
\psfrag{n}[0.9]{\small $n$}
\psfrag{x}[1][Bl]{\small $x_n(t)$}
\begin{center}
\includegraphics[scale=0.3]{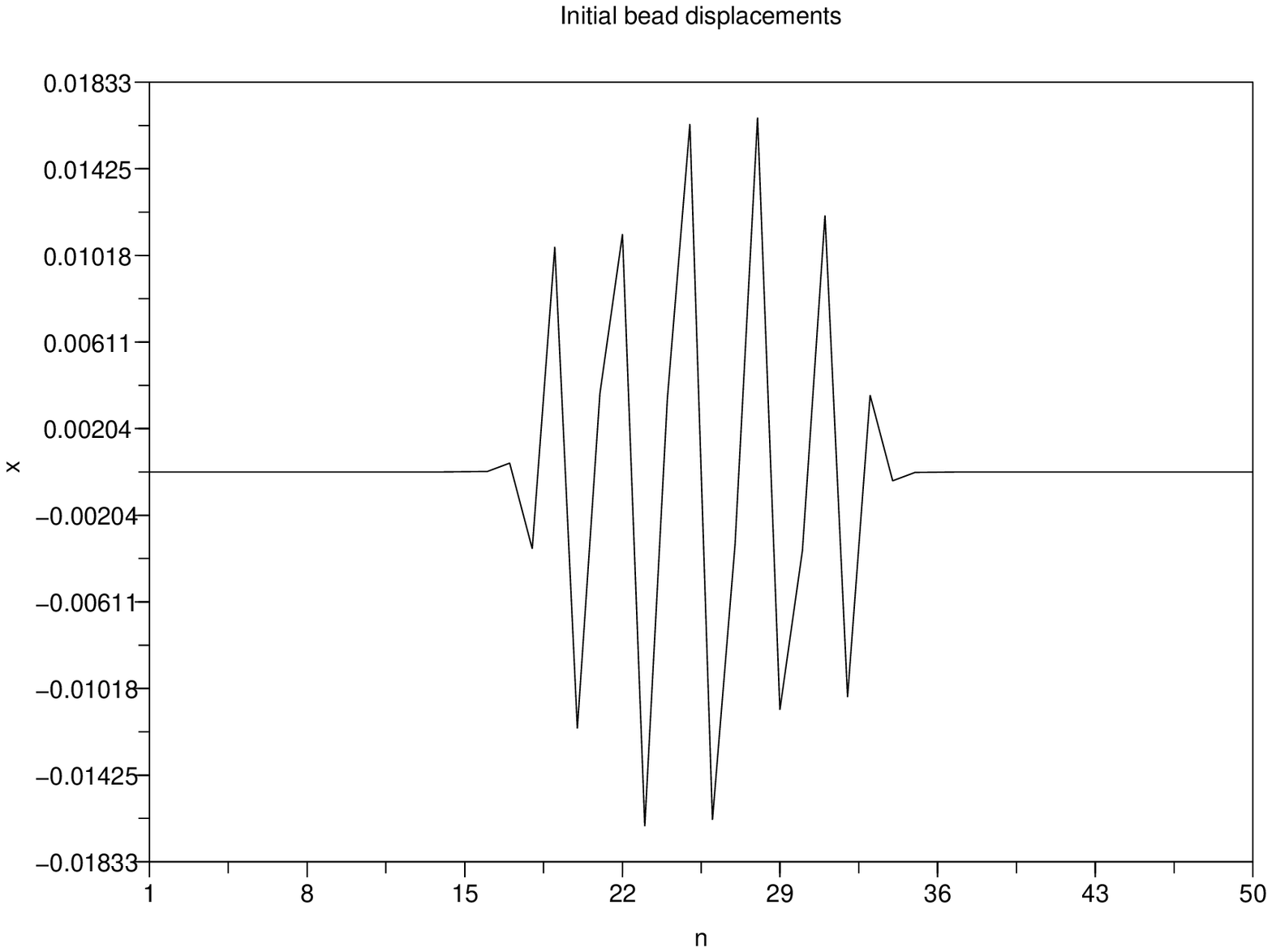}
\includegraphics[scale=0.3]{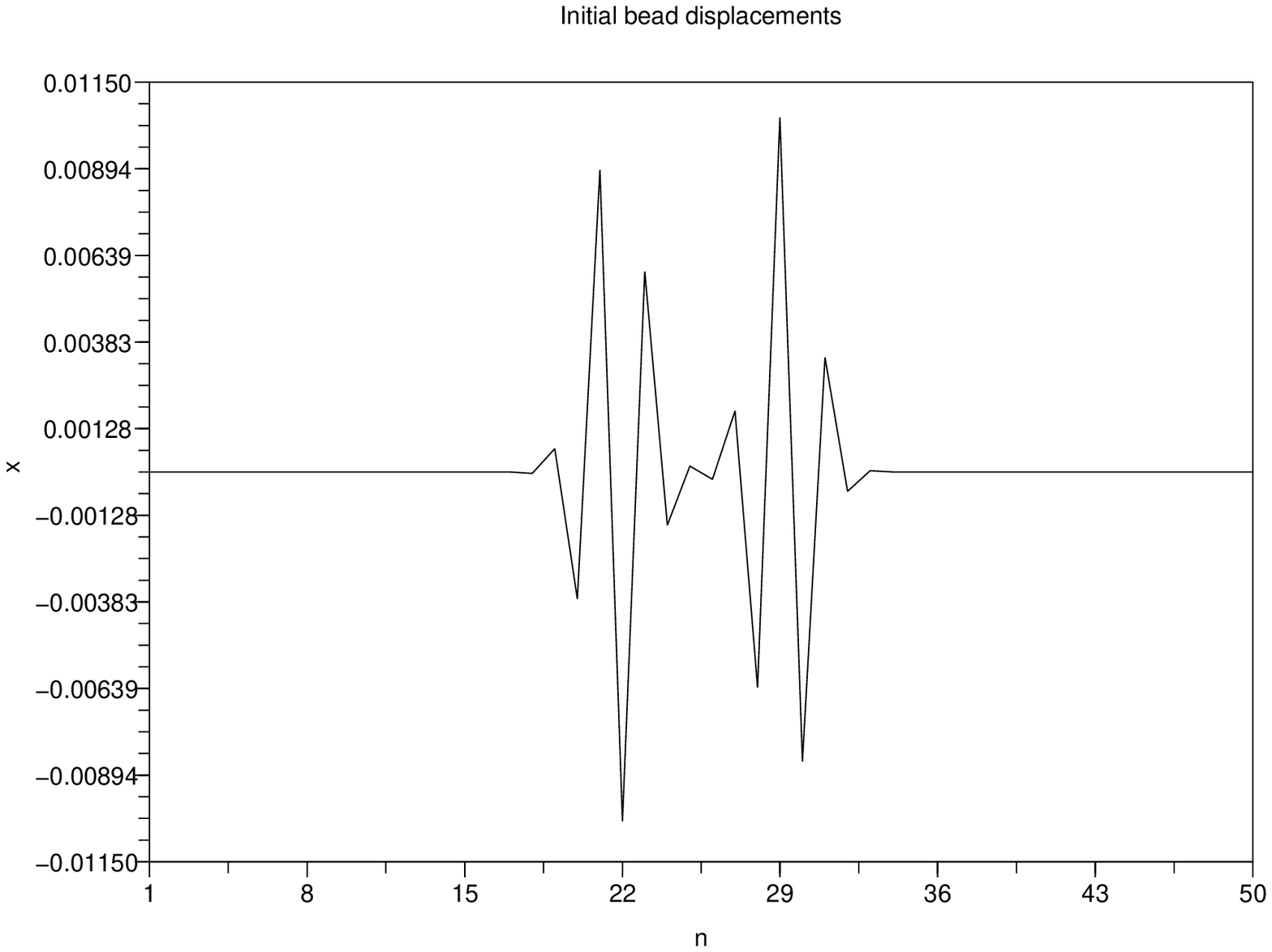}
\includegraphics[scale=0.3]{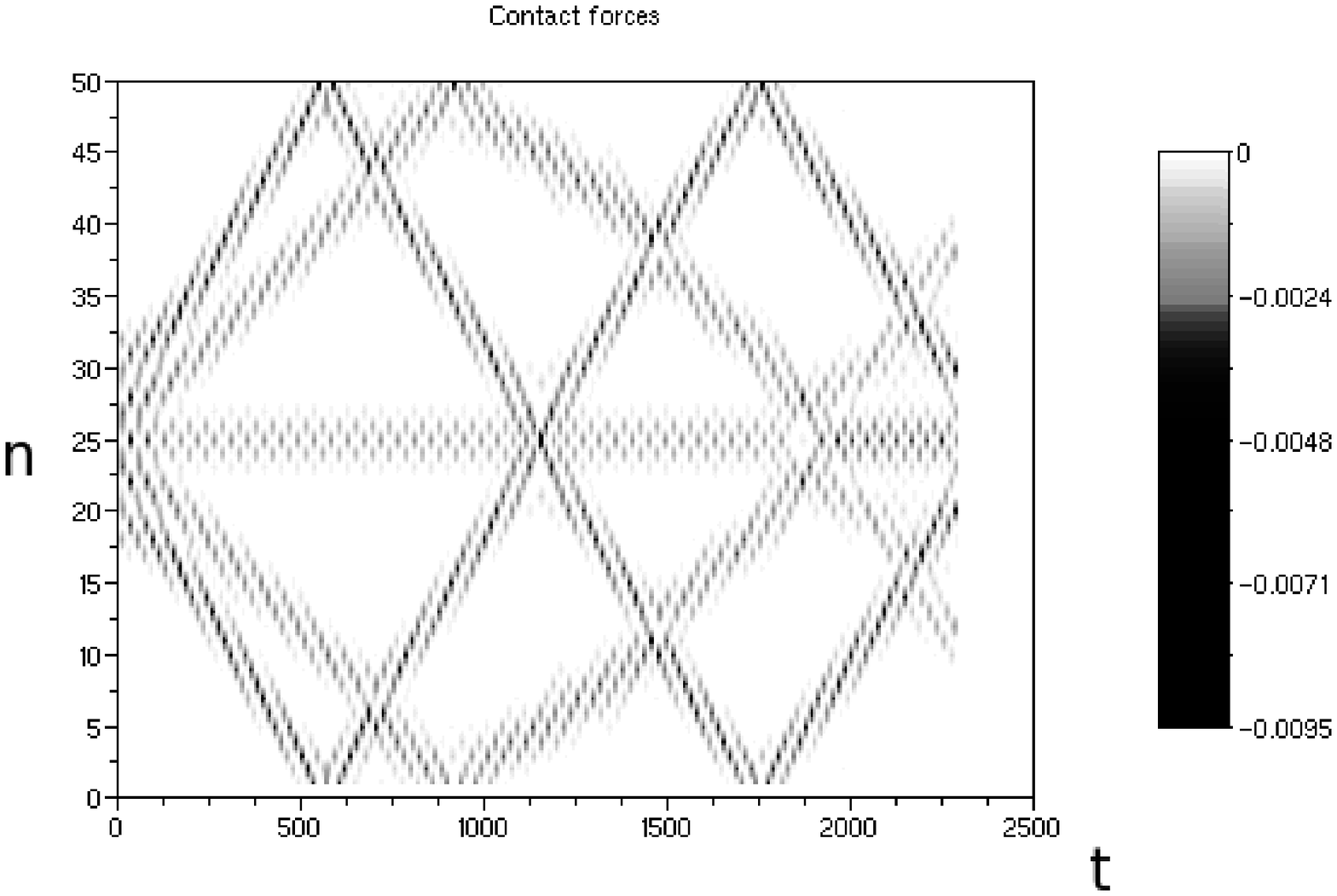}
\includegraphics[scale=0.3]{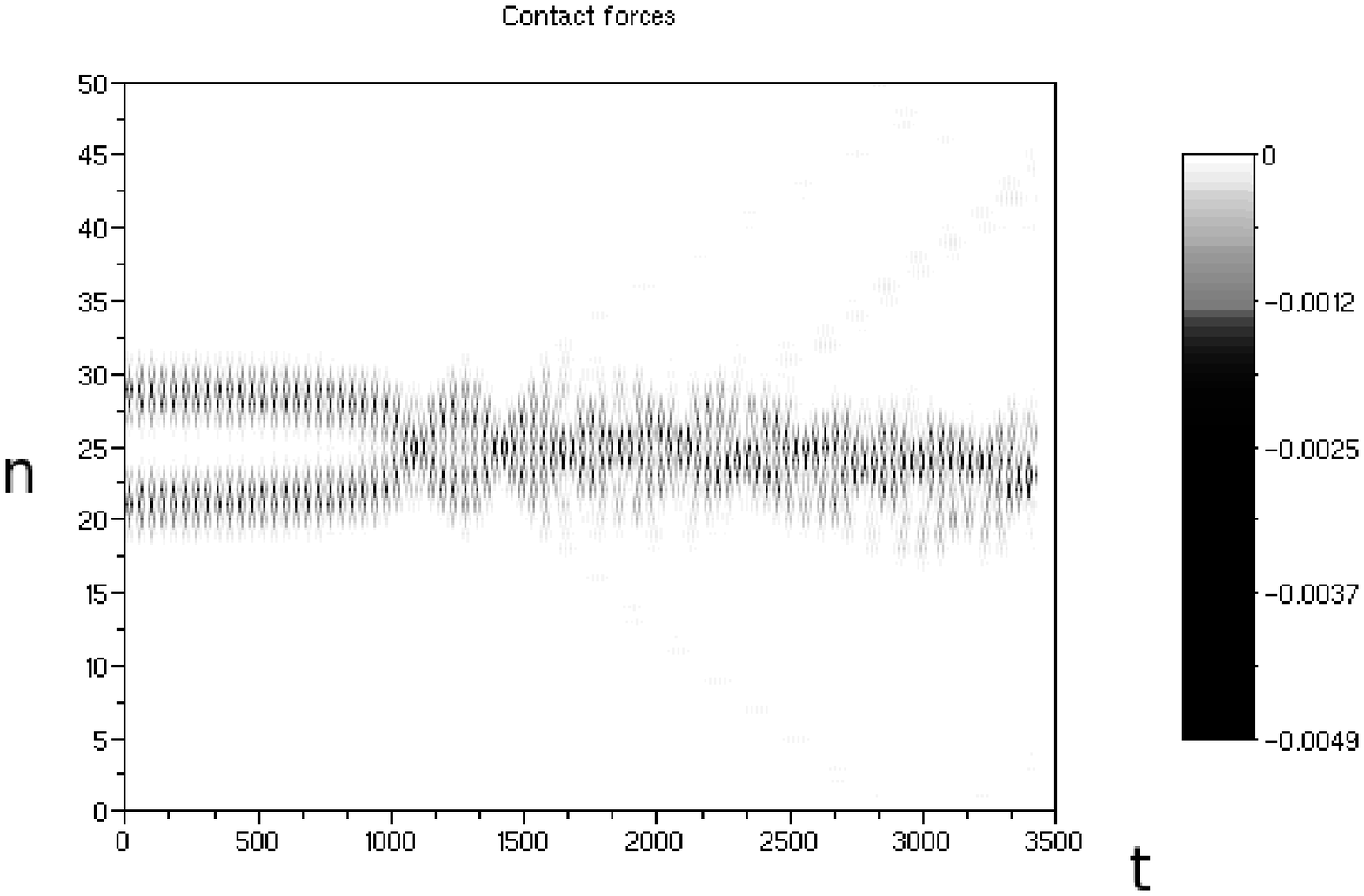}
\includegraphics[scale=0.3]{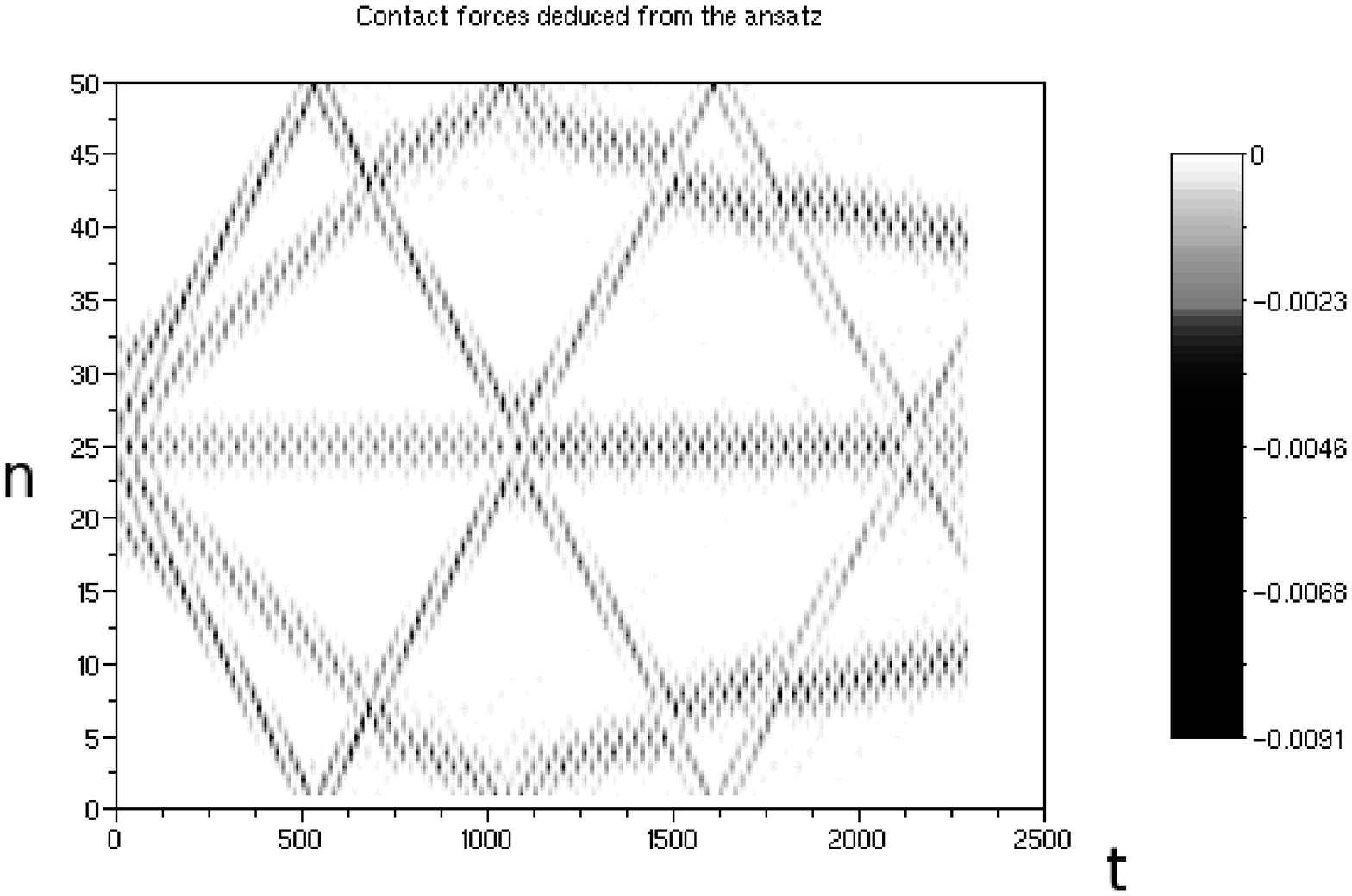}
\includegraphics[scale=0.3]{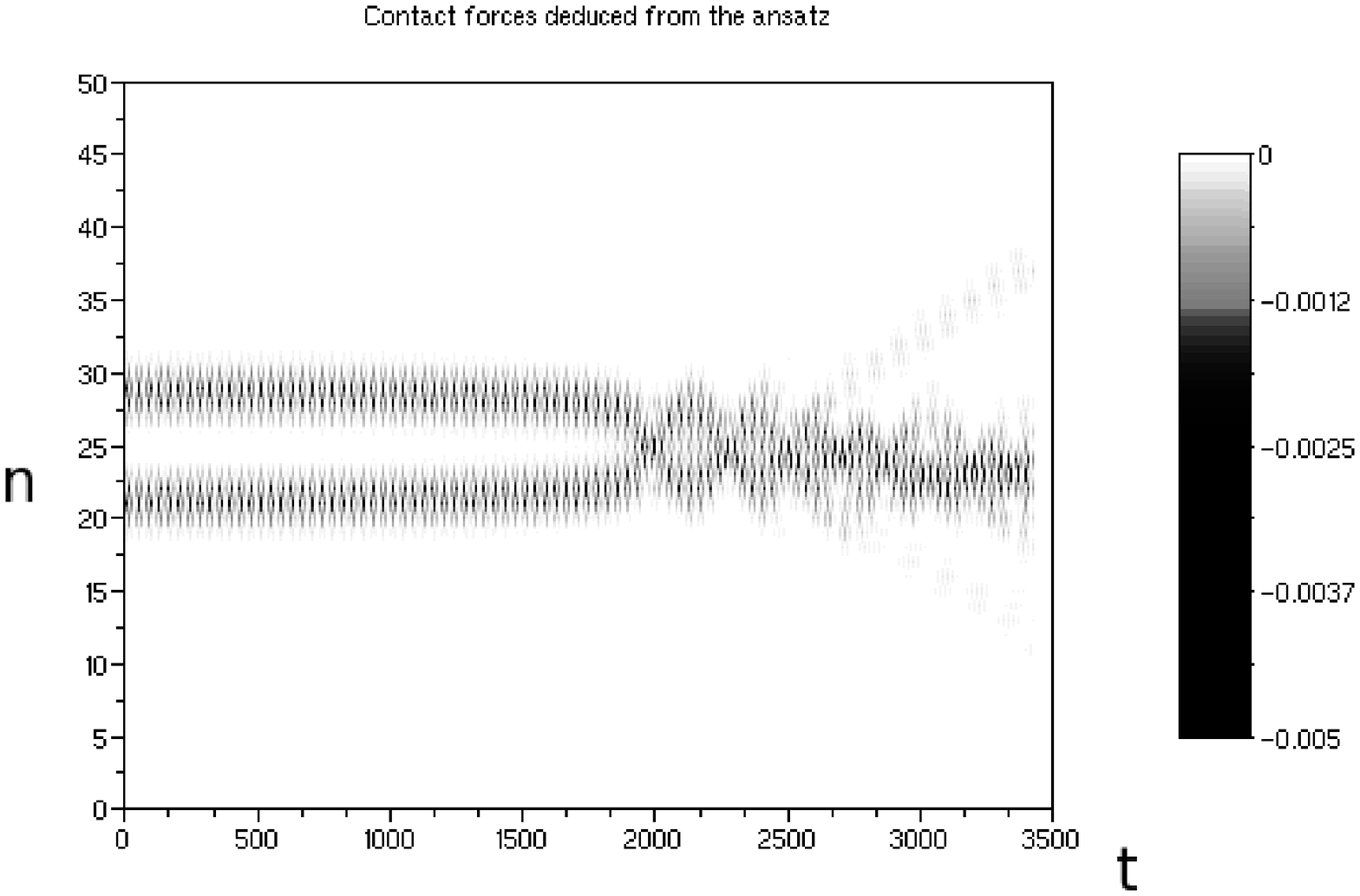}
\end{center} 
\caption{\label{loc1} 
First row~: bead displacements corresponding to two different initial conditions (initial velocities are set to $0$).
Second row~: spatiotemporal evolution of the interaction forces of system (\ref{nc}) for the above initial conditions.
Third row~: interaction forces corresponding to the approximate solutions deduced from the DpS equation
for each initial condition.}
\end{figure}

\ve

Let us end this section with some remarks concerning boundary conditions.
So far we have computed discrete breathers in chains of $N$ beads
with $N\geq 50$ and periodic boundary conditions. From an experimental point of view, system
(\ref{nc}) corresponds in that case to a closed ring of $N$ beads,
where $N$ must be large enough in order to avoid curvature effects. 
However most experiments
with Newton's cradles have been performed with a small number of beads arranged linearly,
which corresponds to system (\ref{nc}) with free end boundary conditions.  
It is important to stress that discrete breathers can be also generated in this context,
as shown by figure \ref{fini}. The left plots show displacements in a chain
of $10$ beads at different times, obtained by numerical integration of system (\ref{nc}) 
with free end boundary conditions. The results show the existence of a spatially
antisymmetric breather involving mainly $4$ beads. The initial condition is generated
using the same method and parameter values as for the small amplitude breather
of figure \ref{plot18} (top left plot), except we use free end boundary conditions
in the stationary DpS equation (\ref{dpsstn}) to compute the breather ansatz (\ref{ansatzapproxst2}).
The right plots show the approximate solution (\ref{ansatzapproxst2})
deduced from the DpS equation. The latter agrees very well with the exact solution,
except both differ by a phase-shift that is increasing with time.  

\begin{figure}[!h]
\psfrag{t}[0.9]{\small $t$}
\psfrag{s           }[1][Bl]{\tiny \raisebox{0.1cm}{\hspace*{-2ex}$\delta x=0.02$ $\updownarrow$}}
\begin{center}
\includegraphics[scale=0.30]{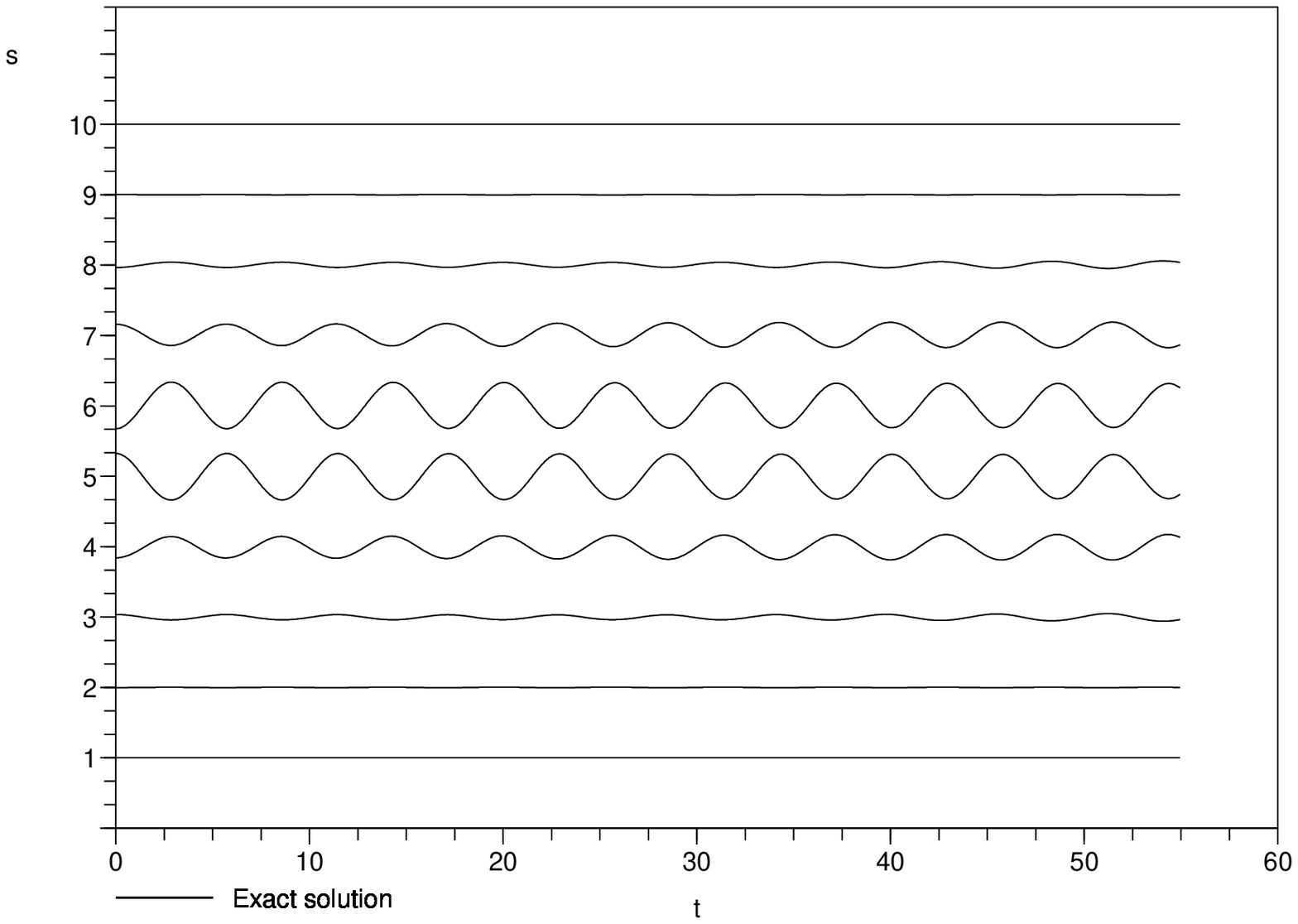}
\includegraphics[scale=0.30]{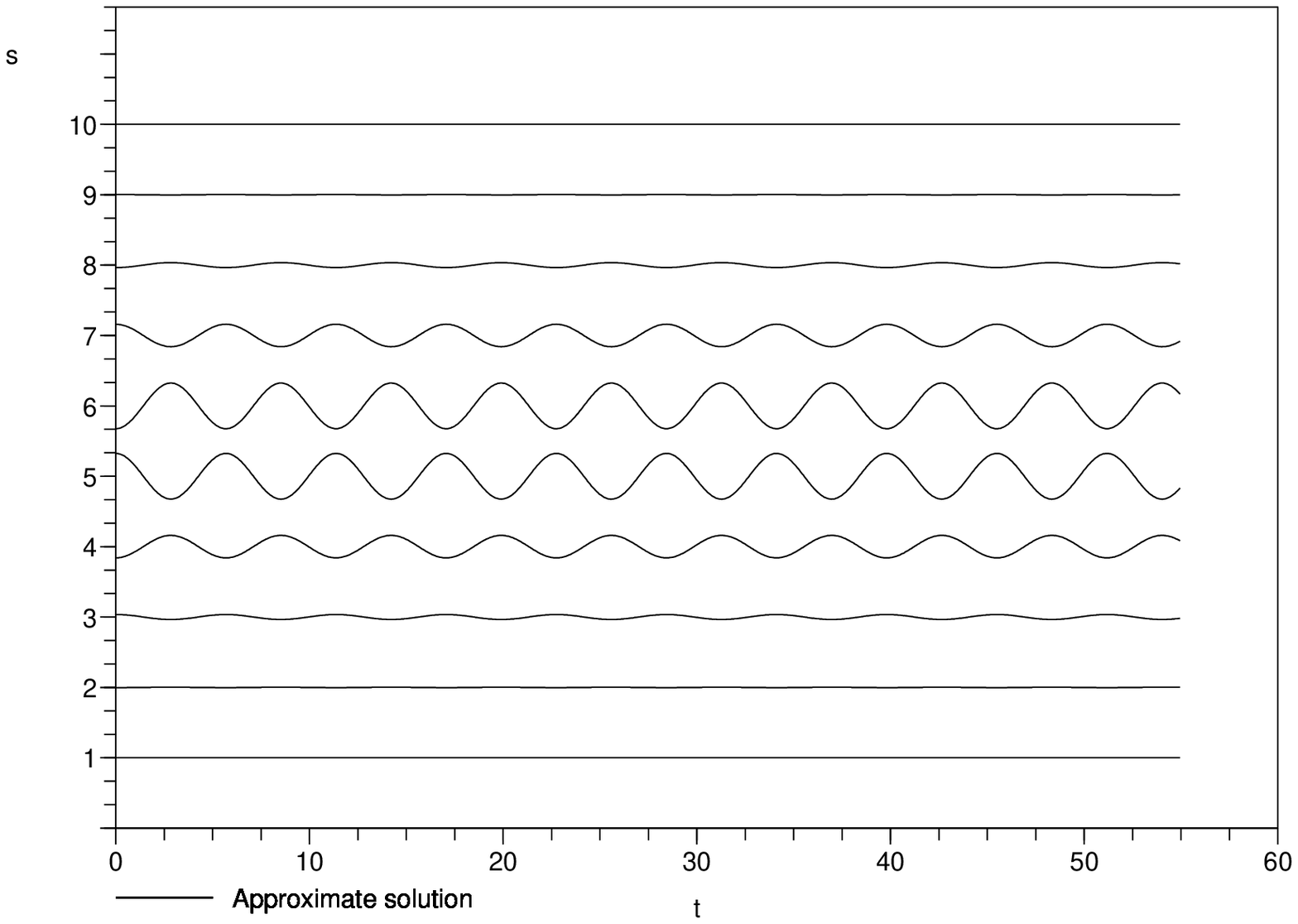}
\includegraphics[scale=0.30]{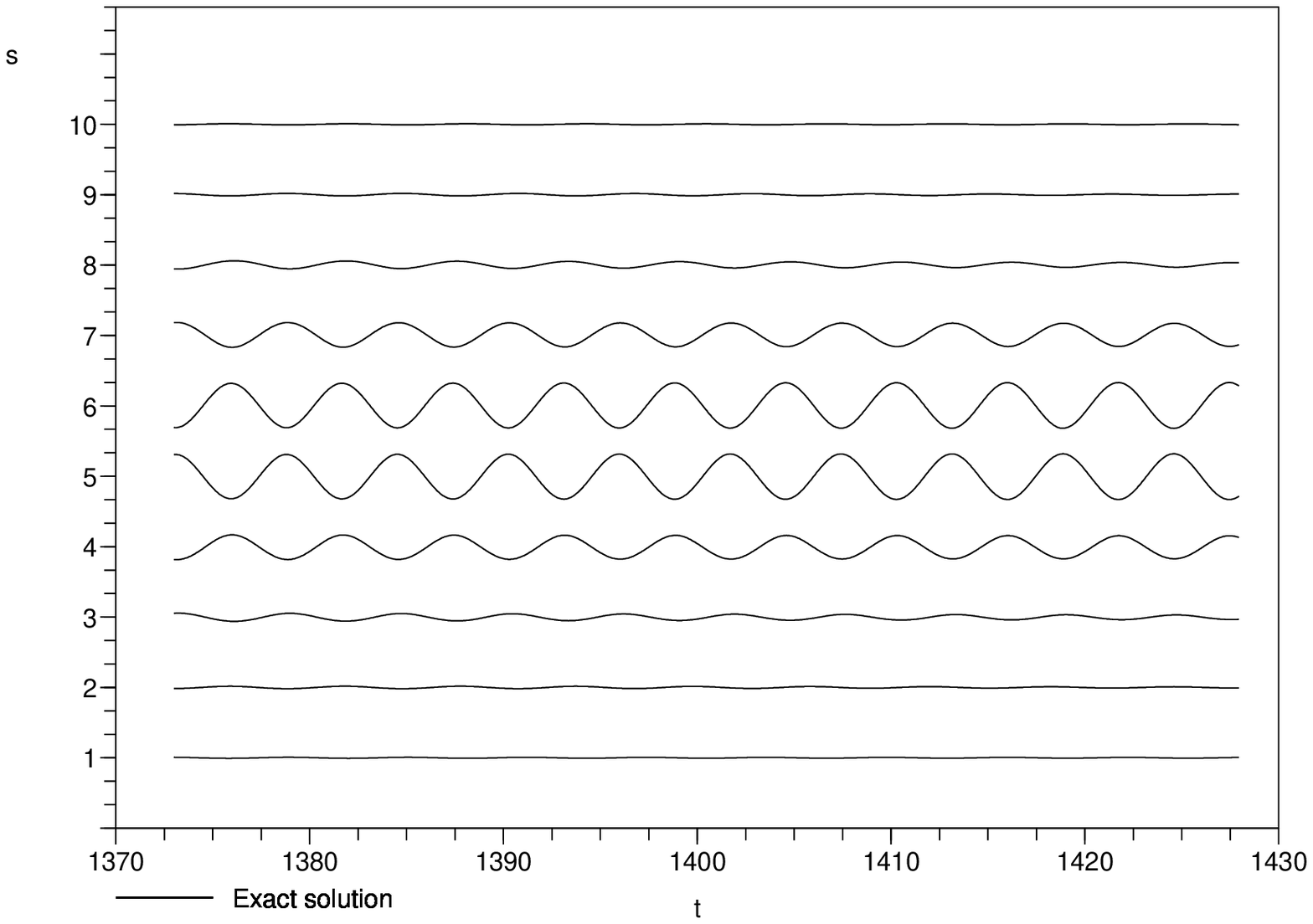}
\includegraphics[scale=0.30]{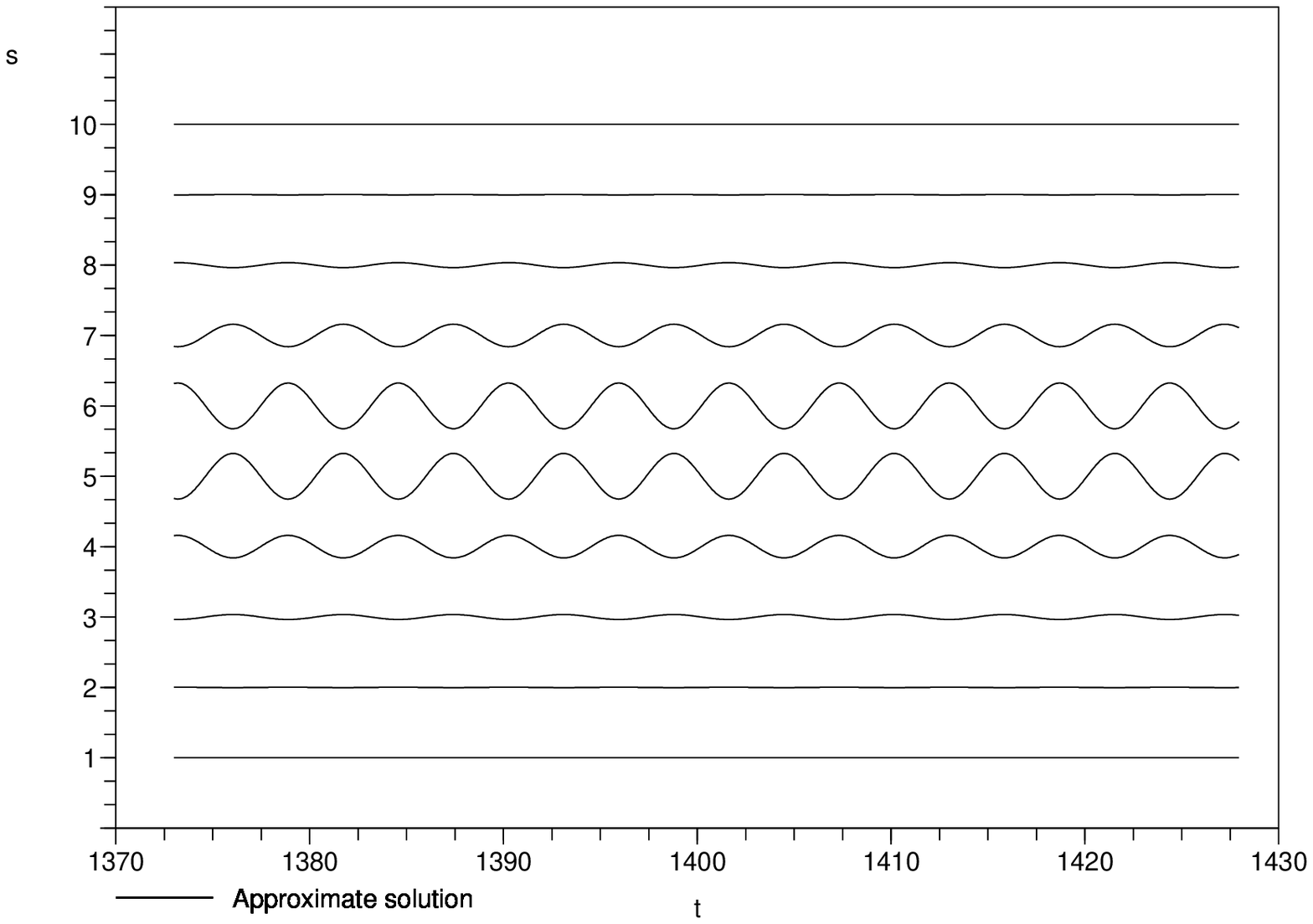}
\end{center} 
\caption{\label{fini} 
Left plots~: bead displacements for the exact solution, obtained 
by integrating (\ref{nc}) numerically with $10$ beads and free end boundary conditions.
The initial condition is determined by (\ref{ansatzapproxst2}),
where the amplitudes $a_n$ correspond to a spatially localized
solution of equation (\ref{dpsstn}) with free end boundary conditions.
Right plots~: approximate breather solution (\ref{ansatzapproxst2}) with frequency $\omega_{\rm{sw}}\approx 1.1$
deduced from the DpS equation. Beads are labelled with an index $n$ ranging from $1$ to $10$,
indicated on the vertical axis. Their
displacements are plotted in different time intervals, at early times
of the simulation (top row) and after approximately $250$ breather periods
(bottom row). The scale of bead displacements is indicated at the top of the vertical axis.
}
\end{figure}

\section{\label{conclu}Conclusion}

Although Newton's cradle has a long history which dates back to the $17$th century
\cite{hutzler}, it is still considered as a benchmark in
recent research at the cutting edge of impact mechanics \cite{liu1,liu2}.
Moreover, this system turns out to be very interesting from the modern
perspective of nonlinear science. It allows a simple experimental
realization of three nontrivial 
nonlinear effects, namely fully nonlinear dispersion \cite{ros},
a limited smoothness of nonlinear interactions and
their unilaterality. The last two aspects are still rather unexplored
in the context of nonlinear waves. 
As we have seen in this paper, these different features require to revisit several classical topics, 
i.e. local bifurcations of periodic waves,
modulation equations with fully nonlinear dispersion described by discrete $p$-Laplacians, 
modulational instability and the bifurcation and interaction of
discrete breathers. These topics are obviously very broad, and our
work is only a first step towards a better understanding of these questions
in a new context.  

From an analytical point of view, we know by theorem \ref{existthm}
the existence of exact periodic travelling waves of (\ref{nc})
close to small amplitude approximate solutions given by the DpS equation.
The above numerical simulations have revealed that the
DpS equation has in fact a much wider applicability. Indeed, it
captures other important features of the dynamics of (\ref{nc})
in the weakly nonlinear regime, namely
modulational instabilities, the existence of static and travelling breathers,
and several types of interactions of these localized structures.
An interesting open question is to justify mathematically this
close relation between the dynamics of (\ref{nc}) and
the DpS equation, for well-prepared initial conditions close to (\ref{ansatz}) at $t=0$
and large finite time intervals. Interestingly, 
such results have been proved in a different context
for the DNLS equation (\ref{dnls}), where this system is used to approximate solutions of the  
Gross-Pitaevskii equation with a periodic potential
(see \cite{peli} and references therein).

In the same spirit, it has been shown recently 
\cite{gian,gian3,bambusi2,schneider2}
that the continuum cubic nonlinear Schr\"odinger (NLS) equation
(which can be derived in continuum limits of the DNLS equation)
approximates the evolution of well-prepared
initial data in many nonlinear lattices.
It would be interesting to know if the above results extend to equation (\ref{cgnls}) 
or its generalizations to finite-wavelength continuum limits of
the DpS equation.
Our situation is more complex than in the NLS case, due to the
limited smoothness of (\ref{nc}) at the origin, and because
the Cauchy problem for (\ref{cgnls}) seems delicate to analyze. 
In addition, the analysis of continuum limits of the DpS equation
may provide interesting informations on the nonlinear waves
of (\ref{nc}), e.g. additional explicit approximate solutions. In particular, explicit compactly supported
travelling or standing waves have been found recently
in some variants of equation (\ref{cgnls}) that were formally
introduced in \cite{yan1,yan2} and in a continuum limit
of the D-QLS equation \cite{ros}
(see remark \ref{remros} p.\pageref{remros}). 
Such solutions may also
exist for continuum limits of the DpS equation, and may accurately
approximate the almost compactly supported breathers
computed in section \ref{twsol} and \ref{loc}.

Interesting open problems related to discrete breathers in Newton's cradle (\ref{nc})
and the DpS equation concern the analytical proof of their existence, their numerical 
continuation, as well as their stability, movability and interaction properties.
It will be interesting to compare their
collision properties in both models,
a problem which requires a statistical treatment due to
its sensitiveness to relative phases and breather positions \cite{bang}.
In the context of impact mechanics, discrete breathers may be used for
stringent numerical tests of multiple impact laws in model (\ref{nc}).
They may be also useful for practical purposes, since 
granular chains are interesting devices for 
the design of shock absorbers (see \cite{sen,frater} and references therein),
and the possibility might exist to trap a part of the energy of an incident wave
into static breathers in an efficient way. However, before comparing our theoretical results to
experiments it will be necessary to evaluate the impact of dissipation \cite{ricardo}
and spatial inhomogeneities on the nonlinear waves we have considered.

Another interesting question is the study of the
modulational instability of travelling waves in the
DpS equation and its more
exhaustive comparison with modulational instabilities in 
system (\ref{nc}).
One of the great interests of the DpS equation
stems from the fact that the first part of the problem
can be worked out analytically, as done in reference \cite{daumont} for the DNLS equation.
In addition a more precise numerical study of modulational
instabilities in Newton's cradle and the DpS equation 
(comparing e.g. the most unstable modes in both models)
will require statistics on power spectra of perturbations.

More generally, one can wonder if the DpS equation or higher-dimensional extensions
could capture some general features of nonlinear waves in more general granular media. 
Interesting questions related to nonlinear waves arise in this context, e.g. the possibility
of earthquake triggering by travelling or standing waves
in granular fault gouges \cite{johnson}.

\vsp{1}

\noindent
{\it Acknowledgements:} 
The author acknowledges stimulating discussions with
V. Acary, B. Brogliato, P. Kevrekidis and M. Peyrard.

%\vsp{9}

\end{document}